\documentclass[a4paper,UKenglish,cleveref, autoref, thm-restate]{lipics-v2021}

\usepackage{multirow, rotating, url, hyperref, xspace}
\usepackage{tikz,pgfmath} \usetikzlibrary{calc} 
\usetikzlibrary{decorations.pathmorphing}
\usetikzlibrary{arrows.meta}
\usetikzlibrary{hobby}

\usepackage{relsize}
\usepackage{bbm,xcolor}

\usepackage{todonotes}

\newcommand{\es}{\varnothing}
\newcommand{\NP}{\ensuremath{\mathbb{NP}}}
\newcommand{\sm}{\setminus}
\newcommand{\NC}{\operatorname{N}}

\newcommand{\pR}{pro\-per rain\-bow\xspace}

\newcommand{\pDR}{pro\-per double rain\-bow\xspace}

\newcommand{\iR}{im\pR}
\newcommand{\IR}{Im\pR}
\newcommand{\PRC}{\operatorname{PRC}}

\newcommand{\ESAT}{\textsf{3SAT}\xspace}
\newcommand{\NSAT}{\textsf{NAE3SAT}\xspace}

\newcommand{\zero}{0}
\newcommand{\eins}{1}
\newcommand{\zwei}{2}
\newcommand{\oder}{v}
\newcommand{\strn}{*}
\newcommand{\glch}{=}
\newcommand{\plus}{+}
\newcommand{\vllt}{?}
\newcommand{\bsmt}{!}

\newcommand{\typesymbol}[1]{\if#1\zero{\normalfont{\textsf{\bfseries0}}}%
  \else\if\eins#1{\normalfont{\textsf{\bfseries1}}}%
  \else\if\zwei#1{\normalfont{\textsf{\bfseries2}}}%
  \else\if\oder#1{\ensuremath{\vee}}%
  \else\if\strn#1{\ensuremath{\star}}%
  \else\if\glch#1{\normalfont{\textsf{\bfseries=}}}%
  \else\if\plus#1{\normalfont{\textsf{\bfseries+}}}%
  \else\if\vllt#1{\normalfont{\textsf{\bfseries?}}}%
  \else\if\bsmt#1{\normalfont{\textsf{\bfseries!}}}%
  \else{\textcolor{red}{\normalfont{#1}}}%
  \fi\fi\fi\fi\fi\fi\fi\fi\fi}

\newcommand{\type}[2]{\typesymbol{#1}\typesymbol{#2}}
\newcommand{\Type}[2]{#1#2}

\newsavebox{\easy}
\sbox{\easy}{\raisebox{-1.25ex}{\tikz{\node[fill=green!60]{polynomial}}}}
\newsavebox{\hard}
\sbox{\hard}{\raisebox{-1.25ex}{\tikz{\node[fill=red!50]{\NP-complete}}}}
\newsavebox{\properRainbow}
\sbox{\properRainbow}{\raisebox{-1.25ex}{\tikz{\node[fill=pink!60]{follows from \pR}}}}
\newsavebox{\properDoubleRainbow}
\sbox{\properDoubleRainbow}{\raisebox{-1.25ex}{\tikz{\node[fill=violet!60]{follows from \pDR}}}}
\newcommand{\etk}[1]{\raisebox{-0.8ex}{\tikz{\node[fill=green!40]{#1}}}}
\newcommand{\e}[1]{\raisebox{-0.8ex}{\tikz{\node[fill=green!60!black]{#1}}}}
\newcommand{\h}[1]{\raisebox{-0.8ex}{\tikz{\node[fill=red!50]{#1}}}}
\newcommand{\htk}[1]{\raisebox{-0.8ex}{\tikz{\node[fill=pink!60]{#1}}}}
\newcommand{\Kh}[3]{\h{\hyperref[#1]{\type{#2}{#3}}}}
\newcommand{\Ke}[3]{\e{\hyperref[#1]{\type{#2}{#3}}}}
\newcommand{\Ktk}[3]{\tk{\hyperref[#1]{\type{#2}{#3}}}}
\newcommand{\Ketk}[3]{\etk{\hyperref[#1]{\type{#2}{#3}}}}

{}%

\hypersetup{hidelinks=true}

\definecolor{c1}{rgb}{1.0,0.0,0.0} 
\definecolor{c2}{rgb}{0.0,0.5,0.0} 
\definecolor{c3}{rgb}{0.0,0.3,1.0} 
\definecolor{g1}{rgb}{1.0,0.6,0.6} 
\definecolor{g2}{rgb}{0.6,0.9,0.4} 
\definecolor{g3}{rgb}{0.6,0.6,1.0} 
\definecolor{g4}{rgb}{1.0,1.0,0.5} 
\definecolor{g5}{RGB}{210, 174, 235}
\definecolor{g6}{RGB}{143, 201, 191} 
\definecolor{c4}{rgb}{1.0,0.9,0.0} 
\definecolor{c5}{RGB}{26, 199, 167}
\definecolor{c6}{RGB}{177, 5, 245}
\definecolor{c7}{RGB}{245, 145, 5}
\tikzstyle{lw1}=[line width=2]
\tikzstyle{a} = [draw, very thick, circle, inner sep=0pt, minimum size=2.3mm]
\tikzstyle{aa} = [draw, very thick, inner sep=0pt, minimum size=2.3mm]
\tikzstyle{s} = [draw, thick, circle, inner sep=0pt, minimum size=1.8mm]
\tikzstyle{w} = [a, fill=white]
\tikzstyle{r} = [a, fill=c1]
\tikzstyle{g} = [a, fill=c2]
\tikzstyle{b} = [a, fill=c3]
\tikzstyle{ye} = [a, fill=c4]
\tikzstyle{tq} = [a, fill=c5]
\tikzstyle{p} = [a, fill=c6]
\tikzstyle{o} = [a, fill=c7]
\tikzstyle{n} = [draw, thick, circle, inner sep=0pt, minimum size=5mm]
\tikzstyle{t} = [draw, thick, rounded corners=2.5mm, inner sep=0pt, minimum height=5mm, minimum width=8mm]
\tikzstyle{l} = [a, color=lightgray]
\tikzstyle{c} = [a, fill=c4]
\tikzstyle{gb}= [l, fill=g3]
\tikzstyle{gr}= [l, fill=g1]
\tikzstyle{gg}= [l, fill=g2]
\tikzstyle{gtq} = [l, fill=g6]
\tikzstyle{gp} = [l, fill=g5]
\tikzstyle{gc}= [l, fill=g4]
\tikzstyle{d} = [circle, inner sep=0pt, minimum size=2.3mm] 
\tikzstyle{A} = [draw, very thick, circle, inner sep=0pt, minimum size=5mm]
\tikzstyle{M} = [very thick, circle, inner sep=0pt, minimum size=10mm]

\title{
Complexity Classification of Colouring Problems with Parity Constraints} 

\author{R\'emy Belmonte}{Universit\'e Gustave Eiffel, CNRS, LIGM, F-77454 Marne-la-Vall\'ee, France}{remy.belmonte@u-pem.fr}{0000-0001-8043-5343}{}

\author{Juan Pablo Bravo}{LIRMM, Université de Montpellier, CNRS, Montpellier, France}{juan-pablo.bravo-garrido@lirmm.fr}{https://orcid.org/0009-0008-4731-9445}{}

\author{Noleen K\"ohler}{University of Leeds, Leeds, LS2 9JT, UK}{N.Koehler@leeds.ac.uk}{https://orcid.org/0000-0002-1023-6530}{}

\author{Haiko M\"uller}{University of Leeds, Leeds, LS2 9JT, UK}{H.Muller@leeds.ac.uk}{https://orcid.org/0000-0002-1123-1774}{}

\authorrunning{R.~Belmonte, J.\,P.~Bravo, N.~K\"ohler and H.~M\"uller}

\Copyright{R\'emy Belmonte, Juan Pablo Bravo, Noleen K\"ohler and Haiko M\"uller} 

\ccsdesc[100]{
G.2.2 
F.2.2 
} 

\keywords{vertex colouring, odd colouring, computational complexity} 

\nolinenumbers 

\EventEditors{} 
\EventNoEds{2}
\EventLongTitle{} 
\EventShortTitle{} 
\EventAcronym{TR} 
\EventYear{2026}
\EventDate{} 
\EventLocation{} 
\EventLogo{}
\SeriesVolume{} 
\ArticleNo{23}

\begin{document}

\maketitle

\begin{abstract}
We study variants of graph colouring with parity constraints. 
More specifically, we consider $q$-colourings $c\colon V(G)\rightarrow \{1,\dots,q\}$ of a graph $G$ where, for every vertex $v\in V(G)$, the number of neighbours $w$ of $v$ with $c(w)=c(v)$ is restricted to be odd, even, positive, zero or a combination thereof. 
For every colour $i\neq c(v)$ the number of neighbours $w$ of $v$ with $c(w)=i$ is restricted by a constraint of similar type. 
Many known colouring problems such as proper colouring, defective colouring, exact defective colouring, odd colouring, and strong odd colouring can be described within this framework of constraining graph colourings, and therefore considering variants constitutes a natural generalisation of known colouring problems.
We provide a comprehensive study of the computational complexity of different combinations of constraints involving parity.
\end{abstract}

\section{Introduction}
\label{sec:intro}

Through the formulation of the $4$-colour conjecture for planar graphs, vertex colouring is among the oldest and most important graph-theoretic problems, see \cite{ApHa1989,McKay2012} for its proof, history and further references.
Many real world problems can be modelled by graph colouring. Consequently, the computational complexity of (variants of) graph colouring problems is central in Theoretical Computer Science. A \emph{(vertex) colouring} with $q\in \mathbb{N}$ colours of a graph $G$ is a mappings $c\colon V(G) \to \{1,\dots, q\}$  (colourings are in particular not required to be proper).
The classical version of vertex colouring, called proper colouring, seeks to decide, for a given graph $G$ and positive integer $q$, whether the vertices of $G$ can be coloured by at most $q$ colours such that adjacent vertices receive different colours, and was (for $q\geq 3$) among the $21$ problems shown to be \NP-complete in Karp's seminal work \cite{Karp2010}. 

\noindent \textbf{Related vertex colourings}
Naturally, many variants of proper colouring obtained by \textsl{e.g.} restricting or relaxing proper colouring have been studied. \emph{Defective colouring} constitutes a very natural example of such a relaxation. Essentially, in a defective colouring $c\colon V(G)\rightarrow [q]$ one allows every vertex to have at most $d$ (or precisely $d$ in the case of exact defective colouring) neighbours of the same colour for a small constant $d$, see \cite{frick1993survey,woodall1990improper} for references to original work.
This relaxation remains \NP-complete even on $q$-regular graphs \cite{cowen1997defective}. The exact version for defect $d=1$ is also called \emph{$q$-colourable perfect matching} and was shown to be \NP-complete for $2$-colours \cite{Sch78} and for $q\geq 3$ colours \cite{DemaineKP25}. $2$-colourable perfect matching remains \NP-complete even on $2$-connected $3$-regular planar
graphs.
A variant akin to defective colouring, called \emph{even/odd colouring}, is obtained by allowing every vertex to have an even/odd number of neighbours of the same colour.  
Gallai showed that every graph admits an even colouring with $2$ colours \cite{Lovasz93}, implying that the problem is polynomial time solvable. On the other hand,  there is no bound on the number of colours needed to odd colour a graph \cite{Scott01} (consider \textsl{e.g.} a $K_n$ with each edge subdivided once, in which no two non-subdivision vertices can receive the same colour in an odd colouring). While every graph without isolated vertices is odd colourable \cite{Scott01}, the associated minimization problem is \NP-complete for $3$ colours \cite{BelmonteS21} and polynomial time solvable for $2$ colours \cite{BelmonteS21}.

While the above variants are relaxations of proper colouring (only the appearance of neighbours of the same colour are restricted) another recent direction of research studies colourings that are restrictions of proper colourings in the sense that the appearance of colours  not equal to a vertex's own colour in the neighbourhood is restricted in some way. Also introduced under the name odd colouring, but for clarity referred to as \emph{proper odd colouring} here, Petrusevski and Skrekovsk introduced proper colourings for which additionally for each vertex $v$ some colour different from the colour of $v$ appears an odd number of times in the neighbourhood of $v$ \cite{PetrS21}. This notion is inspired by a notion of odd colouring for hypergraphs introduced by Cheilaris \textsl{et al.} \cite{cheilaris2013unique} and is \NP-complete to compute \cite{caro2022remarks} even on bipartite graphs \cite{ahn2025proper}. Very recently, a stronger version of proper odd colouring, called \emph{strong odd colouring}, was introduced \cite{kwon2024strong}. Strong odd colourings are proper colouring where every colour appearing in a neighbourhood appears there on an odd number of vertices, see also \cite{CaroPST25,Pilipczuk25} for further references. To the best of our knowledge the complexity of this problem is still unknown. Finally, a similar version, called \emph{conflict-free colouring}, in which each non-isolated vertex $v$ is required to have a neighbour whose colour is unique within the open neighbourhood of $v$  was introduced in \cite{fabrici2023proper} and shown to be \NP-complete even on bipartite graphs (for $4$ colours even on planar graphs) \cite{ahn2025proper}.

\noindent \textbf{Our results} 
Inspired by the colouring variant of Petrusevski and Skrekovsk we study colourings where we place some restriction on the appearance of neighbours of a vertex $v$ with the same colour as $v$ and we place some restriction of the appearance of all other colours in the neighbourhood of $v$. More specifically, we consider colouring $c\colon V(G)\rightarrow \{1,\dots, q\}$ for which we impose, for every vertex $v \in V(G)$, a parity condition on some colour classes restricted to the open neighbourhood $N(v)$ of $v$.
That is, for some colours $i\in \{1,\dots,q\}$, the set $N(v) \cap c^{-1}(i)$ should be of even or odd size. Additionally, we may require that the set is empty or not empty. The parity condition might depend on the colour of $v$, that is, whether $i=c(v)$ or $i \ne c(v)$ holds. All combinations of considered constraints can be seen \textsl{e.g.} in \cref{tab:casesTwoColours}.

\begin{table}[ht]
\begin{center}
  \renewcommand{\arraystretch}{1.7}
  \renewcommand{\l}[1]{\multirow{2}{*}{#1}}
  \newcommand{\m}[1]{\multirow{2}{*}{$#1$0}}
  \begin{tabular}{ll|*{9}{c|}}
    \multicolumn{2}{c|}{\multirow{3}{*}{$|N(v) \cap c^{-1}(i)|$}}
    & \multicolumn{9}{c|}{$i \ne c(v)$} \\[+1.7ex]
    \multicolumn{2}{c|}{}
    &\l{even} &even    &\l{odd} & odd    & arbi-  & \m{=}  & \m{>}  &\l{$<$2} &\l{$=1$}
    \\[-1.7ex] \multicolumn{2}{c|}{}                                 
    &       &and \!$>$0&        & or =0  & trary  &        &        & &
    \\ \hline \multirow{9}{*}{\hspace*{-3mm}\rotatebox{90}{$i=c(v)$}\hspace*{2mm}} & even
    & \Ketk{link1234:00}{0}{0} & \Kh{link2:02}{0}{2} & \Ke{link2:01}{0}{1} & \Ke{link2:0v}{0}{v} &\etk{\type0*} \cite{Lovasz93} & \Ketk{link1234:0=}{0}{=} & \Kh{link2:0+}{0}{+} & \Kh{link2:0?}{0}{?} & \Kh{link234:0!}{0}{!}
    \\ \cline{2-11} & even and $> 0$
    & \Ketk{link1234:20}{2}{0} & \Kh{link2:22}{2}{2} & \Kh{link2:21}{2}{1} & \Kh{link2:2v}{2}{v} & \Kh{link2:2*}{2}{*} & \Ketk{link1234:2=}{2}{=} & \Kh{link2:2+}{2}{+} & \Kh{link2:2?}{2}{?} & \Kh{link2:2!}{2}{!} 
    \\ \cline{2-11} & odd
    & \Ketk{link1234:10}{1}{0} & \Kh{link2:12}{1}{2} & \Ke{link2:11}{1}{1} & \Ke{link2:1v}{1}{v} & \etk{\type1*} \cite{BelmonteS21} & \Ketk{link1234:1=}{1}{=} & \Kh{link2:1+}{1}{+} & \Kh{link2:1?}{1}{?} & \Kh{link2:1!}{1}{!}
    \\ \cline{2-11} & odd or $= 0$
    & \Ke{link2:v0}{v}{0} & \Kh{link2:v2}{v}{2} & \Ke{link2:v1}{v}{1} & \Ke{link2:vv}{v}{v} & \Kh{link2:v*}{v}{*} & \Ketk{link1234:v=}{v}{=} & \Kh{link2:v+}{v}{+} & \Kh{link2:v?}{v}{?} & \Kh{link2:v!}{v}{!}
    \\ \cline{2-11} & arbitrary
    & \Ketk{link1234:*0}{*}{0} & \Kh{link2:*2}{*}{2} & \Ke{link2:*1}{*}{1} & \Ketk{link1234:*v}{*}{v} & \Ketk{link1234:**}{*}{*} & \Ketk{link1234:*=}{*}{=} & \Ketk{link2:*+}{*}{+} & \Ketk{link1234:*?}{*}{?} & \Kh{link2:*!}{*}{!}
    \\ \cline{2-11} & $= 0$
    & \Ketk{link2:=0}{=}{0} & \Ketk{link2:=2}{=}{2} & \Ketk{link2:=1}{=}{1} & \Ketk{link2:=v}{=}{v} & \Ketk{link2:=*}{=}{*} & \Ketk{link1234:==}{=}{=} & \Ketk{link2:=+}{=}{+} & \Ketk{link2:=?}{=}{?} & \Ketk{link2:=!}{=}{!}
    \\ \cline{2-11} & $> 0$
    & \Ketk{link1234:+0}{+}{0} & \Kh{link2:+2}{+}{2} & \Kh{link2:+1}{+}{1} & \Ketk{link1234:+v}{+}{v} & \Ketk{link1234:+*}{+}{*} & \Ketk{link1234:+=}{+}{=} & \Kh{link2:++}{+}{+} & \Ketk{link1234:+?}{+}{?} & \Kh{link2:+!}{+}{!}
    \\ \cline{2-11} & $< 2$
    & \Kh{link2:?0}{?}{0} & \Kh{link2:?2}{?}{2} & \Kh{link2:?1}{?}{1} & \Kh{link234:?v}{?}{v} & \htk{\type?*} \cite{cowen1997defective} & \Ketk{link1234:?=}{?}{=} & \Kh{link234:?+}{?}{+} & \Ketk{link2:??}{?}{?} & \Ketk{link2:?!}{?}{!}
    \\ \cline{2-11} & $=1$
    & \Kh{link234:!0}{!}{0}  & \Kh{link2:!2}{!}{2} & \Kh{link2:!1}{!}{1} & \Kh{link234:!v}{!}{v} & \htk{\type!*} \cite{Sch78} & \Ketk{link1234:!=}{!}{=} & \Kh{link2:!+}{!}{+} & \Ketk{link2:!?}{!}{?} & \Ketk{link2:!!}{!}{!}
    \\ \hline
  \end{tabular}
  \caption{The complexity of the different types of colourings for two colours. Here a green box signifies that the problem is polynomial while a red/pink box indicates that the problem is \NP-hard. A light coloured box indicates known or trivial results (where known results are references and trivial results are explained in \cref{sec:trivial}) while a darker coloured box indicate new results.}
  \label{tab:casesTwoColours}
\end{center}
\end{table}

\begin{table}[p]
\begin{center}
  \renewcommand{\arraystretch}{1.7}
  \renewcommand{\l}[1]{\multirow{2}{*}{#1}}
  \newcommand{\m}[1]{\multirow{2}{*}{$#1$0}}
  \begin{tabular}{ll|*{9}{c|}}
    \multicolumn{2}{c|}{\multirow{3}{*}{$|N(v) \cap c^{-1}(i)|$}}
    & \multicolumn{9}{c|}{$i \ne c(v)$} \\[+1.7ex]
    \multicolumn{2}{c|}{}
    &\l{even} &even    &\l{odd} & odd    & arbi-  & \m{=}  & \m{>}  &\l{$<$2} &\l{$=1$}
    \\[-1.7ex] \multicolumn{2}{c|}{}                                 
    &       &and \!$>$0&        & or =0  & trary  &        &        & &
    \\ \hline \multirow{9}{*}{\hspace*{-3mm}\rotatebox{90}{$i=c(v)$}\hspace*{2mm}} & even
    & \Ketk{link1234:00}{0}{0} & \Kh{link34:02}{0}{2} & \Kh{link3:01}{0}{1} & \Kh{link34:0v}{0}{v} &\etk{\type0*} \cite{Lovasz93} & \Ketk{link1234:0=}{0}{=} & \Kh{link3:0+}{0}{+} & \Kh{link34:0?}{0}{?} & \Kh{link234:0!}{0}{!}
    \\ \cline{2-11} & even and $> 0$
    & \Ketk{link1234:20}{2}{0} & \Kh{link34:22}{2}{2} & \Kh{link34:21}{2}{1} & \Kh{link34:2v}{2}{v} & \Kh{link34:2*}{2}{*} & \Ketk{link1234:2=}{2}{=} & \Kh{link34:2+}{2}{+} & \Kh{link34:2?}{2}{?} & \Kh{link34:2!}{2}{!}
    \\ \cline{2-11} & odd
    & \Ketk{link1234:10}{1}{0} & \Kh{link34:12}{1}{2} & \Kh{link34:11}{1}{1} & \Kh{link34:1v}{1}{v} & \htk{\type1*} \cite{BelmonteS21} & \Ketk{link1234:1=}{1}{=} & \Kh{link34:1+}{1}{+} & \Kh{link3:1?}{1}{?} & \Kh{link34:1!}{1}{!}
    \\ \cline{2-11} & odd or $= 0$
    & \Kh{link34:v0}{v}{0} & \Kh{link34:v2}{v}{2} & \Kh{link34:v1}{v}{1} & \Kh{link34:vv}{v}{v} & \type{v}* & \Ketk{link1234:v=}{v}{=} & \Kh{link3:v+}{v}{+} & \Kh{link3:v?}{v}{?} & \Kh{link34:v!}{v}{!}
    \\ \cline{2-11} & arbitrary
    & \Ketk{link1234:*0}{*}{0} & \Kh{link34:*2}{*}{2} & \Kh{link34:*1}{*}{1} & \Ketk{link1234:*v}{*}{v} & \Ketk{link1234:**}{*}{*} & \Ketk{link1234:*=}{*}{=} & \Kh{link3:*+}{*}{+} & \Ketk{link1234:*?}{*}{?} & \Kh{link34:*!}{*}{!}
    \\ \cline{2-11} & $= 0$
    & \Kh{link34:=0}{=}{0} & \Kh{link34:=2}{=}{2} & \Kh{link34:=1}{=}{1} & \Kh{link34:=v}{=}{v} & \htk{\type=*} \cite{G&J} & \Ketk{link1234:==}{=}{=} & \Kh{link34:=+}{=}{+} & \Ketk{link3:=?}{=}{?} & \Ketk{link3:=!}{=}{!}
    \\ \cline{2-11} & $> 0$
    & \Ketk{link1234:+0}{+}{0} & \Kh{link34:+2}{+}{2} & \Kh{link34:+1}{+}{1} & \Ketk{link1234:+v}{+}{v} & \Ketk{link1234:+*}{+}{*} & \Ketk{link1234:+=}{+}{=} & \Kh{link34:++}{+}{+} & \Ketk{link1234:+?}{+}{?} & \Kh{link34:+!}{+}{!}
    \\ \cline{2-11} & $< 2$
    & \Kh{link34:?0}{?}{0} & \Kh{link34:?2}{?}{2} & \Kh{link34:?1}{?}{1} & \Kh{link234:?v}{?}{v} & \htk{\type?*} \cite{cowen1997defective} & \Ketk{link1234:?=}{?}{=} & \Kh{link234:?+}{?}{+} & \Kh{link34:??}{?}{?} & \Kh{link34:?!}{?}{!}
    \\ \cline{2-11} & $=1$
    & \Kh{link234:!0}{!}{0} & \Kh{link34:!2}{!}{2} & \Kh{link34:!1}{!}{1} & \Kh{link234:!v}{!}{v} & \htk{\type!*} \cite{DemaineKP25} & \Ketk{link1234:!=}{!}{=} & \Kh{link34:!+}{!}{+} & \Kh{link34:!?}{!}{?} & \Kh{link34:!!}{!}{!}
    \\ \hline
  \end{tabular}
  \caption{The complexity of the different types of colourings for three colours. }
  \label{tab:casesThreeColours}
\end{center}
\end{table} 

\begin{table}[p]
\begin{center}
  \renewcommand{\arraystretch}{1.7}
  \renewcommand{\l}[1]{\multirow{2}{*}{#1}}
  \newcommand{\m}[1]{\multirow{2}{*}{$#1$0}}
  \begin{tabular}{ll|*{9}{c|}}
    \multicolumn{2}{c|}{\multirow{3}{*}{$|N(v) \cap c^{-1}(i)|$}}
    & \multicolumn{9}{c|}{$i \ne c(v)$} \\[+1.7ex]
    \multicolumn{2}{c|}{}
    &\l{even} &even    &\l{odd} & odd    & arbi-  & \m{=}  & \m{>}  &\l{$<$2} &\l{$=1$}
    \\[-1.7ex] \multicolumn{2}{c|}{}                                 
    &       &and \!$>$0&        & or =0  & trary  &        &        & &
    \\ \hline \multirow{9}{*}{\hspace*{-3mm}\rotatebox{90}{$i=c(v)$}\hspace*{2mm}} & even
    & \Ketk{link1234:00}{0}{0} & \Kh{link34:02}{0}{2} & \Kh{link4:01}{0}{1} & \Kh{link34:0v}{0}{v} &\etk{\type0*} \cite{Lovasz93} & \Ketk{link1234:0=}{0}{=} & \Kh{link4:0+}{0}{+} & \Kh{link34:0?}{0}{?} & \Kh{link234:0!}{0}{!}
    \\ \cline{2-11} & even and $> 0$
    & \Ketk{link1234:20}{2}{0} & \Kh{link34:22}{2}{2} & \Kh{link34:21}{2}{1} & \Kh{link34:2v}{2}{v} & \Kh{link34:2*}{2}{*} & \Ketk{link1234:2=}{2}{=} & \Kh{link34:2+}{2}{+} & \Kh{link34:2?}{2}{?} & \Kh{link34:2!}{2}{!}
    \\ \cline{2-11} & odd
    & \Ketk{link1234:10}{1}{0} & \Kh{link34:12}{1}{2} & \Kh{link34:11}{1}{1} & \Kh{link34:1v}{1}{v} & \htk{\type1*} \cite{BelmonteS21} & \Ketk{link1234:1=}{1}{=} & \Kh{link34:1+}{1}{+} & \Kh{link4:1?}{1}{?} & \Kh{link34:1!}{1}{!}
    \\ \cline{2-11} & odd or $= 0$
    & \Kh{link34:v0}{v}{0} & \Kh{link34:v2}{v}{2} & \Kh{link34:v1}{v}{1} & \Kh{link34:vv}{v}{v} & \type{v}* & \Ketk{link1234:v=}{v}{=} & \Kh{link4:v+}{v}{+} & \Kh{link4:v?}{v}{?} & \Kh{link34:v!}{v}{!}
    \\ \cline{2-11} & arbitrary
    & \Ketk{link1234:*0}{*}{0} & \Kh{link34:*2}{*}{2} & \Kh{link34:*1}{*}{1} & \Ketk{link1234:*v}{*}{v} & \Ketk{link1234:**}{*}{*} & \Ketk{link1234:*=}{*}{=} & \Kh{link4:*+}{*}{+} & \Ketk{link1234:*?}{*}{?} & \Kh{link34:*!}{*}{!}
    \\ \cline{2-11} & $= 0$
    & \Kh{link34:=0}{=}{0} & \Kh{link34:=2}{=}{2} & \Kh{link34:=1}{=}{1} & \Kh{link34:=v}{=}{v} & \htk{\type=*} \cite{G&J} & \Ketk{link1234:==}{=}{=} & \Kh{link34:=+}{=}{+} & \Kh{link4:=?}{=}{?} & \Kh{link4:=!}{=}{!}
    \\ \cline{2-11} & $> 0$
    & \Ketk{link1234:+0}{+}{0} & \Kh{link34:+2}{+}{2} & \Kh{link34:+1}{+}{1} & \Ketk{link1234:+v}{+}{v} & \Ketk{link1234:+*}{+}{*} & \Ketk{link1234:+=}{+}{=} & \Kh{link34:++}{+}{+} & \Ketk{link1234:+?}{+}{?} & \Kh{link34:+!}{+}{!}
    \\ \cline{2-11} & $< 2$
    & \Kh{link34:?0}{?}{0} & \Kh{link34:?2}{?}{2} & \Kh{link34:?1}{?}{1} & \Kh{link234:?v}{?}{v} & \htk{\type?*} \cite{cowen1997defective} & \Ketk{link1234:?=}{?}{=} & \Kh{link234:?+}{?}{+} & \Kh{link34:??}{?}{?} & \Kh{link34:?!}{?}{!}
    \\ \cline{2-11} & $=1$
    & \Kh{link234:!0}{!}{0} & \Kh{link34:!2}{!}{2} & \Kh{link34:!1}{!}{1} & \Kh{link234:!v}{!}{v} & \htk{\type!*} \cite{DemaineKP25} & \Ketk{link1234:!=}{!}{=} & \Kh{link34:!+}{!}{+} & \Kh{link34:!?}{!}{?} & \Kh{link34:!!}{!}{!}
    \\ \hline
  \end{tabular}
  \caption{The complexity of the different types of colourings for four or more colours. }
  \label{tab:cases}
\end{center}
\end{table}

Several known colouring problems correspond to one of the considered variants. For example, the type \type=* refers to the usual proper colourings, \type?* refers to defective colouring with defect $d=1$, \type!* corresponds to exact defective colouring with defect $d=1$, \type0* refers to even colouring, \type1* is odd colouring and finally \type=v corresponds to strong odd colouring. 

We systematically study the computational complexity of combinations of such constraints. Our results are summarised in \cref{tab:casesTwoColours} for two colours, \cref{tab:casesThreeColours} for three colours and \cref{tab:cases} for $q\geq 4$ colours. As can be seen in the tables, the only remaining open case constitutes \type{v}{*}-colouring for $q\geq 3$ colours. For two colours we show that this problem is \NP-complete. We observe differences in the complexity for small number of colours between different variations as can be seen from the three tables.
Unfortunately, besides in the restricted case of two colours, we conclude that all variants are either trivially solvable or \NP-complete.

An interesting phenomenon that we encounter, through considering different variants, is that colouring of regular graphs is of particular importance as many variants collapse to the same problem when considering this restriction. We use this extensively for our reductions for two colours.  Furthermore, when considering three or more colours, we  define 
the following variants of colouring of regular graphs that may be of independent interest. An \emph{improper rainbowcolouring} of a graph $G$ is a colouring $c\colon V(G)\rightarrow \{1,\dots,q\}$ in which for each vertex $v\in V(G)$ every colour $i\in \{1,\dots,q\}$ appears exactly once within the open neighbourhood $N(v)$ of $v$. A \emph{proper rainbow colouring} of a graph $G$ is a proper colouring $c\colon V(G)\rightarrow \{1,\dots,q\}$ in which for each vertex $v\in V(G)$ every colour $i\in \{1,\dots,q\}\setminus \{c(v)\}$ appear exactly once within the open neighbourhood $N(v)$ of $v$. Naturally, both variants are only sensible to consider on regular graphs ($q+1$-regular for improper and $q$-regular for proper rainbow colouring). While a reduction from edge colouring for improper rainbow colouring is not too difficult  (see \cref{ss:improperRainbow}), this seems more difficult in the proper rainbow case. However, proper rainbow colouring is a special case of a covering problem in which the graph $H$ that the input graph $G$ is ``covered'' with is equal to $K_q$ for $q$ the number of colours (see \cite{kratochvil1997covering}). A result by Fiala and Kratochv{\'\i}l implies \NP-completeness of proper rainbow colouring for $q\geq 4$ colours (see \cref{ss:proper-rainbow} for details). To the best of our knowledge, it is unknown whether proper rainbow colouring is hard to solve for three colours.

\noindent \textbf{Other related colouring variants}
A \emph{b-colouring} \cite{b-col} is a proper colouring where one vertex per colour class is adjacent
to vertices of every other colour, and a \emph{fall colouring} \cite{fall} this is required
for all vertices. Several decision problems related to these colourings are \NP-complete,
but solvable in polynomial time for restricted classes of inputs \cite{b&fall}. For graphs regular of degree $d$ the fall $(d+1)$-colourings coincide with \emph{proper rainbow colourings}.

\noindent \textbf{Structure of the paper}
Due to space restriction, we only provide a small sample of results within the main text. The remaining results can be found in the appendix, where they are grouped according to whether they hold from two (\cref{sec:twoColours}) or  three (\cref{sec:threeColours}) colours and according to which problem the corresponding reduction is obtained from.
To establish easy navigation, Tables \ref{tab:casesTwoColours}, \ref{tab:casesThreeColours} and \ref{tab:cases} provide clickable links to the proofs of the respective results.

\section{Preliminaries}
We denote the set of natural number excluding $0$ by $\mathbb{N}$. For $n\in \mathbb{N}$ we denote by $[n]$ the set $\{1,\dots, n\}\subseteq \mathbb{N}$.

In the following all graphs are undirected and simple. For a graph $G$ we denote the set of vertices of $G$ by $V(G)$ and the set of edges of $G$ by $V(G)$. For $X \subseteq V(G)$ we set $X^{(2)} = \{vw \mid v \in X, w \in X, v \ne w\}$.
We use the shorthand notation $vw$ to denote edge $\{v,w\}$. 
For a set $X\subseteq V(G)$ we denote the graph \emph{induced} by the set $X$ by $G[X]$.
For a vertex $v\in V(G)$ we let $N_G(v)$ be the \emph{open neighbourhood} of $G$, \textsl{i.e.} the set of all vertices $w\in V(G)$ for which $vw\in E(G)$. Furthermore, we let $N_G[v]:=N_G(v)\cup \{v\}$ be the \emph{closed neighbourhood} of $v$. We extend the notation to sets by defining $N_G(X):= \big(\bigcup_{v\in X}N_G(v)\big) \setminus X$ and $N_G[X]:=\bigcup_{v\in X}N_G[v]$ for any subset $X\subseteq V(G)$. The \emph{degree} of vertex $v\in V(G)$ is the number $|N_G(v)|$ and is denoted by $\deg_G(v)$. In the preceding notation we omit $G$ if the graph is clear from context. 

A \emph{colouring} of a graph $G$ with $q\in \mathbb{N}$ colours is a map $c\colon V(G)\rightarrow [q]$. We say that vertex $v\in V(G)$ \emph{receives} colour $i\in [q]$ if $c(v)=i$. 
For $i\in [q]$ we call the set $\{v\in V(G)\mid c(v)=i\}$ of all vertices receiving colour $i$ the \emph{colour class of colour $i$}. We encode parity constraint of colourings as follows. We say that constraint $\sigma\in \{\typesymbol{0}, \typesymbol{1}, \typesymbol{2}, \typesymbol{v}, \typesymbol{*}, \typesymbol{=}, \typesymbol{+}, \typesymbol{?}, \typesymbol{!}\}$ holds at a vertex $v\in V(G)$ for a colour $i\in [q]$ in a colouring $c\colon V(G)\rightarrow [q]$ if $|N(v)\cap c^{-1}(i)|$ satisfies the condition as specified in Table~\ref{tab:TypeSymbols}.

For $\sigma_1,\sigma_2\in \{\typesymbol{0}, \typesymbol{1}, \typesymbol{2}, \typesymbol{v}, \typesymbol{*}, \typesymbol{=}, \typesymbol{+}, \typesymbol{?}, \typesymbol{!}\}$ we say that a colouring $c$ with $q\in \mathbb{N}$ colours is a $\sigma_1\sigma_2$-$q$-colouring if constraint $\sigma_1$ holds at every vertex $v\in V(G)$ for colour $c(v)$ and constraint $\sigma_2$ holds for every vertex $v\in V(G)$ for colour $i\neq c(v)$ in the colouring $c$.

We denote the complete graph on $n\in \mathbb{N}$ vertices by $K_n$ and the complete bipartite graph with parts of size $m\in \mathbb{N}$ and $n\in \mathbb{N}$ by $K_{m,n}$.
For a $k\in \mathbb{N}$ we call a graph $G$ \emph{$k$-regular} if every vertex $v\in V(G)$ has degree $k$ and we call $3$-regular graphs \emph{cubic}. A \emph{$2$-factor} of a graph $G$ is a spanning subgraph in which every vertex has degree $2$.
Two vertices $v$ and $w$ are called \emph{true twins} if their closed
neighbourhoods are equal, $N[v]=N[w]$.  Two vertices $v$ and $w$ are called \emph{false twins} if their open
neighbourhoods are equal, $N(v)=N(w)$.

\begin{table}[ht]
  \centering
\begin{tabular}{|c|c|c|} 
 \hline
 $\sigma$ & formal constraint on $C = N(u)\cap c^{-1}(i)$ & informal constraint on $C$ \\ 
 \hline\hline
 $\typesymbol{0}$ & $|C| \equiv 0 \pmod 2$ & $C$ is even\\ 
 \hline
 $\typesymbol{1}$ & $|C| \equiv 1 \pmod 2$ & $C$ is odd\\ 
 \hline
 $\typesymbol{2}$ & $|C| \equiv 0 \pmod 2$ and $|C|\ge2$ & $C$ is even and nonempty \\ 
 \hline
 $\typesymbol{v}$ & $|C| \equiv 1 \pmod 2$ or $|C|=0$ & $C$ is odd or empty \\ 
 \hline
 $\typesymbol{*}$ & $|C|$ is not constrained & none \\ 
 \hline
 $\typesymbol{=}$ & $|C|=0$ & $C$ is empty \\ 
 \hline
 $\typesymbol{+}$ & $|C|>0$ & $C$ is nonempty \\ 
 \hline
 $\typesymbol{?}$ & $|C|<2$ & $|C|$ is zero or one \\ 
 \hline
 $\typesymbol{!}$ & $|C|=1$ & $|C|$ is exactly one \\ 
 \hline
\end{tabular}
 \caption{The encoding of the different colouring constraints.}
 \label{tab:TypeSymbols}
\end{table}

\section{Complexity for two and potentially more colours}\label{sec:twoColours}


\subsection{Reductions from \NSAT}
\ESAT is the satisfiability problem restricted to formulae with exactly three
literals per clause. That is, its input is a boolean formula $\varphi$ in
conjunctive normal form, $\varphi = \bigwedge_{j=1}^{m} c_j$. Each clause
$c_j$ is the disjunction of exactly three literals, $c_j = (\ell_{j,1} \vee
\ell_{j,2} \vee \ell_{j,3})$. If $X = \{x_1, x_2, \dots, x_n\}$ is the set of
variables that appear in $\varphi$ then every literal is a variable $x_i \in X$
or its negation $\neg x_i$, often denoted as $\bar{x}_i$. The question is whether
there exists a truth assignment $\alpha\colon X \to \{\mathsf{true}, \mathsf{false}\}$ 
such that $\alpha(\varphi) = \mathsf{true}$.
\NSAT (not-all-equal \ESAT) has the same domain as \ESAT. Instead of asking for at least one true
literal per clause, we ask for at least one true literal and at least one false literal. 
Formally, for $\alpha\colon X \to \{\mathsf{true}, \mathsf{false}\}$ let
$\bar{\alpha}$ be defined by
$\bar{\alpha}\colon X \to \{\mathsf{true}, \mathsf{false}\} \quad x \mapsto \bar{\alpha}(x) = \alpha(\bar{x})$. 
The instance $\varphi$ is accepted for \NSAT if there is a truth
assignment $\alpha$ such that both $\alpha(\varphi) = \mathsf{true}$ and
$\bar{\alpha}(\varphi) = \mathsf{true}$.

As a consequence, we may assume that for every clause
$(\ell_i \vee \ell_2 \vee \ell_3)$ of $\varphi$ there is also a clause
$(\neg\ell_i \vee \neg\ell_2 \vee \neg\ell_3)$ in $\varphi$.
In this way, we ensure that every variable $x \in X$ appears an even number of
times in $\varphi$, half of them positively as $x$ and half of them negatively
as $\bar{x}$.

Our constructions for $q=2$ colours use the component design technique. Presume we are considering $\sigma$-colourings for one of the variants considered in table \cref{tab:cases}.
For each variable, we construct a \emph{truth assignment component} or
\emph{tac} for short. Typically, it allows for two $\sigma$-$2$-colourings only, encoding
the two possible truth values. More precisely, the tac has an outlet (outlets can consist of more than one vertex and in the tac there can be edges between outlet vertices) for each occurrence of the variable in the input formula $\varphi$. Essentially, the outlet specifies how the gadget is glued to the rest of the graph. There is a colouring of the closed neighbourhood of an outlets encoding that the literal which the outlet represents is set to true 
and a different colouring representing that the literal is set to false.  A colouring of the tac is a $\sigma$-2-colouring if the conditions of $\sigma$-2-colourings are satisfied for each vertex in the gadget apart from the outlets. The tac must have the following property:

\begin{align*}
\tag{$\Pi_{\operatorname{tac}}$}\label{property:tac}
     &\textit{For every variable $x$ the corresponding tac has a $\sigma$-2-colouring satisfying that }\\
     &\textit{the colours of the closed neighbourhood of each outlet representing an occurrence}\\
     &\textit{of literal $x$ encodes true/false and the colours at the closed neighbourhood of each}\\
     &\textit{outlet representing an occurrence of $\overline{x}$ encodes false/true. Additionally, up to }\\ 
     &\textit{permutation of colours, each  $\sigma$-$2$-colouring of the tac satisfies one of the two }\\
     &\textit{conditions above.}
\end{align*}
For each variable, we construct a
\emph{satisfaction test component} or \emph{stc} that tests whether the three
literals that make up the clause have the right truth values: a true and a false literal. The tac will have one outlet for each literal of the clause. 
Additionally, there is a colouring of the closed neighbourhood of an outlet of the tac that encodes that the literal is true and a different colouring that encodes that the literal is false. A $2$-colouring of the stc is a $\sigma$-2-colouring if the conditions of $\sigma$-2-colourings are satisfied for each vertex in the gadget apart from the outlets. A $\sigma$-2-colouring of the stc is \emph{valid} if the closed neighbourhood of each outlet represents either true or false. The stc must satisfy the following:
\begin{align*}
\tag{$\Pi_{\operatorname{stc}}$}\label{property:stc}
&\textit{For each clause $c$ of $\varphi$ every valid $\sigma$-2-colouring of the stc corresponding to $c$ }\\
&\textit{satisfies that the colours of the closed neighbourhood of at least one outlet encode }\\
&\textit{true and the colours of the closed neighbourhood of at least one outlet encode false.}
\end{align*}

We further require that the stc and tac can be glued together and appropriate $\sigma$-2-colourings of the tac and stc can be combined. Assume $S$ is the stc of a clause $c$ containing a literal $\ell$ equal to $x$ or $\overline{x}$ and $s_1,\dots, s_k$ are the outlets representing $\ell$. Additionally, assume  that $T$ is the tac of variable $x$ and $t_1,\dots,t_\ell$ are the vertices of the outlet of $T$ corresponding to the occurrence of $x$ in $c$. The aim is to identify the outlets of the stc with the neighbours of the outlets of the tac and vice versa. Formally, we say that the tac $T$ and stc $S$ are \emph{compatible} if there is an isomorphism $\iota: N_S[s_1,\dots, s_k]\rightarrow N_T[t_1,\dots, t_\ell]$ between the graph $S[N_S[s_1,\dots, s_k]]$ and $T[N_T[t_1,\dots, t_\ell]]$ such that $\iota(s_1,\dots,s_k)=N_T(t_1,\dots, t_\ell)$ (which also implies that $\iota^{-1}(t_1,\dots,t_\ell)=N_S(s_1,\dots, s_k)$). In this case, we define the \emph{glueing operation} by identifying vertex $v$ of $S$ with vertex $\iota(v)$ of $T$. 
We further say that a $\sigma$-$2$-colouring $c_S$ of $S$ and a $\sigma$-$2$-colouring of $T$ are \emph{compatible} if $c_S(v)=c_T(\iota(v))$ for every $v\in N_S[s_1,\dots, s_k]$. Define the following property:
\begin{align*}
    \tag{$\Pi_+$}\label{property:combiningSTCandTAC}
    &\textit{The stc and tac are compatible and, after globally fixing two $\sigma$-$2$-colourings of }\\
    &\textit{the tac (encoding true and false) and one $\sigma$-$2$-colouring for each assignment}\\
    &\textit{of literals of the stc the following holds. For every clause $c$ containing literal }\\
    &\textit{$\ell\in \{x,\overline{x}\}$ any $\sigma$-$2$-colouring of $c$ that encodes that $\ell$ is true/false is compatible}\\
    &\textit{with the $\sigma$-$2$-colouring of the tac of $x$ that encodes that $\ell$ is true/false.}
\end{align*}
We call $H_\varphi$ the graph obtained by gluing the stc and tac along matching occurances of literals as described above. From the properties \ref{property:tac}, \ref{property:stc} and \ref{property:combiningSTCandTAC} it easily follows that $\phi$ is a YES-instance of \NSAT if and only if $H_\phi$ has a $\sigma$-$2$-colouring and therefore the following holds.

\begin{lemma}\label{lem:reductionFromNAE3SAT}
    If we have a tac satisfying \ref{property:tac} and an stc satisfying \ref{property:tac} which together satisfy \ref{property:combiningSTCandTAC} then $\sigma$-$q$-colouring is \NP-complete for $q=2$ colours.
\end{lemma}

\subsubsection{\Type{2}{?}, \Type{0}{?}, \Type{2}{1}, \Type{2}{v}, \Type{2}{*}, \Type{2}{!}, \Type{*}{!}, \Type{+}{1} and \Type{+}{!}-colouring for \boldmath$q=2$}
First recall that $c:V(G)\rightarrow [q]$ is a \type2?-$q$-colouring of $G$ if for every vertex $v$ we have $|N_G(v)\cap c^{-1}(c(v))|$ is even and larger than $0$  and $|N_G(v)\cap c^{-1}(c(v))|<2$ for every $i\neq c(v)$. 

\begin{theorem}\label{link2:2?}
    For $q = 2$ it is \NP-complete to decide whether a graph admits a \type2?-$q$-colouring.
\end{theorem}
\begin{proof}
    We use a reduction from \NSAT to show that \type2?-colouring is
\NP-complete for two colours even when restricted to cubic graphs. For $q=2$
and general graphs, each vertex of even degree has the same colour as all its
neighbours in every \type2?-colouring. Vertices of odd degree have exactly one
neighbour of the opposite colour. For a cubic graph $G$ that means that $G$
is \type2?-colourable by two colours if and only if $G$ has a perfect matching
such that every cycle contains an even number of matching edges. Especially,
all vertices in a triangle receive the same colour, and all their neighbours
outside the triangle receive the opposite colour in every \type2?-colouring with
$q=2$ colours.
For a variable $x$ that appears $k$ times positively in the \NSAT-formula and $k$
times negatively we create a tac consisting of a ring of $2k$ triangles with an outlet consisting of a single vertex attached to each triangle This
tac allows for exactly two \type2?-colourings shown in Figure~\ref{fig:2?tac}. True is encoded by the outlet vertex having colour $1$ (and its neighbour having colour $2$ and false is encoded by the outlet vertex having colour $2$ (and its neighbour colour $1$). The tac naturally satisfies property \ref{property:tac}. 

\begin{figure}[hbtp]
  \hspace*{\fill}
  \begin{tikzpicture}[scale=0.5]
    \foreach \n in {1,...,24} {
      \pgfmathtruncatemacro{\a}{30*\n-15}
      \pgfmathtruncatemacro{\v}{\n/2}
      \node[\ifodd\v{b}\else{r}\fi] (\n) at (\a:1.932) {};
    }
    \foreach \n in {1,...,6} {
      \pgfmathtruncatemacro{\a}{60*\n}
      \node[\ifodd\n{b}\else{r}\fi]   (t\n) at (\a:2.732) {};
      \node[\ifodd\n{gr}\else{gb}\fi] (o\n) at (\a:3.732) {};
      \draw[lightgray,double] (t\n)--(o\n);
    }
    \draw[blue](2)--(t1)--(3)--(2)  (6)--(t3)--(7)--(6)  (10)--(t5)--(11)--(10);
    \draw[red] (4)--(t2)--(5)--(4)  (8)--(t4)--(9)--(8)  (12)--(t6)--(1) --(12);
    \draw[double] (1)--(2)  (3)--(4)  (5)--(6)  (7)--(8)  (9)--(10)  (11)--(12);
  \end{tikzpicture}
  \hspace*{\fill}
  \begin{tikzpicture}[scale=0.5]
    \foreach \n in {1,...,24} {
      \pgfmathtruncatemacro{\a}{30*\n-15}
      \pgfmathtruncatemacro{\v}{\n/2}
      \node[\ifodd\v{r}\else{b}\fi] (\n) at (\a:1.932) {};
    }
    \foreach \n in {1,...,6} {
      \pgfmathtruncatemacro{\a}{60*\n}
      \node[\ifodd\n{r}\else{b}\fi]   (t\n) at (\a:2.732) {};
      \node[\ifodd\n{gb}\else{gr}\fi] (o\n) at (\a:3.732) {};
      \draw[lightgray,double] (t\n)--(o\n);
    }
    \draw[red] (2)--(t1)--(3)--(2)  (6)--(t3)--(7)--(6)  (10)--(t5)--(11)--(10);
    \draw[blue](4)--(t2)--(5)--(4)  (8)--(t4)--(9)--(8)  (12)--(t6)--(1) --(12);
    \draw[double] (1)--(2)  (3)--(4)  (5)--(6)  (7)--(8)  (9)--(10)  (11)--(12);
  \end{tikzpicture}
  \hspace*{\fill}
  \caption{The two \type2?-colourings of a tac for a variable that appears
    trice positively and trice negatively. Edges
are single black lines if their endpoints colours are the same,
and by double black lines otherwise. So the double black lines indicate
the perfect matching.}
  \label{fig:2?tac}
\end{figure}
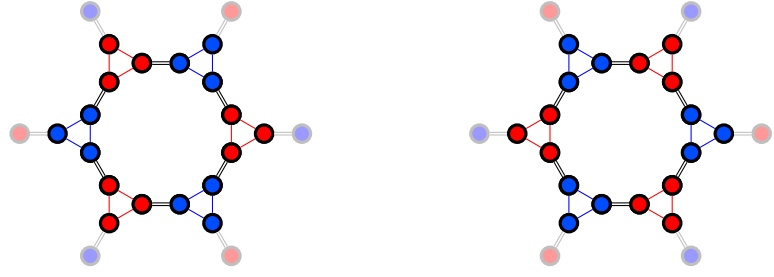

For each clause $c$ we create an stc that is obtained from a Petersen graph by
subdividing three edges whose endpoints induce a $6$-cycle. Outlets consisting of a single vertex are attached to the three subdivision vertices (again colour $1$ at an outlet encodes true and colour $2$ encodes false). If all three subdivision vertices receive the same colour then the entire outer cycle has to be
coloured that way. A claw remains from the stc for the opposite colour. But to
be \type2?-colourable these four vertices should induce a $4$-cycle, which
does not exist in the Petersen graph. Hence, property \ref{property:stc} is satisfied. On the other hand, the three subdivision vertices in the stc allow for the not-all-equal colourings depicted in
Figure~\ref{fig:2?stc}.
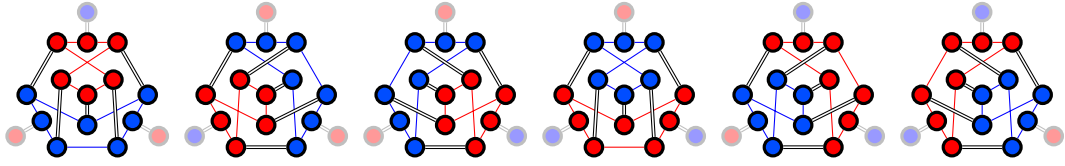
\begin{figure}[hbtp] 
  \begin{tikzpicture}[scale=0.4]
    \node[r] (o) at   (0,0) {};
    \node[r] (r) at ( 30:1) {};
    \node[r] (l) at (150:1) {};
    \node[b] (u) at (270:1) {};
    \foreach[count=\n] \c in {r,r,b,b,b,b} {
      \pgfmathtruncatemacro{\a}{60*\n}
      \node[\c] (\n) at (\a:2) {};
    }
    \node[r]  (12) at ($0.5*(1)+0.5*(2)$) {};
    \node[gb] (21) at ($(12)   +( 90:1)$) {}; \draw[lightgray,double] (12)--(21);
    \node[b]  (34) at ($0.5*(3)+0.5*(4)$) {};
    \node[gr] (43) at ($(34)   +(210:1)$) {}; \draw[lightgray,double] (34)--(43);
    \node[b]  (56) at ($0.5*(5)+0.5*(6)$) {};
    \node[gr] (65) at ($(56)   +(330:1)$) {}; \draw[lightgray,double] (56)--(65);
    \draw[red]  (1)--(12)--(2)--(r)-- (o)--(l)--(1);
    \draw[blue] (3)--(34)--(4)--(5)--(56)--(6)--(u)--(3);
    \draw[double] (2)--(3)  (4)--(l)  (5)--(r)  (6)--(1)  (u)--(o);
  \end{tikzpicture}
  \hfill 
  \begin{tikzpicture}[scale=0.4]
    \node[r] (o) at   (0,0) {};
    \node[b] (r) at ( 30:1) {};
    \node[r] (l) at (150:1) {};
    \node[r] (u) at (270:1) {};
    \foreach[count=\n] \c in {b,b,r,r,b,b} {
      \pgfmathtruncatemacro{\a}{60*\n}
      \node[\c] (\n) at (\a:2) {};
    }
    \node[b]  (12) at ($0.5*(1)+0.5*(2)$) {};
    \node[gr] (21) at ($(12)   +( 90:1)$) {}; \draw[lightgray,double] (12)--(21);
    \node[r]  (34) at ($0.5*(3)+0.5*(4)$) {};
    \node[gb] (43) at ($(34)   +(210:1)$) {}; \draw[lightgray,double] (34)--(43);
    \node[b]  (56) at ($0.5*(5)+0.5*(6)$) {};
    \node[gr] (65) at ($(56)   +(330:1)$) {}; \draw[lightgray,double] (56)--(65);
    \draw[red]  (3)--(34)--(4)--(l)-- (o)--(u)--(3);
    \draw[blue] (5)--(56)--(6)--(1)--(12)--(2)--(r)--(5);
    \draw[double] (1)--(l)  (2)--(3)  (4)--(5)  (6)--(u)  (r)--(o);
  \end{tikzpicture}
  \hfill 
  \begin{tikzpicture}[scale=0.4]
    \node[r] (o) at   (0,0) {};
    \node[r] (r) at ( 30:1) {};
    \node[b] (l) at (150:1) {};
    \node[r] (u) at (270:1) {};
    \foreach[count=\n] \c in {b,b,b,b,r,r} {
      \pgfmathtruncatemacro{\a}{60*\n}
      \node[\c] (\n) at (\a:2) {};
    }
    \node[b]  (12) at ($0.5*(1)+0.5*(2)$) {};
    \node[gr] (21) at ($(12)   +( 90:1)$) {}; \draw[lightgray,double] (12)--(21);
    \node[b]  (34) at ($0.5*(3)+0.5*(4)$) {};
    \node[gr] (43) at ($(34)   +(210:1)$) {}; \draw[lightgray,double] (34)--(43);
    \node[r]  (56) at ($0.5*(5)+0.5*(6)$) {};
    \node[gb] (65) at ($(56)   +(330:1)$) {}; \draw[lightgray,double] (56)--(65);
    \draw[red]  (5)--(56)--(6)--(u)-- (o)--(r)--(5);
    \draw[blue] (1)--(12)--(2)--(3)--(34)--(4)--(l)--(1);
    \draw[double] (1)--(6)  (2)--(r)  (3)--(u)  (4)--(5)  (l)--(o);
  \end{tikzpicture}
  \hfill 
  \begin{tikzpicture}[scale=0.4]
    \node[b] (o) at   (0,0) {};
    \node[b] (r) at ( 30:1) {};
    \node[b] (l) at (150:1) {};
    \node[r] (u) at (270:1) {};
    \foreach[count=\n] \c in {b,b,r,r,r,r} {
      \pgfmathtruncatemacro{\a}{60*\n}
      \node[\c] (\n) at (\a:2) {};
    }
    \node[b]  (12) at ($0.5*(1)+0.5*(2)$) {};
    \node[gr] (21) at ($(12)   +( 90:1)$) {}; \draw[lightgray,double] (12)--(21);
    \node[r]  (34) at ($0.5*(3)+0.5*(4)$) {};
    \node[gb] (43) at ($(34)   +(210:1)$) {}; \draw[lightgray,double] (34)--(43);
    \node[r]  (56) at ($0.5*(5)+0.5*(6)$) {};
    \node[gb] (65) at ($(56)   +(330:1)$) {}; \draw[lightgray,double] (56)--(65);
    \draw[blue] (1)--(12)--(2)--(r)-- (o)--(l)--(1);
    \draw[red]  (3)--(34)--(4)--(5)--(56)--(6)--(u)--(3);
    \draw[double] (2)--(3)  (4)--(l)  (5)--(r)  (6)--(1)  (u)--(o);
  \end{tikzpicture}
  \hfill 
  \begin{tikzpicture}[scale=0.4]
    \node[b] (o) at   (0,0) {};
    \node[r] (r) at ( 30:1) {};
    \node[b] (l) at (150:1) {};
    \node[b] (u) at (270:1) {};
    \foreach[count=\n] \c in {r,r,b,b,r,r} {
      \pgfmathtruncatemacro{\a}{60*\n}
      \node[\c] (\n) at (\a:2) {};
    }
    \node[r]  (12) at ($0.5*(1)+0.5*(2)$) {};
    \node[gb] (21) at ($(12)   +( 90:1)$) {}; \draw[lightgray,double] (12)--(21);
    \node[b]  (34) at ($0.5*(3)+0.5*(4)$) {};
    \node[gr] (43) at ($(34)   +(210:1)$) {}; \draw[lightgray,double] (34)--(43);
    \node[r]  (56) at ($0.5*(5)+0.5*(6)$) {};
    \node[gb] (65) at ($(56)   +(330:1)$) {}; \draw[lightgray,double] (56)--(65);
    \draw[blue] (3)--(34)--(4)--(l)-- (o)--(u)--(3);
    \draw[red]  (5)--(56)--(6)--(1)--(12)--(2)--(r)--(5);
    \draw[double] (1)--(l)  (2)--(3)  (4)--(5)  (6)--(u)  (r)--(o);
  \end{tikzpicture}
  \hfill 
  \begin{tikzpicture}[scale=0.4]
    \node[b] (o) at   (0,0) {};
    \node[b] (r) at ( 30:1) {};
    \node[r] (l) at (150:1) {};
    \node[b] (u) at (270:1) {};
    \foreach[count=\n] \c in {r,r,r,r,b,b} {
      \pgfmathtruncatemacro{\a}{60*\n}
      \node[\c] (\n) at (\a:2) {};
    }
    \node[r]  (12) at ($0.5*(1)+0.5*(2)$) {};
    \node[gb] (21) at ($(12)   +( 90:1)$) {}; \draw[lightgray,double] (12)--(21);
    \node[r]  (34) at ($0.5*(3)+0.5*(4)$) {};
    \node[gb] (43) at ($(34)   +(210:1)$) {}; \draw[lightgray,double] (34)--(43);
    \node[b]  (56) at ($0.5*(5)+0.5*(6)$) {};
    \node[gr] (65) at ($(56)   +(330:1)$) {}; \draw[lightgray,double] (56)--(65);
    \draw[blue] (5)--(56)--(6)--(u)-- (o)--(r)--(5);
    \draw[red]  (1)--(12)--(2)--(3)--(34)--(4)--(l)--(1);
    \draw[double] (1)--(6)  (2)--(r)  (3)--(u)  (4)--(5)  (l)--(o);
  \end{tikzpicture}
  \caption{The six different \type2?-colourings of an stc}
  \label{fig:2?stc}
\end{figure}
We finally observe that the tac and stc are compatible and the \type2?-$2$-colourings of the tac and stc are compatible as well. Hence, property \ref{property:combiningSTCandTAC} is satisfied and by \cref{lem:reductionFromNAE3SAT} we obtain that \type2?-$2$-colouring is \NP-complete.
\end{proof}

\begin{corollary}\label{link2:0?}\label{link2:21}\label{link2:2v}\label{link2:2*}\label{link2:2!}\label{link2:*!}\label{link2:+1}\label{link2:+!}
    For $q=2$ it is \NP-complete to decide whether a graph admits a  \type0?, \type21, \type2v, \type2*, \type2!, \type*!, \type+1, or a \type+!-$q$-colouring.
\end{corollary}
\begin{proof}
    Note that on cubic graphs  \type0?, \type21, \type2v, \type2*, \type2!, \type*!, \type+1, and  \type+!-$2$-colouring coincide with \type2?-$2$-colouring. Since the graph $H_\varphi$ constructed in the proof of \cref{link2:2?} is cubic, the reduction also shows \NP-completeness of \type0?, \type21, \type2v, \type2*, \type2!, \type*!, \type+1, and \type+!-$2$-colouring. 
\end{proof}

\subsubsection{\Type{2}{2}, \Type{2}{+} and \Type{+}{2}-colouring for \boldmath$q=2$}
First recall that $c:V(G)\rightarrow [q]$ is a \type22-$q$-colouring of $G$ if for every vertex $v$ we have $|N_G(v)\cap c^{-1}(i)|$ is even and larger than $0$ for every $i\in [q]$. 

\begin{theorem}\label{link2:22}
    For $q = 2$ it is \NP-complete to decide whether a graph admits a \type22-$q$-colouring.
\end{theorem}
\begin{proof}
    We provide a reduction from \NSAT. For our reduction the graph $H_\varphi$ for an instance $\varphi$ of \NSAT is regular of degree four. First observe the following.
\begin{claim}\label{claim:222factor}
    A $4$-regular  graph allows for a \type22-colouring if it
    has a $2$-factor such that
    \begin{itemize}
        \item all cycles are induced, that is, without chords, and
        \item contracting these cycles leads to a bipartite graph.
    \end{itemize}
\end{claim}
We start with a component called \emph{bead}. If a bead appears in a
$4$-regular graph $G$ then every \type22-colouring of $G$ restricted to the
bead looks as depicted in Figure~\ref{fig:22-bead}.
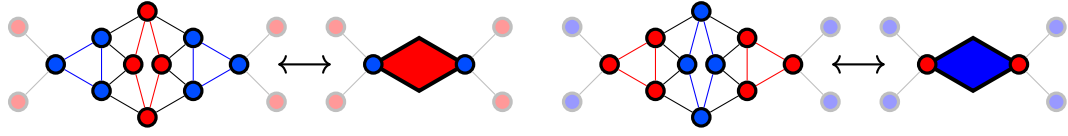
\begin{figure}[hbtp]
  \begin{tikzpicture}[scale=0.7]
    \node[b] (le) at (0,0) {};
    \node[b] (lo) at ($(le)+( 30:1)$) {};
    \node[b] (lu) at ($(le)+(330:1)$) {};
    \draw[blue] (le)--(lo)--(lu)--(le);
    \node[r] (mo) at ($(le)+( 30:2)$) {};
    \node[r] (mu) at ($(le)+(330:2)$) {};
    \node[b] (ro) at ($(mo)+(330:1)$) {};
    \node[b] (ru) at ($(mu)+( 30:1)$) {};
    \node[b] (re) at ($(mu)+( 30:2)$) {};
    \draw[blue] (re)--(ro)--(ru)--(re);
    \node[r] (mr) at ($(le)+(  0:2)$) {};
    \node[r] (ml) at ($(re)+(180:2)$) {};
    \draw[red]  (ml)--(mu)--(mr)--(mo)--(ml);
    \draw (ml)--(lu)--(mu)--(ru)--(mr)--(ro)--(mo)--(lo)--(ml);
    \node[gr] (LO) at ($(le)+(135:1)$) {};
    \node[gr] (LU) at ($(le)+(225:1)$) {};
    \draw[lightgray] (LO)--(le)--(LU);
    \node[gr] (RO) at ($(re)+( 45:1)$) {};
    \node[gr] (RU) at ($(re)+(315:1)$) {};
    \draw[lightgray] (RO)--(re)--(RU);
    \draw[<->,thick] (4.2,0)--(5.2,0);
    \coordinate (sl) at (6,0);
    \coordinate (so) at ($(sl)+( 30:1)$);
    \coordinate (sr) at ($(so)+(330:1)$);
    \coordinate (su) at ($(sl)+(330:1)$);
    \draw[fill=red, line width=1.5pt] (sl)--(so)--(sr)--(su)--cycle;
    \node[b] (sl) at (sl) {};
    \node[b] (sr) at (sr) {};
    \node[gr] (Lo) at ($(sl)+(135:1)$) {};
    \node[gr] (Lu) at ($(sl)+(225:1)$) {};
    \draw[lightgray] (Lo)--(sl)--(Lu);
    \node[gr] (Ro) at ($(sr)+( 45:1)$) {};
    \node[gr] (Ru) at ($(sr)+(315:1)$) {};
    \draw[lightgray] (Ro)--(sr)--(Ru);
  \end{tikzpicture}
  \hfill
  \begin{tikzpicture}[scale=0.7]
    \node[r] (le) at (0,0) {};
    \node[r] (lo) at ($(le)+( 30:1)$) {};
    \node[r] (lu) at ($(le)+(330:1)$) {};
    \draw[red]  (le)--(lo)--(lu)--(le);
    \node[b] (mo) at ($(le)+( 30:2)$) {};
    \node[b] (mu) at ($(le)+(330:2)$) {};
    \node[r] (ro) at ($(mo)+(330:1)$) {};
    \node[r] (ru) at ($(mu)+( 30:1)$) {};
    \node[r] (re) at ($(mu)+( 30:2)$) {};
    \draw[red]  (re)--(ro)--(ru)--(re);
    \node[b] (mr) at ($(le)+(  0:2)$) {};
    \node[b] (ml) at ($(re)+(180:2)$) {};
    \draw[blue] (ml)--(mu)--(mr)--(mo)--(ml);
    \draw (ml)--(lu)--(mu)--(ru)--(mr)--(ro)--(mo)--(lo)--(ml);
    \node[gb] (LO) at ($(le)+(135:1)$) {};
    \node[gb] (LU) at ($(le)+(225:1)$) {};
    \draw[lightgray] (LO)--(le)--(LU);
    \node[gb] (RO) at ($(re)+( 45:1)$) {};
    \node[gb] (RU) at ($(re)+(315:1)$) {};
    \draw[lightgray] (RO)--(re)--(RU);
    \draw[<->,thick] (4.2,0)--(5.2,0);
    \coordinate (sl) at (6,0);
    \coordinate (so) at ($(sl)+( 30:1)$);
    \coordinate (sr) at ($(so)+(330:1)$);
    \coordinate (su) at ($(sl)+(330:1)$);
    \draw[fill=blue, line width=1.5pt] (sl)--(so)--(sr)--(su)--cycle;
    \node[r] (sl) at (sl) {};
    \node[r] (sr) at (sr) {};
    \node[gb] (Lo) at ($(sl)+(135:1)$) {};
    \node[gb] (Lu) at ($(sl)+(225:1)$) {};
    \draw[lightgray] (Lo)--(sl)--(Lu);
    \node[gb] (Ro) at ($(sr)+( 45:1)$) {};
    \node[gb] (Ru) at ($(sr)+(315:1)$) {};
    \draw[lightgray] (Ro)--(sr)--(Ru);
  \end{tikzpicture}
  \caption{The two \type22-colourings of a bead and
    their symbolic representations. Edges of the $2$-factor satisfying \cref{claim:222factor} are coloured blue or red.}
  \label{fig:22-bead}
\end{figure}
A tac representing a variable $x$ that appears $k$ times consists of $k$ beads
that are circularly linked by single edges. For consecutive beads the last and first unused outlet vertex of the bead is used as an outlet of an occurrence of $x$ in the tac. The two
possible \type22-colourings are shown in Figure~\ref{fig:22tac}.
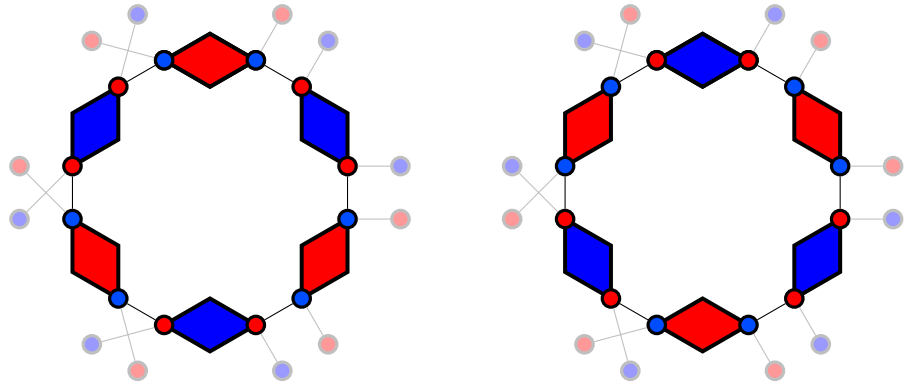
\begin{figure}[hbtp]
  \hspace*{\fill}
  \begin{tikzpicture}[scale=0.7]
    \foreach \n in {1,...,7} {
      \pgfmathtruncatemacro{\a}{ 30+60*\n}
      \pgfmathtruncatemacro{\b}{ 90+60*\n}
      \pgfmathtruncatemacro{\c}{-30+60*\n}
      \coordinate (i\n) at (\a:2);
      \coordinate (a\n) at ($(i\n)+(\a:1)$) {};;
      \coordinate (l\n) at ($(i\n)+(\b:1)$);
      \coordinate (r\n) at ($(i\n)+(\c:1)$);
      \ifodd\n\draw[fill=red, line width=1.5pt] (l\n)--(a\n)--(r\n)--(i\n)--cycle;
      \else\draw[fill=blue, line width=1.5pt] (l\n)--(a\n)--(r\n)--(i\n)--cycle;\fi
      \node[\ifodd\n b\else r\fi] (ll\n) at (l\n) {};
      \node[\ifodd\n b\else r\fi] (rr\n) at (r\n) {};
    }
    \foreach \n in {1,...,6} {
      \pgfmathtruncatemacro{\m}{1+\n}
      \pgfmathtruncatemacro{\a}{60*\m}
      \draw (ll\n)--(rr\m);
      \ifnum\n<4%
        \ifodd\n\node[gr] (L\n) at ($(r\m)+(\a:1)$) {};
          \node[gb] (R\m) at ($(l\n)+(\a:1)$) {};
        \else\node[gb] (L\n) at ($(r\m)+(\a:1)$) {};
          \node[gr] (R\m) at ($(l\n)+(\a:1)$) {};
        \fi
      \else%
        \ifodd\n\node[gr] (L\n) at ($(l\n)+(\a:1)$) {};
          \node[gb] (R\m) at ($(r\m)+(\a:1)$) {};
        \else\node[gb] (L\n) at ($(l\n)+(\a:1)$) {};
          \node[gr] (R\m) at ($(r\m)+(\a:1)$) {};
        \fi
      \fi
      \draw[lightgray] (L\n)--(ll\n)  (R\m)--(rr\m);
    }
  \end{tikzpicture}
  \hspace*{\fill}
  \begin{tikzpicture}[scale=0.7]
    \foreach \n in {1,...,7} {
      \pgfmathtruncatemacro{\a}{ 30+60*\n}
      \pgfmathtruncatemacro{\b}{ 90+60*\n}
      \pgfmathtruncatemacro{\c}{-30+60*\n}
      \coordinate (i\n) at (\a:2);
      \coordinate (a\n) at ($(i\n)+(\a:1)$) {};;
      \coordinate (l\n) at ($(i\n)+(\b:1)$);
      \coordinate (r\n) at ($(i\n)+(\c:1)$);
      \ifodd\n\draw[fill=blue, line width=1.5pt] (l\n)--(a\n)--(r\n)--(i\n)--cycle;
      \else\draw[fill=red, line width=1.5pt] (l\n)--(a\n)--(r\n)--(i\n)--cycle;\fi
      \node[\ifodd\n r\else b\fi] (ll\n) at (l\n) {};
      \node[\ifodd\n r\else b\fi] (rr\n) at (r\n) {};
    }
    \foreach \n in {1,...,6} {
      \pgfmathtruncatemacro{\m}{1+\n}
      \pgfmathtruncatemacro{\a}{60*\m}
      \draw (ll\n)--(rr\m);
      \ifnum\n<4%
        \ifodd\n\node[gb] (L\n) at ($(r\m)+(\a:1)$) {};
          \node[gr] (R\m) at ($(l\n)+(\a:1)$) {};
        \else\node[gr] (L\n) at ($(r\m)+(\a:1)$) {};
          \node[gb] (R\m) at ($(l\n)+(\a:1)$) {};
        \fi
      \else%
        \ifodd\n\node[gb] (L\n) at ($(l\n)+(\a:1)$) {};
          \node[gr] (R\m) at ($(r\m)+(\a:1)$) {};
        \else\node[gr] (L\n) at ($(l\n)+(\a:1)$) {};
          \node[gb] (R\m) at ($(r\m)+(\a:1)$) {};
        \fi
      \fi
        \draw[lightgray] (L\n)--(ll\n)  (R\m)--(rr\m);
    }
  \end{tikzpicture}
  \hspace*{\fill}
  \caption{The two \type22-colourings of a tac appearing in six clauses,
    trice negatively and trice positively.}
    \label{fig:22tac}
\end{figure}
To be able to encode true or false at occurrences of $x$ or $\overline{x}$ we define one of the two vertices of each outlet to be the first vertex of that outlet. For every occurrence of $x$ the first vertex of the outlet is the last unused outlet vertex of the previous bead (when looking at the cycle of beads in clockwise order) while for every occurrence of $\overline{x}$ the first vertex of the outlet is the first unused outlet vertex of the current bead. True is encoded by the first vertex of the outlet having colour $1$ while false is encoded by the first vertex of the outlet having colour $2$. We obtain that property \ref{property:tac} is satisfied. 

\begin{figure}[hbtp]
  \begin{tikzpicture}[xscale=0.5, yscale=0.3] 
    \node[b] (44) at (4,4) {};
    \node[r] (54) at (5,4) {};
    \foreach[count=\n] \y in {1,3,5,7} {
      \node[\ifodd\n b\else r\fi] (3\y) at (3,\y) {}; 
      \node[\ifodd\n r\else b\fi] (6\y) at (6,\y) {}; 
    }
    \foreach \l/\y/\r in {b/2/r, r/4/b, r/6/b} {
      \node[\l]  (2\y) at (2,\y) {};
      \node[g\r] (1\y) at (1,\y) {}; \draw[lightgray] (1\y)--(2\y);
      \node[\r]  (7\y) at (7,\y) {};
      \node[g\l] (8\y) at (8,\y) {}; \draw[lightgray] (8\y)--(7\y);
    }
    \draw[red]  (24)--(33)--(26)--(37)--(24)  (72)--(61)--(54)--(65)--(72);
    \draw[blue] (22)--(31)--(44)--(35)--(22)  (74)--(63)--(76)--(67)--(74);
    \draw (22)--(33)--(44)--(37)--(67)--(54)--(63)--(72);
    \draw (24)--(31)--(61)--(74)  (26)--(35)--(65)--(76);
  \end{tikzpicture}
  \hfill
  \begin{tikzpicture}[xscale=0.5, yscale=0.3] 
    \node[b] (44) at (4,4) {};
    \node[r] (54) at (5,4) {};
    \foreach[count=\n] \y in {1,3,5,7} {
      \pgfmathtruncatemacro{\i}{\n/2}
      \node[\ifodd\i r\else b\fi] (3\y) at (3,\y) {}; 
      \node[\ifodd\i b\else r\fi] (6\y) at (6,\y) {}; 
    }
    \foreach \l/\y/\r in {r/2/b, b/4/r, r/6/b} {
      \node[\l]  (2\y) at (2,\y) {};
      \node[g\r] (1\y) at (1,\y) {}; \draw[lightgray] (1\y)--(2\y);
      \node[\r]  (7\y) at (7,\y) {};
      \node[g\l] (8\y) at (8,\y) {}; \draw[lightgray] (8\y)--(7\y);
    }
    \draw[red]  (22)--(33)--(26)--(35)--(22)  (74)--(61)--(54)--(67)--(74);
    \draw[blue] (24)--(31)--(44)--(37)--(24)  (72)--(63)--(76)--(65)--(72);
    \draw (24)--(33)--(44)--(35)--(65)--(54)--(63)--(74);
    \draw (22)--(31)--(61)--(72)  (26)--(37)--(67)--(76);
  \end{tikzpicture}
  \hfill
  \begin{tikzpicture}[xscale=0.5, yscale=0.3] 
    \node[b] (44) at (4,4) {};
    \node[r] (54) at (5,4) {};
    \foreach[count=\n] \y in {1,3,5,7} {
      \pgfmathtruncatemacro{\i}{(\n+1)/2}
      \node[\ifodd\i r\else b\fi] (3\y) at (3,\y) {}; 
      \node[\ifodd\i b\else r\fi] (6\y) at (6,\y) {}; 
    }
    \foreach \l/\y/\r in {r/2/b, r/4/b, b/6/r} {
      \node[\l]  (2\y) at (2,\y) {};
      \node[g\r] (1\y) at (1,\y) {}; \draw[lightgray] (1\y)--(2\y);
      \node[\r]  (7\y) at (7,\y) {};
      \node[g\l] (8\y) at (8,\y) {}; \draw[lightgray] (8\y)--(7\y);
    }
    \draw[red]  (22)--(33)--(24)--(31)--(22)  (76)--(65)--(54)--(67)--(76);
    \draw[blue] (26)--(35)--(44)--(37)--(26)  (72)--(63)--(74)--(61)--(72);
    \draw (26)--(33)--(44)--(31)--(61)--(54)--(63)--(76);
    \draw (22)--(35)--(65)--(72)  (24)--(37)--(67)--(74);
  \end{tikzpicture}

  \quad
  
  \begin{tikzpicture}[xscale=0.5, yscale=0.3] 
    \node[r] (44) at (4,4) {};
    \node[b] (54) at (5,4) {};
    \foreach[count=\n] \y in {1,3,5,7} {
      \node[\ifodd\n r\else b\fi] (3\y) at (3,\y) {}; 
      \node[\ifodd\n b\else r\fi] (6\y) at (6,\y) {}; 
    }
    \foreach \l/\y/\r in {r/2/b, b/4/r, b/6/r} {
      \node[\l]  (2\y) at (2,\y) {};
      \node[g\r] (1\y) at (1,\y) {}; \draw[lightgray] (1\y)--(2\y);
      \node[\r]  (7\y) at (7,\y) {};
      \node[g\l] (8\y) at (8,\y) {}; \draw[lightgray] (8\y)--(7\y);
    }
    \draw[blue] (24)--(33)--(26)--(37)--(24)  (72)--(61)--(54)--(65)--(72);
    \draw[red]  (22)--(31)--(44)--(35)--(22)  (74)--(63)--(76)--(67)--(74);
    \draw (22)--(33)--(44)--(37)--(67)--(54)--(63)--(72);
    \draw (24)--(31)--(61)--(74)  (26)--(35)--(65)--(76);
  \end{tikzpicture}
  \hfill
  \begin{tikzpicture}[xscale=0.5, yscale=0.3] 
    \node[r] (44) at (4,4) {};
    \node[b] (54) at (5,4) {};
    \foreach[count=\n] \y in {1,3,5,7} {
      \pgfmathtruncatemacro{\i}{\n/2}
      \node[\ifodd\i b\else r\fi] (3\y) at (3,\y) {}; 
      \node[\ifodd\i r\else b\fi] (6\y) at (6,\y) {}; 
    }
    \foreach \l/\y/\r in {b/2/r, r/4/b, b/6/r} {
      \node[\l]  (2\y) at (2,\y) {};
      \node[g\r] (1\y) at (1,\y) {}; \draw[lightgray] (1\y)--(2\y);
      \node[\r]  (7\y) at (7,\y) {};
      \node[g\l] (8\y) at (8,\y) {}; \draw[lightgray] (8\y)--(7\y);
    }
    \draw[blue] (22)--(33)--(26)--(35)--(22)  (74)--(61)--(54)--(67)--(74);
    \draw[red]  (24)--(31)--(44)--(37)--(24)  (72)--(63)--(76)--(65)--(72);
    \draw (24)--(33)--(44)--(35)--(65)--(54)--(63)--(74);
    \draw (22)--(31)--(61)--(72)  (26)--(37)--(67)--(76);
  \end{tikzpicture}
  \hfill
  \begin{tikzpicture}[xscale=0.5, yscale=0.3] 
    \node[r] (44) at (4,4) {};
    \node[b] (54) at (5,4) {};
    \foreach[count=\n] \y in {1,3,5,7} {
      \pgfmathtruncatemacro{\i}{(\n+1)/2}
      \node[\ifodd\i b\else r\fi] (3\y) at (3,\y) {}; 
      \node[\ifodd\i r\else b\fi] (6\y) at (6,\y) {}; 
    }
    \foreach \l/\y/\r in {b/2/r, b/4/r, r/6/b} {
      \node[\l]  (2\y) at (2,\y) {};
      \node[g\r] (1\y) at (1,\y) {}; \draw[lightgray] (1\y)--(2\y);
      \node[\r]  (7\y) at (7,\y) {};
      \node[g\l] (8\y) at (8,\y) {}; \draw[lightgray] (8\y)--(7\y);
    }
    \draw[blue] (22)--(33)--(24)--(31)--(22)  (76)--(65)--(54)--(67)--(76);
    \draw[red]  (26)--(35)--(44)--(37)--(26)  (72)--(63)--(74)--(61)--(72);
    \draw (26)--(33)--(44)--(31)--(61)--(54)--(63)--(76);
    \draw (22)--(35)--(65)--(72)  (24)--(37)--(67)--(74);
  \end{tikzpicture}
  \caption{Six different colourings of an stc}
  \label{fig:22stc}
\end{figure}
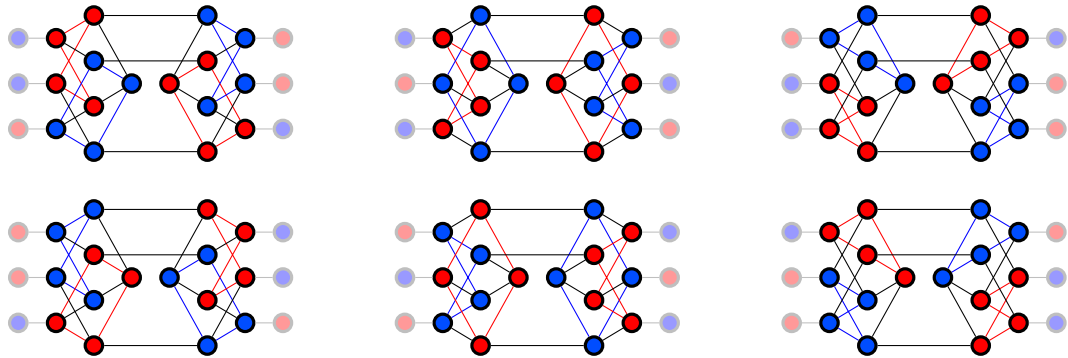

The stc of a clause $c$ consists of two bipartite graphs, $A^c$ and $B^c$, that are linked by some edged. Additionally, there are three outlet vertices attached to $A^c$ and three outlet vertices to $B^c$.  Each outlet of a literal of $c$ consists of one outlet attached to $A^c$ and one outlet attached to $B^c$. We encode true by the outlet attached to $A^c$ having colour $1$, its neighbour colour $2$, the outlet attached to $B^c$ having colour $2$ and its neighbour colour $1$. We encode false by switching the colours at the closed neighbourhood of the outlets. 
Note that in $A^c$ the three vertices to which outlet vertices are attached have a common neighbour $a^c$.  By \cref{claim:222factor} this implies that the three vertices of $A^c$ to which outlets are attached cannot all receive the same colour. Hence, property \ref{property:stc} is satisfied. The 6 valid colourings of the stc are depicted in \cref{fig:22stc}. We finally observe that the tac and stc are compatible and by glueing the first outlet of a literal in the tac to the neighbour of the outlet attached at $A^c$ we obtain a consistent encoding of truth values as stated in property \ref{property:combiningSTCandTAC}. We conclude that by \cref{lem:reductionFromNAE3SAT} \type22-$2$-colouring is \NP-complete.
\end{proof}

\begin{corollary}\label{link2:+2}\label{link2:2+}
    For $q=2$ it is \NP-complete to decide whether a graph admits a  \type2+ or a \type+2-$q$-colouring.
\end{corollary}
\begin{proof}
    Note that on $4$-regular graphs  \type2+ and \type+2-$2$-colouring coincide with \type22-$2$-colouring. Since the graph $H_\varphi$ constructed in the proof of \cref{link2:22} is regular of degree $4$, the reduction also shows \NP-completeness of \type2+ and \type+2-$2$-colouring. 
\end{proof}

\subsubsection{\Type{1}{2}, \Type{v}{2}, \Type{*}{2}, \Type{?}{2}, \Type{!}{2}, \Type{?}{0}, \Type{!}{0}, \Type{1}{+} and \Type{!}{+}-colouring for \boldmath$q=2$}

First recall that $c:V(G)\rightarrow [q]$ is a \type12-$q$-colouring of $G$ if for every vertex $v$ we have $|N_G(v)\cap c^{-1}(c(v))|$ is odd and $|N_G(v)\cap c^{-1}(i)|$ is even and larger than $0$ for every $i\not=c(v)$. 

\begin{theorem}\label{thm:212}\label{link2:12}
    For $q = 2$ it is \NP-complete to decide whether a graph admits a \type12-$q$-colouring.
\end{theorem}
\begin{proof}
    We give a reduction from \NSAT. The reduction graph $H_\varphi$ for an input formula $\phi$ of \NSAT will be cubic. We observe:
    \begin{claim}\label{claim:12cubic}
        In any \type12-$2$-colouring of a cubic graph each colour class induces a disjoint union of $K_2$'s.
    \end{claim}
    \begin{claimproof}
        Presume that $u$ is an arbitrary vertex of a cubic graph $G$ with \type12-colouring $c:V(G)\rightarrow[2]$. Let $v$ be a neighbour of $u$ such that $c(v)=c(u)$ which exists as $u$ must have an odd number of neighbours with the same colour. Since $u$ must have an even number and at least one neighbour of colour $i\neq c(u)$, the two remaining neighbours $u_1,u_2$ of $u$ must have colour $i\not=c(u)$. Similarly, the two remaining neighbours $v_1,v_2$ of $v$ must have colour $i\neq c(v)=c(u)$. Hence the connected component containing $u$ in the graph induced by the colour class containing $u$ is a $K_2$ proving the statement.
    \end{claimproof}

The tac of a variable $x$ is built from beads. The bead depicted in Figure~\ref{fig:12bead} has only the two \type12-colourings by \cref{claim:12cubic}.
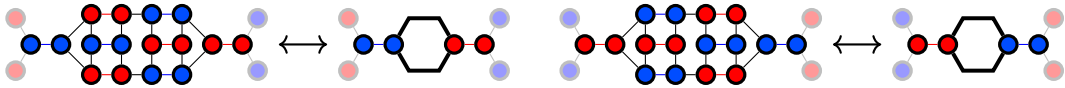
\begin{figure}[hbtp]
  \begin{tikzpicture}[scale=0.4] 
    \foreach \x in {1,2,3,4} \node[b] (\x2) at (\x,2) {};
    \foreach \x in {5,6,7,8} \node[r] (\x2) at (\x,2) {};
    \foreach \x in {3,4} \foreach \y in {1,3} \node[r] (\x\y) at (\x,\y) {};
    \foreach \x in {5,6} \foreach \y in {1,3} \node[b] (\x\y) at (\x,\y) {};
    \draw[blue] (12)--(22)  (32)--(42)  (51)--(61)  (53)--(63);
    \draw[red]  (31)--(41)  (33)--(43)  (52)--(62)  (72)--(82);
    \draw (22)--(33)--(32)--(31)--(22)  (72)--(63)--(62)--(61)--(72);
    \draw (41)--(42)--(43)--(53)--(52)--(51)--(41);
    \node[gr] (01) at ($(12)+(240:1)$) {};
    \node[gr] (03) at ($(12)+(120:1)$) {};
    \draw[lightgray] (01)--(12)--(03);
    \node[gb] (91) at ($(82)+( 60:1)$) {};
    \node[gb] (93) at ($(82)+(300:1)$) {};
    \draw[lightgray] (91)--(82)--(93);
    \draw[<->,thick] (9.2,2)--(10.8,2);
    \foreach \x in {12,13} \node[b] (\x2) at (\x,2) {};
    \draw[blue] (122)--(132);              
    \foreach \x in {15,16} \node[r] (\x2) at (\x,2) {};
    \draw[red]  (152)--(162);              
    \draw[line width=1.5pt] (132)--($(132)+( 60:1)$)--($(152)+(120:1)$)--(152);
    \draw[line width=1.5pt] (132)--($(132)+(300:1)$)--($(152)+(240:1)$)--(152);
    \node[gr] (111) at ($(122)+(240:1)$) {};
    \node[gr] (113) at ($(122)+(120:1)$) {};
    \draw[lightgray] (111)--(122)--(113);
    \node[gb] (171) at ($(162)+( 60:1)$) {};
    \node[gb] (173) at ($(162)+(300:1)$) {};
    \draw[lightgray] (171)--(162)--(173);
  \end{tikzpicture}
  \hfill
  \begin{tikzpicture}[scale=0.4] 
    \foreach \x in {1,2,3,4} \node[r] (\x2) at (\x,2) {};
    \foreach \x in {5,6,7,8} \node[b] (\x2) at (\x,2) {};
    \foreach \x in {3,4} \foreach \y in {1,3} \node[b] (\x\y) at (\x,\y) {};
    \foreach \x in {5,6} \foreach \y in {1,3} \node[r] (\x\y) at (\x,\y) {};
    \draw[red]  (12)--(22)  (32)--(42)  (51)--(61)  (53)--(63);
    \draw[blue] (31)--(41)  (33)--(43)  (52)--(62)  (72)--(82);
    \draw (22)--(33)--(32)--(31)--(22)  (72)--(63)--(62)--(61)--(72);
    \draw (41)--(42)--(43)--(53)--(52)--(51)--(41);
    \node[gb] (01) at ($(12)+(240:1)$) {};
    \node[gb] (03) at ($(12)+(120:1)$) {};
    \draw[lightgray] (01)--(12)--(03);
    \node[gr] (91) at ($(82)+( 60:1)$) {};
    \node[gr] (93) at ($(82)+(300:1)$) {};
    \draw[lightgray] (91)--(82)--(93);
    \draw[<->,thick] (9.2,2)--(10.8,2);
    \foreach \x in {12,13} \node[r] (\x2) at (\x,2) {};
    \draw[red]  (122)--(132);              
    \foreach \x in {15,16} \node[b] (\x2) at (\x,2) {};
    \draw[blue] (152)--(162);              
    \draw[line width=1.5pt] (132)--($(132)+( 60:1)$)--($(152)+(120:1)$)--(152);
    \draw[line width=1.5pt] (132)--($(132)+(300:1)$)--($(152)+(240:1)$)--(152);
    \node[gb] (111) at ($(122)+(240:1)$) {};
    \node[gb] (113) at ($(122)+(120:1)$) {};
    \draw[lightgray] (111)--(122)--(113);
    \node[gr] (171) at ($(162)+( 60:1)$) {};
    \node[gr] (173) at ($(162)+(300:1)$) {};
    \draw[lightgray] (171)--(162)--(173);
  \end{tikzpicture}
  \caption{The two \type12-colourings of a bead and
    their symbolic representations}
  \label{fig:12bead}
\end{figure}
For each variable $x$ we build a tac out of cyclically connected beads. More
precisely, we use one bead for each occurrence of the variable $x$. For consecutive beads the last and first unused outlet vertex of the bead is used as an outlet of an occurrence of $x$ in the tac. The two
possible \type12-colourings are shown in Figure~\ref{fig:12tac}.
\begin{figure}[hbtp]
  \hspace*{\fill}
  \begin{tikzpicture}[scale=0.7]
    \node[a] (1) at (102:3) {};
    \foreach \n in {2,3,...,31} {
      \pgfmathtruncatemacro{\a}{ 90 + 12*\n}
      \pgfmathtruncatemacro{\b}{102 + 12*\n}
      \pgfmathtruncatemacro{\m}{\n-1}
      \pgfmathtruncatemacro{\z}{ 90 + 12*\m}
      \pgfmathtruncatemacro{\q}{\n / 5}
      \pgfmathtruncatemacro{\r}{\n - 5*\q}
      \ifcase\r  \coordinate (\n) at (\a:3);   
      \or         \node[b]   (\n) at (\a:3) {};
        \pgfmathtruncatemacro{\l}{\n-2}
        \pgfmathtruncatemacro{\d}{\z - 30}
        \pgfmathtruncatemacro{\e}{\z + 30}
        \draw[line width=1.5pt] (\l)--($(\m)+(\d:0.5)$)--($(\m)+(\e:0.5)$)--(\n);
        \draw[line width=1.5pt] (\l)--($(\m)+(\e:-.6)$)--($(\m)+(\d:-.6)$)--(\n);
      \or         \node[b]   (\n) at (\a:3) {}; 
        \ifnum\q<3\node[gr] (g\n) at (\b:4) {}; \draw[lightgray] (g\n)--(\n);
        \else     \node[gr] (g\n) at (\a:4) {}; \draw[lightgray] (g\n)--(\n);
        \fi
        \draw[blue] (\n)--(\m);
      \or         \node[r]   (\n) at (\a:3) {}; 
        \ifnum\q<3\node[gb] (g\n) at (\z:4) {}; \draw[lightgray] (g\n)--(\n);
        \else     \node[gb] (g\n) at (\a:4) {}; \draw[lightgray] (g\n)--(\n);
        \fi
        \draw (\n)--(\m);
      \or         \node[r]   (\n) at (\a:3) {}; 
        \draw[red]  (\n)--(\m);
      \fi
    }
  \end{tikzpicture}
  \hspace*{\fill}
  \begin{tikzpicture}[scale=0.7]
    \node[a] (1) at (102:3) {};
    \foreach \n in {2,3,...,31} {
      \pgfmathtruncatemacro{\a}{ 90 + 12*\n}
      \pgfmathtruncatemacro{\b}{102 + 12*\n}
      \pgfmathtruncatemacro{\m}{\n-1}
      \pgfmathtruncatemacro{\z}{ 90 + 12*\m}
      \pgfmathtruncatemacro{\q}{\n / 5}
      \pgfmathtruncatemacro{\r}{\n - 5*\q}
      \ifcase\r  \coordinate (\n) at (\a:3);   
      \or         \node[r]   (\n) at (\a:3) {}; 
        \pgfmathtruncatemacro{\l}{\n-2}
        \pgfmathtruncatemacro{\d}{\z - 30}
        \pgfmathtruncatemacro{\e}{\z + 30}
        \draw[line width=1.5pt] (\l)--($(\m)+(\d:0.5)$)--($(\m)+(\e:0.5)$)--(\n);
        \draw[line width=1.5pt] (\l)--($(\m)+(\e:-.6)$)--($(\m)+(\d:-.6)$)--(\n);
      \or         \node[r]   (\n) at (\a:3) {}; 
        \ifnum\q<3\node[gb] (g\n) at (\b:4) {}; \draw[lightgray] (g\n)--(\n);
        \else     \node[gb] (g\n) at (\a:4) {}; \draw[lightgray] (g\n)--(\n);
        \fi
        \draw[red] (\n)--(\m);
      \or         \node[b]   (\n) at (\a:3) {}; 
        \ifnum\q<3\node[gr] (g\n) at (\z:4) {}; \draw[lightgray] (g\n)--(\n);
        \else     \node[gr] (g\n) at (\a:4) {}; \draw[lightgray] (g\n)--(\n);
        \fi
        \draw (\n)--(\m);
      \or         \node[b]   (\n) at (\a:3) {}; 
        \draw[blue] (\n)--(\m);
      \fi
    }
  \end{tikzpicture}
  \hspace*{\fill}
  \caption{Two \type12-colourings of a tac for a variable occurring trice negatively and trice positively.}
  \label{fig:12tac}
\end{figure}
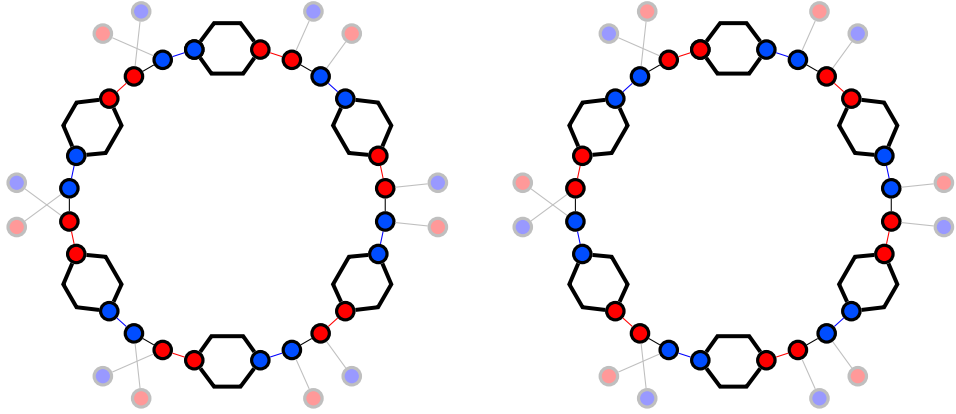
In order to encode true or false at occurrences of $x$ or $\overline{x}$ we make the convention that one of the two vertices of each outlet is the first vertex of that outlet. For every occurrence of $x$ the first vertex of the outlet is the last unused outlet vertex of the previous bead (when looking at the cycle of beads in clockwise order) while for every occurrence of $\overline{x}$ the first vertex of the outlet is the first unused outlet vertex of the current bead. True is encoded by the first vertex of the outlet having colour $1$ while false is encoded by the first vertex of the outlet having colour $2$. With this convention, property \ref{property:tac} is satisfied.

The stc of a clause $c$ is depicted in \cref{fig:12stc}. 
\begin{figure}[hbtp]
  \hspace*{\fill}
  \begin{tikzpicture}[scale=0.5]
    \node[b] (a) at (  0:3) {}; \node[r] (A) at ( 20:4) {};
    \node[r] (b) at ( 40:3) {}; \node[b] (B) at ( 40:4) {};
    \node[r] (c) at ( 60:3) {}; \node[b] (C) at ( 60:4) {};
    \node[b] (d) at ( 80:3) {}; \node[r] (D) at ( 80:4) {};
    \node[b] (e) at (100:3) {}; \node[r] (E) at (100:4) {};
    \node[r] (f) at (140:3) {}; \node[b] (F) at (120:4) {};
    \node[r] (g) at (160:3) {}; \node[b] (G) at (160:4) {};
    \node[b] (h) at (180:3) {}; \node[r] (H) at (180:4) {};
    \node[b] (i) at (200:3) {}; \node[r] (I) at (200:4) {};
    \node[r] (j) at (220:3) {}; \node[b] (J) at (220:4) {};
    \node[r] (k) at (260:3) {}; \node[b] (K) at (240:4) {};
    \node[b] (l) at (280:3) {}; \node[r] (L) at (280:4) {};
    \node[b] (m) at (300:3) {}; \node[r] (M) at (300:4) {};
    \node[r] (n) at (320:3) {}; \node[b] (N) at (320:4) {};
    \node[r] (p) at (340:3) {}; \node[b] (P) at (340:4) {};
    \node[b] (o) at ($0.3*(a) + 0.3*(f) + 0.4*(k)$) {};
    \node[r] (O) at ($0.3*(A) + 0.4*(F) + 0.3*(K)$) {};
    \draw[blue] (o)--(a)  (B)--(C)  (d)--(e)  (F)--(G)  (h)--(i)  (J)--(K)  (l)--(m)  (N)--(P);
    \draw[red]  (O)--(A)  (b)--(c)  (D)--(E)  (f)--(g)  (H)--(I)  (j)--(k)  (L)--(M)  (n)--(p);
    \draw (a)--(b)--(B)--(A)--(P)--(p)--(a)  (c)--(d)--(D)--(C);
    \draw (o)--(k)--(l)--(L)--(K)--(O)--(F)--(E)--(e)--(f)--(o);
    \draw (i)--(j)--(J)--(I)--(i)  (H)--(G)--(g)--(h);
    \draw (c)--(d)--(D)--(C)  (m)--(n)--(N)--(M);
    \node[gb] (x) at ( 60:2) {}; \node[gr] (X) at ( 60:5) {};
    \draw[lightgray] (x)--(c)  (C)--(X);
    \node[gr] (y) at (180:2) {}; \node[gb] (Y) at (180:5) {};
    \draw[lightgray] (y)--(h)  (H)--(Y);
    \node[gr] (z) at (300:2) {}; \node[gb] (Z) at (300:5) {};
    \draw[lightgray] (z)--(m)  (M)--(Z);
  \end{tikzpicture}
  \hspace*{\fill}
  \caption{One of six \type12-colourings of the stc}
  \label{fig:12stc}
\end{figure}
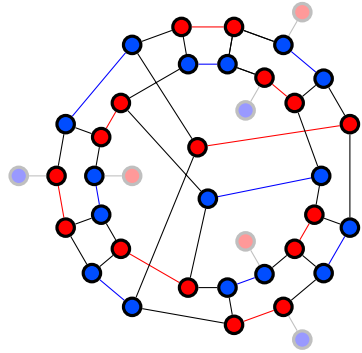
We encode true by the outer vertex of an outlet being colour $1$, its neighbour colour $2$, the inner vertex of the outlet being colour $2$ and its neighbour being colour $1$. Let $\ell$ and $\ell'$ be two arbitrary literals of $c$. Let $u_\ell$ be the neighbour of the outlet of $\ell$ on the outer cycle and $v_\ell$ be the neighbour of the outlet of $\ell$ on the inner cycle. Define $u_{\ell'}, v_{\ell'}$ in the same way. Furthermore, let $(u_\ell,u_1,\dots,u_4, u_{\ell'})$ be the path along the outer cycle from $u_\ell$ to $u_{\ell'}$ (that does not contain the neighbour of the outlet of the third literal of $c$) and, similarly, let $(v_\ell,v_1,\dots,v_4, v_{\ell'})$ be the path along the outer cycle from $v_\ell$ to $v_{\ell'}$ (see \cref{fig:12stc}). Without loss of generality, assume $u_iv_i$ is an edge in the stc for all $i\neq 2$ (this must be the case for $i\neq 2$ or $i\neq 3$). 
\begin{claim}\label{claim:12stc}
    Given in a valid \type12-$2$-colouring $c$ of the stc it hold that $c(u_\ell)=c(u_{\ell'})$, then $c(u_\ell)=c(u_1)=c(u_4)=c(u_{\ell'})=c(v_2)=c(v_3)$ and $c(v_\ell)=c(v_1)=c(v_4)=c(v_{\ell'})=c(u_2)=c(u_3)\neq c(u_\ell)$.
\end{claim}
\begin{claimproof}
    Since $c$ is a valid colouring, we know by assumption of the encoding of true and false that $c({u_\ell})\neq c(v_\ell)$ and $c(u_{\ell'})\neq c(v_{\ell'})$.
    Now, observe that $c(u_1)\neq c(v_1)$ as otherwise either $u_1$ or $u_2$ must have two neighbour of its own colour contradicting \cref{claim:12cubic}. Similarly, $c(u_4)\neq c(v_4)$ which implies, using the same argument, that $c(u_3)\neq c(v_3)$.
    
    We now argue that $c(u_3)\neq c(u_{\ell'})$. For this consider two cases. If $c(u_{\ell'})=c(u_4)$, then $c(u_3)\neq c(u_4)$ by \cref{claim:12cubic}. On the other hand, if $c(u_{\ell'})\neq c(u_4)$, then $c(u_3)=c(u_4)$ by \cref{claim:12cubic} as $u_4$'s other two neighbours $u_{\ell'}$ and $v_4$ do not have the same colour as $u_4$. 

    Next, we argue that $c(u_1)\neq c(u_3)$. For this assume that $c(u_2)\neq c(u_3)$. Since $c(u_3)\neq c(u_{\ell'})=c(u_\ell)$ we further get that $c(u_2)=c(u_\ell)$. As $u_1$ is adjacent to $u_2$ and $u_\ell$ this implies by \cref{claim:12cubic} that $c(u_1)\neq c(u_2)$ and $c(v_1)=c(u_1)$. Using that $c(v_\ell)\neq c(u_\ell)$ and we contradict \cref{claim:12cubic} as $v_1$ has two neighbours of the same colour. Therefore $c(u_2)=c(u_3)$ which implies, using \cref{claim:12cubic} that $c(u_1)\neq c(u_3)$. 

    Using that $c(u_\ell)=c(u_1)=c(u_{\ell'})=c(v_3)$ and $c(v_\ell)=c(v_1)=c(v_{\ell'})=c(u_3)\neq c(u_\ell)$ it follows that the remaining colours have to be assigned as claimed by \cref{claim:12cubic}.
\end{claimproof}
Now, presume that $c=\ell_1\lor \ell_2 \lor \ell_3$ and $u_{\ell_i}$ is the vertex adjacent to an outlet vertex of $\ell_i$ on the outer cycle of the stc. Let $c$ be a valid \type12-$2$-colouring of the stc with $c(u_{\ell_1})=c(u_{\ell_2})=c(u_{\ell_3})$. Presume $u$ and $u'$ are the two vertices on the outer cycle adjacent to $u_{\ell_1}$. By \cref{claim:12stc} $c(u)=c(u_{\ell_1})=c(u')$ which contradicts \cref{claim:12cubic}. Hence, we cannot have $c(u_{\ell_1})=c(u_{\ell_2})=c(u_{\ell_3})$ implying property \ref{property:stc} is satisfied.   
Note that all six \type12-colourings of the stc can be obtained from the one shown in
Figure~\ref{fig:12stc} by rotation and swap of colours. Finally, we observe that the tac and stc are compatible and that the appropriate colourings are compatiple implying property \ref{property:combiningSTCandTAC}. Hence, \type12-$2$-colouring is \NP-complete.
\end{proof}

\begin{corollary}\label{link2:v2}\label{link2:*2}\label{link2:?2}\label{link2:!2}\label{link2:?0}\label{link2:!0}\label{link2:1+}\label{link2:!+}
    For $q=2$ it is \NP-complete to decide whether a graph admits a \type{v}2, \type*2, \type?2, \type!2, \type?0, \type!0, \type1+, or  \type!+-$q$-colouring.
\end{corollary}
\begin{proof}
    It is easy to verify that on cubic graphs \type{v}2, \type*2,  \type?2, \type!2, \type?0, \type!0, \type1+, and  \type!+-$2$-colouring coincide with \type12-$2$-colouring. Since the graph $H_\varphi$ constructed in the proof of \cref{thm:212} is cubic, the reduction also shows \NP-completeness of \type{v}2, \type*2, \type?2, \type!2, \type?0, \type!0, \type1+, and  \type!+-$2$-colouring.
\end{proof}

\subsubsection{\Type{!}{1}-colouring for \boldmath$q=2$}

First recall that $c:V(G)\rightarrow [q]$ is a \type!1-$q$-colouring of $G$ if for every vertex $v$ we have $|N_G(v)\cap c^{-1}(c(v))|=1$ and $|N_G(v)\cap c^{-1}(i)|$ is odd for every $i\not=c(v)$. 

\begin{theorem}\label{link2:!1}
    For $q = 2$ it is \NP-complete to decide whether a graph admits a \type!1-$q$-colouring.
\end{theorem}
\begin{proof}
We show here by a reduction
from \NSAT that \type!1-$2$-colouring is \NP-complete for $4$-regular graphs.
\begin{figure}[hbtp]
  \centering
  \begin{tikzpicture}[xscale=0.4, yscale=0.7]
    \node[b, label=left: $u$] (u) at (2,1) {};
    \node[r, label=right:$v$] (v) at (4,1) {};
    \draw (u)--(v);
    \foreach \n/\x/\c in {a/1/b, b/3/r, c/5/r} {
      \node[\c] (\n) at (\x,2) {};
      \draw (u)--(\n)--(v);
    }
    \draw (a)--(b);
    \draw[lightgray] (0.5,2.5)--(a)--(1.5,2.5);
    \draw[lightgray] (2.5,2.5)--(b)--(3.5,2.5);
    \draw[lightgray] (4.5,2.5)--(c)--(5.5,2.5);
  \end{tikzpicture}
  \caption{Two true twins of different colours.}
  \label{fig:!1tt}
\end{figure}
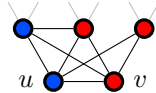
First we observe the following.
\begin{claim}\label{claim:!1trueTwins}
    True twins receive the same colour in every \type!1-colouring with two colours.
\end{claim}
\begin{claimproof}
    Presume this is  not the case and $u,v$ are true twins in a graph $G$ and $c:V(G)\rightarrow [2]$ is a \type!1-colouring of $G$ in which $c(u)\not=c(v)$. In this case, there must be exactly one neighbour $w$ of $u$ with $c(w)=c(u)$. Therefore, $v$ has exactly two neighbours $u$ and $w$ of colour not equal two $c(v)$ (see \cref{fig:!1tt}), contradicting the definition of \type!1-colouring.
\end{claimproof}

We combine copies of the graph shown in Figure~\ref{fig:!1tt} to form the tac for variable $x$ given in
Figure~\ref{fig:!1tac}. Using \cref{claim:!1trueTwins} it is easy to verify that the tac has precisely two \type!1-$2$-colourings. Each outlet consist of a pair of vertices attached to two consecutive copies of the graph depicted in \cref{fig:!1tt}. In order to encode different truth values, for each outlet we declare one vertex to be the first vertex of the outlet. Presume an outlet corresponds to a literal $x$, then we choose the vertex attached to the even copy to be the first vertex, where every second copy of the graph from \cref{fig:!1tt} along the cycle is even (starting from an arbitrarily chosen copy). For outlets corresponding to literal $\overline{x}$ we choose the vertex attached to the even copy to be the first vertex of the outlet. We encode that the literal corresponding to an outlet is true if its first vertex has colour 1 and we encode false if its first vertex has colour 2.
With this encoding the tac satisfies property \ref{property:tac}.

\begin{figure}[hbtp]
  \hspace*{\fill}
  \begin{tikzpicture}[scale=0.7]
    \newcommand{\rb}[1]{\ifodd#1{r}\else{b}\fi}
    \newcommand{\br}[1]{\ifodd#1{b}\else{r}\fi}
    \newcommand{\Br}[1]{\ifodd#1{gb}\else{gr}\fi}
    \newcommand{\RB}[1]{\ifodd#1{red}\else{blue}\fi}
    \newcommand{\BR}[1]{\ifodd#1{blue}\else{red}\fi}
    \node[gr] (G3)  at ( 3.3, 2.7) {}; 
    \node[gb] (G9)  at ( 2.7, 3.3) {}; 
    \node[gb] (G21) at (-2.7, 3.3) {}; 
    \node[gr] (G15) at (-3.3, 2.7) {}; 
    \node[gb] (G33) at (-3.3,-2.7) {}; 
    \node[gr] (G27) at (-2.7,-3.3) {}; 
    \node[gr] (G39) at ( 2.7,-3.3) {}; 
    \node[gb] (G45) at ( 3.3,-2.7) {}; 
    \node[a] (-1) at (-7.5:3.0) {};
    \foreach \n in {1,...,48} {
      \pgfmathtruncatemacro{\z}{\n-1}
      \pgfmathtruncatemacro{\y}{\n-2}
      \pgfmathtruncatemacro{\x}{\n-3}
      \pgfmathsetmacro{\a}{7.5 * \n}
      \pgfmathtruncatemacro{\m}{\n / 6}
      \pgfmathtruncatemacro{\r}{\n - 6*\m}
      \ifcase\r\relax                         
      \or \node[\rb{\m}] (\n) at (\a:3.0) {}; 
          \draw (\n)--(\y);
      \or \node[\br{\m}] (\n) at (\a:2.3) {}; 
          \draw (\n)--(\z);
      \or \node[\rb{\m}] (\n) at (\a:3.0) {}; 
          \draw (\n)--(\z);
          \draw[lightgray] (\n)--(G\n);
      \or \node[\br{\m}] (\n) at (\a:2.3) {}; 
          \draw (\x)--(\n)--(\z);
          \draw[\BR{\m}] (\n)--(\y);
      \or \node[\rb{\m}] (\n) at (\a:3.0) {}; 
          \draw (\x)--(\n)--(\z);
          \draw[\RB{\m}] (\n)--(\y);
      \fi
    }
    \draw[blue] (1)--(25)  (13)--(37);
    \draw[red]  (7)--(31)  (19)--(43);
  \end{tikzpicture}
  \hspace*{\fill}
  \begin{tikzpicture}[scale=0.7]
    \newcommand{\rb}[1]{\ifodd#1{r}\else{b}\fi}
    \newcommand{\br}[1]{\ifodd#1{b}\else{r}\fi}
    \newcommand{\Rb}[1]{\ifodd#1{gr}\else{gb}\fi}
    \newcommand{\RB}[1]{\ifodd#1{red}\else{blue}\fi}
    \newcommand{\BR}[1]{\ifodd#1{blue}\else{red}\fi}
    \node[gb] (G3)  at ( 3.3, 2.7) {}; 
    \node[gr] (G9)  at ( 2.7, 3.3) {}; 
    \node[gr] (G21) at (-2.7, 3.3) {}; 
    \node[gb] (G15) at (-3.3, 2.7) {}; 
    \node[gr] (G33) at (-3.3,-2.7) {}; 
    \node[gb] (G27) at (-2.7,-3.3) {}; 
    \node[gb] (G39) at ( 2.7,-3.3) {}; 
    \node[gr] (G45) at ( 3.3,-2.7) {}; 
    \node[a] (-1) at (-7.5:3.0) {};
    \foreach \n in {1,...,48} {
      \pgfmathtruncatemacro{\z}{\n-1}
      \pgfmathtruncatemacro{\y}{\n-2}
      \pgfmathtruncatemacro{\x}{\n-3}
      \pgfmathsetmacro{\a}{7.5 * \n}
      \pgfmathtruncatemacro{\m}{\n / 6}
      \pgfmathtruncatemacro{\r}{\n - 6*\m}
      \ifcase\r\relax                         
      \or \node[\br{\m}] (\n) at (\a:3.0) {}; 
          \draw (\n)--(\y);
      \or \node[\rb{\m}] (\n) at (\a:2.3) {}; 
          \draw (\n)--(\z);
      \or \node[\br{\m}] (\n) at (\a:3.0) {};
          \draw (\n)--(\z);
          \draw[lightgray] (\n)--(G\n);
      \or \node[\rb{\m}] (\n) at (\a:2.3) {}; 
          \draw (\x)--(\n)--(\z);
          \draw[\RB{\m}] (\n)--(\y);
      \or \node[\br{\m}] (\n) at (\a:3.0) {}; 
          \draw (\x)--(\n)--(\z);
          \draw[\BR{\m}] (\n)--(\y);
      \fi
    }
    \draw[red]  (1)--(25)  (13)--(37);
    \draw[blue] (7)--(31)  (19)--(43);
  \end{tikzpicture}
  \hspace*{\fill}
  \caption{The two possible \type{!}{1}-colourings of the tac
  representing a variable that appears twice negatively and twice positively.
  }
  \label{fig:!1tac}
\end{figure}
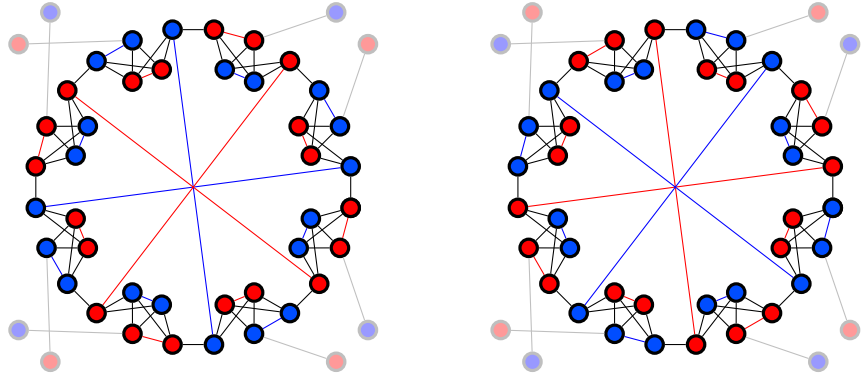

The stc of a clause $c$ depicted in Figure~\ref{fig:!1stc} consists of two copies of $K_4$, called $A^c$ and $B^c$,
linked by a single edge and outlets consisting of pairs of vertices attached. 
\begin{figure}[hbtp]
  \hspace*{\fill}
  \begin{tikzpicture}[scale=0.5]
    \node[r] (lb) at (2,1) {};  \node[b] (lt) at (2,3) {};
    \node[r] (ll) at (1,2) {};  \node[b] (lr) at (3,2) {};
    \draw (lb)--(lt)--(ll)--(lr)--(lb);
    \draw[red]  (lb)--(ll);
    \draw[blue] (lt)--(lr);
    \node[b] (rb) at (6,1) {};  \node[r] (rt) at (6,3) {};
    \node[b] (rl) at (5,2) {};  \node[r] (rr) at (7,2) {};
    \draw (rb)--(rt)--(rl)--(rr)--(rb)--(lb);
    \draw[blue] (rb)--(rl);
    \draw[red]  (rt)--(rr);
    \node[gb] (sl) at (1.0,4) {}; \draw[lightgray] (ll)--(sl);
    \node[gr] (tl) at (2.0,4) {}; \draw[lightgray] (rl)--(tl);
    \node[gr] (sm) at (3.5,4) {}; \draw[lightgray] (lt)--(sm);
    \node[gb] (tm) at (4.5,4) {}; \draw[lightgray] (rt)--(tm);
    \node[gr] (sr) at (6.0,4) {}; \draw[lightgray] (lr)--(sr);
    \node[gb] (tr) at (7.0,4) {}; \draw[lightgray] (rr)--(tr);
  \end{tikzpicture}
  \hspace*{\fill}
  \begin{tikzpicture}[scale=0.5]
    \node[r] (lb) at (2,1) {};  \node[r] (lt) at (2,3) {};
    \node[b] (ll) at (1,2) {};  \node[b] (lr) at (3,2) {};
    \draw (lb)--(ll)--(lt)--(lr)--(lb);
    \draw[red]  (lb)--(lt);
    \draw[blue] (ll)--(lr);
    \node[b] (rb) at (6,1) {};  \node[b] (rt) at (6,3) {};
    \node[r] (rl) at (5,2) {};  \node[r] (rr) at (7,2) {};
    \draw (rb)--(rl)--(rt)--(rr)--(rb)--(lb);
    \draw[blue] (rb)--(rt);
    \draw[red]  (rl)--(rr);
    \node[gr] (sl) at (1.0,4) {}; \draw[lightgray] (ll)--(sl);
    \node[gb] (tl) at (2.0,4) {}; \draw[lightgray] (rl)--(tl);
    \node[gb] (sm) at (3.5,4) {}; \draw[lightgray] (lt)--(sm);
    \node[gr] (tm) at (4.5,4) {}; \draw[lightgray] (rt)--(tm);
    \node[gr] (sr) at (6.0,4) {}; \draw[lightgray] (lr)--(sr);
    \node[gb] (tr) at (7.0,4) {}; \draw[lightgray] (rr)--(tr);
  \end{tikzpicture}
  \hspace*{\fill}
  \begin{tikzpicture}[scale=0.5]
    \node[r] (lb) at (2,1) {};  \node[b] (lt) at (2,3) {};
    \node[b] (ll) at (1,2) {};  \node[r] (lr) at (3,2) {};
    \draw (lb)--(lt)--(lr)--(ll)--(lb);
    \draw[red]  (lb)--(lr);
    \draw[blue] (lt)--(ll);
    \node[b] (rb) at (6,1) {};  \node[r] (rt) at (6,3) {};
    \node[r] (rl) at (5,2) {};  \node[b] (rr) at (7,2) {};
    \draw (rb)--(rt)--(rr)--(rl)--(rb)--(lb);
    \draw[blue] (rb)--(rr);
    \draw[red]  (rt)--(rl);
    \node[gr] (sl) at (1.0,4) {}; \draw[lightgray] (ll)--(sl);
    \node[gb] (tl) at (2.0,4) {}; \draw[lightgray] (rl)--(tl);
    \node[gr] (sm) at (3.5,4) {}; \draw[lightgray] (lt)--(sm);
    \node[gb] (tm) at (4.5,4) {}; \draw[lightgray] (rt)--(tm);
    \node[gb] (sr) at (6.0,4) {}; \draw[lightgray] (lr)--(sr);
    \node[gr] (tr) at (7.0,4) {}; \draw[lightgray] (rr)--(tr);
  \end{tikzpicture}
  \hspace*{\fill}
  \caption{Three different \type{!}{1}-colourings of the stc. Three more can
    be obtained by swapping colours. }
  \label{fig:!1stc}
\end{figure}
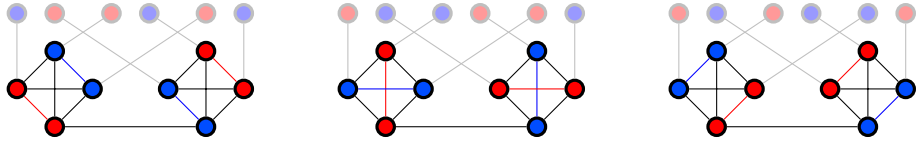
We encode that the colouring at an outlet encodes that the corresponding literal is true if the vertex of the outlet attached to $A^c$ receives colour $1$ and false if the vertex of the outlet attached to $A^c$ receives colour two. 
Note that in any \type!1-$2$-colouring of the stc each $K_4$ must have two vertices of each colour. Hence, the three vertices of the outlets that are attached to $A^c$ cannot all receive the same colour. We obtain that property \ref{property:stc} holds. To combine the tac and stc, we need to glue the first vertex of the outlet of the tac corresponding to literal $\ell$ to the neighbour of the outlet which is attached to $A^c$. This ensures that the correct colouring are compatible and hence property \ref{property:combiningSTCandTAC} holds. We conclude that by \cref{lem:reductionFromNAE3SAT} \type02-$2$-colouring is \NP-complete.
\end{proof}

\subsubsection{\Type{0}{+}-colouring for \boldmath$q=2$}
First recall that $c:V(G)\rightarrow [q]$ is a \type0+-$q$-colouring of $G$ if for every vertex $v$ we have $|N_G(v)\cap c^{-1}(c(v))|$ is even and $|N_G(v)\cap c^{-1}(i)|>0$ for every $i\not=c(v)$. 

\begin{theorem}\label{link2:0+}
    For $q = 2$ it is \NP-complete to decide whether a graph admits a \type0+-$q$-colouring.
\end{theorem}
\begin{proof}
    We provide a reduction from \NSAT.
    First we observe that if a vertex $v$ of $G$ has exactly two neighbours then
    these receive the same colour in every $2$-colouring of $G$ of type \type{0}{+}.
    This leads to the tac of variable $x$ shown in Figure \ref{fig:0+tac}, which is a cycle with outlets consisting of a single vertex attached. 
    \begin{figure}[hbtp]
  \hspace*{\fill}
  \newcommand{\br}[1]{\ifodd#1{b}\else{r}\fi}
  \newcommand{\rb}[1]{\ifodd#1{r}\else{b}\fi}
  \begin{tikzpicture}
    \pgfmathsetmacro{\a}{-360/7}
    \node[a] (1) at (\a:1) {};
    \foreach[count=\m] \n in {2,...,15} {
      \pgfmathsetmacro{\a}{(\n-3)*360/14}
      \node[\rb{\n}] (\n) at (\a:1) {};
      \draw (\n)--(\m);
    }
    \foreach \n in {1,3,5} {
      \pgfmathsetmacro{\a}{(\n-3)*360/14}
      \node[gb] (g\n) at (\a:1.5) {};
      \draw[lightgray] (g\n)--(\n);
    }
    \foreach \n in {8,10,12} {
      \pgfmathsetmacro{\a}{(\n-3)*360/14}
      \node[gr] (g\n) at (\a:1.5) {};
      \draw[lightgray] (g\n)--(\n);
    }
  \end{tikzpicture}
  \hspace*{\fill}
  \begin{tikzpicture}
    \pgfmathsetmacro{\a}{-360/7}
    \node[a] (1) at (\a:1) {};
    \foreach[count=\m] \n in {2,...,15} {
      \pgfmathsetmacro{\a}{(\n-3)*360/14}
      \node[\br{\n}] (\n) at (\a:1) {};
      \draw (\n)--(\m);
    }
    \foreach \n in {1,3,5} {
      \pgfmathsetmacro{\a}{(\n-3)*360/14}
      \node[gr] (g\n) at (\a:1.5) {};
      \draw[lightgray] (g\n)--(\n);
    }
    \foreach \n in {8,10,12} {
      \pgfmathsetmacro{\a}{(\n-3)*360/14}
      \node[gb] (g\n) at (\a:1.5) {};
      \draw[lightgray] (g\n)--(\n);
    }
  \end{tikzpicture}
  \hspace*{\fill}
  \caption{The two $2$-colourings of type \type{0}{+} of a tac representing a
    variable that appears trice negatively and trice positively.}
  \label{fig:0+tac}
\end{figure}
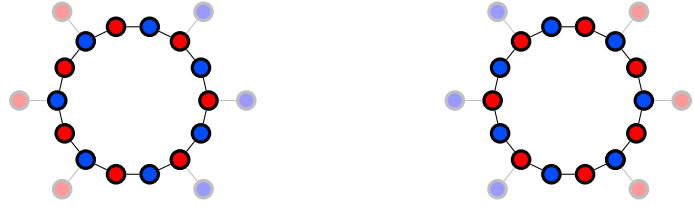
    Since the vertices the outlets are attached to have two degree $2$ neighbours on the cycle, the definition of \type0+-colouring implies that the outlet has the opposite colour to the vertex it is attached to. We furthermore encode that the literal an outlet encodes is true by the outlet having colour $1$ and false by the outlet having colour $2$. Note that property \ref{property:tac} is guaranteed through the choice of location on the cycle the outlets for occurrences of $x$ and $\overline{x}$.

The stc is obtained from a claw $K_{1,3}$ by adding a false twin to the central
vertex and a true twin to each leaf.
One outlet is attached to one copy of each leaf. True and false are encoded by the outlet having colour $1$ or $2$ and their neighbours having the opposite colour, respectively.
Observe that in every $2$-colouring of
type \type{0}{+} the true twins receive the same colour, and the false twins receive
opposite colours.

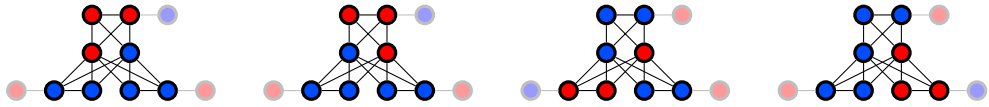
\begin{figure}[hbtp]
  \hspace*{\fill}
  \begin{tikzpicture}[scale=0.5]
    \node[r] (c!l) at (0,0) {};
    \node[b] (c!r) at (1,0) {};
    \node[r] (t!l) at (0,1) {};
    \node[r] (t!r) at (1,1) {};
    \draw (t!l)--(c!l)--(t!r)--(c!r)--(t!l)--(t!r);
    \foreach \d/\a/\c in {l/225/b, r/315/b} {
      \foreach \s in {l,r} {
        \node[\c] (\d!\s) at ($(c!\s) + (\a:1.414)$) {};
        \draw (c!l)--(\d!\s)--(c!r);
      }
      \draw (\d!l)--(\d!r);
    }
    \node[gr] (l!g) at ($(l!l)+(-1,0)$) {}; \draw[lightgray] (l!g)--(l!l);
    \node[gr] (r!g) at ($(r!r)+( 1,0)$) {}; \draw[lightgray] (r!g)--(r!r);
    \node[gb] (t!g) at ($(t!r)+( 1,0)$) {}; \draw[lightgray] (t!g)--(t!r);
  \end{tikzpicture}
  \hspace*{\fill}
  \begin{tikzpicture}[scale=0.5]
    \node[b] (c!l) at (0,0) {};
    \node[r] (c!r) at (1,0) {};
    \node[r] (t!l) at (0,1) {};
    \node[r] (t!r) at (1,1) {};
    \draw (t!l)--(c!l)--(t!r)--(c!r)--(t!l)--(t!r);
    \foreach \d/\a/\c in {l/225/b, r/315/b} {
      \foreach \s in {l,r} {
        \node[\c] (\d!\s) at ($(c!\s) + (\a:1.414)$) {};
        \draw (c!l)--(\d!\s)--(c!r);
      }
      \draw (\d!l)--(\d!r);
    }
    \node[gr] (l!g) at ($(l!l)+(-1,0)$) {}; \draw[lightgray] (l!g)--(l!l);
    \node[gr] (r!g) at ($(r!r)+( 1,0)$) {}; \draw[lightgray] (r!g)--(r!r);
    \node[gb] (t!g) at ($(t!r)+( 1,0)$) {}; \draw[lightgray] (t!g)--(t!r);
  \end{tikzpicture}
  \hspace*{\fill}
  \begin{tikzpicture}[scale=0.5]
    \node[b] (c!l) at (0,0) {};
    \node[r] (c!r) at (1,0) {};
    \node[b] (t!l) at (0,1) {};
    \node[b] (t!r) at (1,1) {};
    \draw (t!l)--(c!l)--(t!r)--(c!r)--(t!l)--(t!r);
    \foreach \d/\a/\c in {l/225/r, r/315/b} {
      \foreach \s in {l,r} {
        \node[\c] (\d!\s) at ($(c!\s) + (\a:1.414)$) {};
        \draw (c!l)--(\d!\s)--(c!r);
      }
      \draw (\d!l)--(\d!r);
    }
    \node[gb] (l!g) at ($(l!l)+(-1,0)$) {}; \draw[lightgray] (l!g)--(l!l);
    \node[gr] (r!g) at ($(r!r)+( 1,0)$) {}; \draw[lightgray] (r!g)--(r!r);
    \node[gr] (t!g) at ($(t!r)+( 1,0)$) {}; \draw[lightgray] (t!g)--(t!r);
  \end{tikzpicture}
  \hspace*{\fill}
  \begin{tikzpicture}[scale=0.5]
    \node[b] (c!l) at (0,0) {};
    \node[r] (c!r) at (1,0) {};
    \node[b] (t!l) at (0,1) {};
    \node[b] (t!r) at (1,1) {};
    \draw (t!l)--(c!l)--(t!r)--(c!r)--(t!l)--(t!r);
    \foreach \d/\a/\c in {l/225/b, r/315/r} {
      \foreach \s in {l,r} {
        \node[\c] (\d!\s) at ($(c!\s) + (\a:1.414)$) {};
        \draw (c!l)--(\d!\s)--(c!r);
      }
      \draw (\d!l)--(\d!r);
    }
    \node[gr] (l!g) at ($(l!l)+(-1,0)$) {}; \draw[lightgray] (l!g)--(l!l);
    \node[gb] (r!g) at ($(r!r)+( 1,0)$) {}; \draw[lightgray] (r!g)--(r!r);
    \node[gr] (t!g) at ($(t!r)+( 1,0)$) {}; \draw[lightgray] (t!g)--(t!r);
  \end{tikzpicture}
  \hspace*{\fill}
  \caption{Some $2$-colourings of type \type{0}{+} of the stc.}
  \label{fig:0+stc}
\end{figure}
To ensure that both false twins see the opposite colour at least once, the twin-pairs cannot all have the same colour which enforces property \ref{property:stc}. 
Finally, property \ref{property:combiningSTCandTAC} clearly holds and therefore \type0+-$2$-colouring is \NP-complete by \cref{lem:reductionFromNAE3SAT}.
\end{proof}

\subsubsection{\Type{0}{2}-colouring for \boldmath$q=2$}
First recall that $c:V(G)\rightarrow [q]$ is a \type02-$q$-colouring of $G$ if for every vertex $v$ we have $|N_G(v)\cap c^{-1}(c(v))|$ is even and $|N_G(v)\cap c^{-1}(i)|$ is even and larger than $0$ for every $i\not=c(v)$. We prove the following.

\begin{theorem}\label{link2:02}
    For $q = 2$ it is \NP-complete to decide whether a graph admits a \type02-$q$-colouring.
\end{theorem}
\begin{proof}

    We give a reduction from \NSAT. 
    First observe the following.
    \begin{claim}\label{claim:02degree2}
        All vertices of degree two have two
        neighbours of opposite to their own colour in every
        \type02-colouring with two colours.
    \end{claim}

    The tac of variable $x$ is a subdivision of a prism. 
    Each
    of the original vertices of the prism is connected to one of two outlet vertices forming one outlet. The tac is bipartite and its two
    $2$-colourings of type \type02 are the proper $2$-colourings. This is enforced by the subdivision vertices and \cref{claim:02degree2}. 

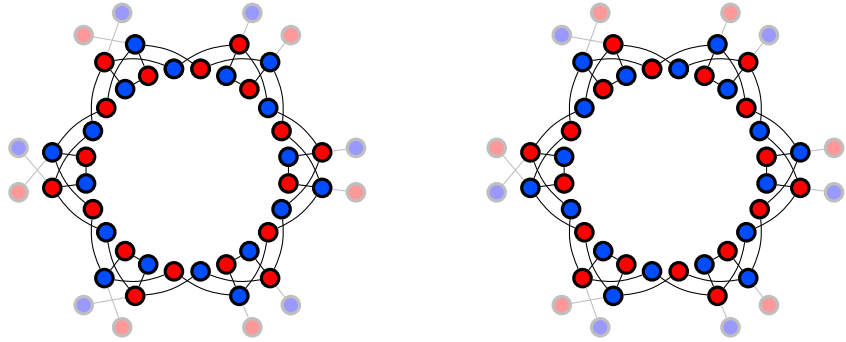
\begin{figure}[hbtp]
  \hspace*{\fill}
  \begin{tikzpicture}[scale=0.45, bend angle=15]
    \foreach \n in {0,...,25} {
      \pgfmathtruncatemacro{\q}{\n / 4}
      \pgfmathtruncatemacro{\r}{\n - 4*\q}
      \pgfmathsetmacro{\z}{67.5 + 15*\n}
      \pgfmathsetmacro{\a}{82.5 + 15*\n}
      \pgfmathsetmacro{\b}{97.5 + 15*\n}
      \pgfmathtruncatemacro{\p}{\n - 1}
      \pgfmathtruncatemacro{\m}{\n - 2}
      \ifcase \n         \node[a]  (i\n) at (\a:3) {}; 
        \or              \node[a]  (i\n) at (\a:3) {}; 
        \else \ifcase \r \node[r]  (i\n) at (\a:3) {}; 
                         \draw (o\m) to[bend right] (i\n);
                 \or     \node[b]  (i\n) at (\a:3) {}; 
                         \draw (o\m) to[bend right] (i\n);
                 \or     \node[r]  (i\n) at (\a:3) {}; 
                         \node[b]  (o\n) at (\a:4) {};
                         \draw (i\m) to[bend right] (o\n);
                         \draw (o\n)--(i\n);
                         \ifnum\n<13\node[gr] (g\n) at (\b:5) {};
                         \else\node[gr] (g\n) at (\a:5) {};\fi
                         \draw[lightgray] (o\n)--(g\n);
                 \or     \node[b]  (i\n) at (\a:3) {}; 
                         \node[r]  (o\n) at (\a:4) {};
                         \draw (i\m) to[bend right] (o\n);
                         \draw (o\n)--(i\n)--(i\p);
                         \ifnum\n<13\node[gb] (g\n) at (\z:5) {};
                         \else\node[gb] (g\n) at (\a:5) {};\fi
                         \draw[lightgray] (o\n)--(g\n);
              \fi     
      \fi
    }
  \end{tikzpicture}
  \hspace*{\fill}
  \begin{tikzpicture}[scale=0.45, bend angle=15]
    \foreach \n in {0,...,25} {
      \pgfmathtruncatemacro{\q}{\n / 4}
      \pgfmathtruncatemacro{\r}{\n - 4*\q}
      \pgfmathsetmacro{\z}{67.5 + 15*\n}
      \pgfmathsetmacro{\a}{82.5 + 15*\n}
      \pgfmathsetmacro{\b}{97.5 + 15*\n}
      \pgfmathtruncatemacro{\p}{\n - 1}
      \pgfmathtruncatemacro{\m}{\n - 2}
      \ifcase \n         \node[a]  (i\n) at (\a:3) {}; 
        \or              \node[a]  (i\n) at (\a:3) {}; 
        \else \ifcase \r \node[b]  (i\n) at (\a:3) {}; 
                         \draw (o\m) to[bend right] (i\n);
                 \or     \node[r]  (i\n) at (\a:3) {}; 
                         \draw (o\m) to[bend right] (i\n);
                 \or     \node[b]  (i\n) at (\a:3) {}; 
                         \node[r]  (o\n) at (\a:4) {};
                         \draw (i\m) to[bend right] (o\n);
                         \draw (o\n)--(i\n);
                         \ifnum\n<13\node[gb] (g\n) at (\b:5) {};
                         \else\node[gb] (g\n) at (\a:5) {};\fi
                         \draw[lightgray] (o\n)--(g\n);
                 \or     \node[r]  (i\n) at (\a:3) {}; 
                         \node[b]  (o\n) at (\a:4) {};
                         \draw (i\m) to[bend right] (o\n);
                         \draw (o\n)--(i\n)--(i\p);
                         \ifnum\n<13\node[gr] (g\n) at (\z:5) {};
                         \else\node[gr] (g\n) at (\a:5) {};\fi
                         \draw[lightgray] (o\n)--(g\n);
              \fi     
      \fi
    }
  \end{tikzpicture}
  \hspace*{\fill}
  \caption{The tac for type \type02 representing a variable that appears
    trice negatively and trice positively.}
    \label{fig:tac02}
\end{figure}
    
    Let $Y$ and $Z$ be the two parts of the bipartition of the tac. Observe that the two vertices forming an outlet are in different parts of the bipartition and hence receive different colours in any \type02-$2$-colouring. To distinguish occurrences of $x$ and $\overline{x}$ we declare one vertex the first and one vertex second vertex of an outlet. For occurrences of $x$ the first vertex of the outlet is the one in $Y$ and for occurrences of $\overline{x}$ the first vertex of the corresponding outlet is the vertex in $Z$. We encode that the literal corresponding to an outlet is true by its first vertex having colour 1 and we encode false by the first vertex having colour 2. With this encoding we enforce that property \ref{property:tac} is satisfied.

The stc of a clause $c$ is obtained from two copies of $K_{1,3}$ by adding a true twin to the central vertex  called $A^c$ and $B^c$. Additionally, $A^c$ and $B^c$ are linked by  paths of length 3, and outlets attached (see \cref{fig:stc02}). First, we observe 
that the two
true twins of $A^c$ receive the 
same colour in every \type02-colouring of the stc. The same is true for the two true twins
 of $B^c$. Thereafter, we see that the remaining three vertices of $A^c$ (or $B^c$, respectively) will not be coloured all equally.
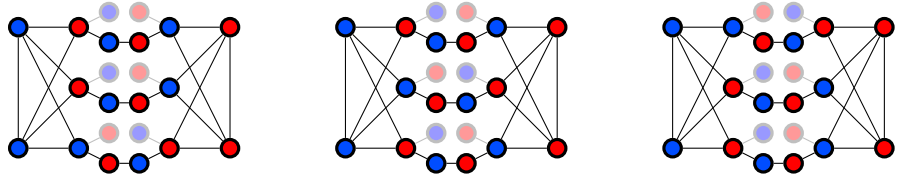
\begin{figure}[hbtp]
  \hspace*{\fill}
  \begin{tikzpicture}[xscale=0.4, yscale=0.2]
    \foreach \x/\c in {1/b, 8/r} \foreach \y in {1,9} \node[\c] (\x\y) at (\x,\y) {};
    \foreach \y/\c in {1/b, 5/r, 9/r} \node[\c] (3\y) at (3,\y) {};
    \foreach \y/\c in {1/r, 5/b, 9/b} \node[\c] (6\y) at (6,\y) {};
    \draw (11)--(19)  (81)--(89);
    \foreach \n in {31, 35, 39} \draw (11)--(\n)--(19);
    \foreach \n in {61, 65, 69} \draw (81)--(\n)--(89);
    \foreach \y/\c in {0/r, 2/gr, 4/b, 6/gb, 8/b, 10/gb} \node[\c] (4\y) at (4,\y) {};
    \foreach \y/\c in {0/b, 2/gb, 4/r, 6/gr, 8/r, 10/gr} \node[\c] (5\y) at (5,\y) {};
    \draw (31)--(40)--(50)--(61)  (35)--(44)--(54)--(65)  (39)--(48)--(58)--(69);
    \draw[lightgray] (31)--(42)  (52)--(61)  (35)--(46)  (56)--(65)  (39)--(410)  (510)--(69);
  \end{tikzpicture}
  \hspace*{\fill}
  \begin{tikzpicture}[xscale=0.4, yscale=0.2]
    \foreach \x/\c in {1/b, 8/r} \foreach \y in {1,9} \node[\c] (\x\y) at (\x,\y) {}; 
    \foreach \y/\c in {1/r, 5/b, 9/r} \node[\c] (3\y) at (3,\y) {};
    \foreach \y/\c in {1/b, 5/r, 9/b} \node[\c] (6\y) at (6,\y) {};
    \draw (11)--(19)  (81)--(89);
    \foreach \n in {31, 35, 39} \draw (11)--(\n)--(19);
    \foreach \n in {61, 65, 69} \draw (81)--(\n)--(89);
    \foreach \y/\c in {0/b, 2/gb, 4/r, 6/gr, 8/b, 10/gb} \node[\c] (4\y) at (4,\y) {};
    \foreach \y/\c in {0/r, 2/gr, 4/b, 6/gb, 8/r, 10/gr} \node[\c] (5\y) at (5,\y) {};
    \draw (31)--(40)--(50)--(61)  (35)--(44)--(54)--(65)  (39)--(48)--(58)--(69);
    \draw[lightgray] (31)--(42)  (52)--(61)  (35)--(46)  (56)--(65)  (39)--(410)  (510)--(69);
  \end{tikzpicture}
  \hspace*{\fill}
  \begin{tikzpicture}[xscale=0.4, yscale=0.2]
    \foreach \x/\c in {1/b, 8/r} \foreach \y in {1,9} \node[\c] (\x\y) at (\x,\y) {};
    \foreach \y/\c in {1/r, 5/r, 9/b} \node[\c] (3\y) at (3,\y) {};
    \foreach \y/\c in {1/b, 5/b, 9/r} \node[\c] (6\y) at (6,\y) {};
    \draw (11)--(19)  (81)--(89);
    \foreach \n in {31, 35, 39} \draw (11)--(\n)--(19);
    \foreach \n in {61, 65, 69} \draw (81)--(\n)--(89);
    \foreach \y/\c in {0/b, 2/gb, 4/b, 6/gb, 8/r, 10/gr} \node[\c] (4\y) at (4,\y) {};
    \foreach \y/\c in {0/r, 2/gr, 4/r, 6/gr, 8/b, 10/gb} \node[\c] (5\y) at (5,\y) {};
    \draw (31)--(40)--(50)--(61)  (35)--(44)--(54)--(65)  (39)--(48)--(58)--(69);
    \draw[lightgray] (31)--(42)  (52)--(61)  (35)--(46)  (56)--(65)  (39)--(410)  (510)--(69);
  \end{tikzpicture}
  \hspace*{\fill}
  \caption{Three different \type02-colourings of the stc. The other three 
    colourings can be obtained by swapping colours.
    }
  \label{fig:stc02}
\end{figure}
This implies that the three vertices that are part of the three outlets and are attached at $A^c$ cannot all have the same colour.  We encode that a literal is true by the vertex of the corresponding outlet which is attached to $A^c$ having colour $1$ and false by the vertex of the corresponding outlet which is attached to $A^c$ having colour $2$. We observe that this choice guarantees property  \ref{property:stc}. For property \ref{property:combiningSTCandTAC} to be satisfies we need to glue the first vertex of the outlet of the tac corresponding to literal $\ell$ to the neighbour of the outlet which is attached to $A^c$. This ensures that the correct colouring are compatible. We conclude that by \cref{lem:reductionFromNAE3SAT} \type02-$2$-colouring is \NP-complete.
\end{proof}

\subsubsection{\Type{v}{$\star$}-colouring for \boldmath$q=2$}
Recall that $c:V(G)\rightarrow [q]$ is a \type{v}{*}-$q$-colouring of $G$  if for every vertex $v$ it holds that $|N_G(v)\cap c^{-1}(c(v))|$ is odd or $0$.  
We prove the following.

\begin{theorem}\label{thm2:v*}\label{link2:v*}
    For $q = 2$ it is \NP-complete to decide whether a graph admits a \type{v}{*}-$q$-colouring.
\end{theorem}
\begin{proof}
    We provide a reduction from \NSAT. First, observe the following.
    \begin{claim}\label{claim:K4v*}
      Let $G$ be a graph and $v_1,\dots, v_4$ the vertices of an induces $K_4$ in $G$ such that $v_1,\dots, v_3$ have no neighbours in $V(G)\setminus \{v_1,\dots, v_4\}$. In any \type{v}{*}-colouring $c: V(G)\rightarrow [2]$ of $G$, $v_4$ must have an odd number of neighbours of its own colour among $v_1,v_2,v_3$.  
    \end{claim}
    \begin{proof}
        Let $C,C'$ be the two colour classes of $c$ and assume $|C\cap \{v_1,\dots, v_4\}|\leq |C'\cap \{v_1,\dots, v_4\}|$.
         If $C\cap \{v_1,\dots, v_4\}$ is either $0$ or $2$, the statement holds. Hence, assume $C\cap \{v_1,\dots, v_4\}=1$. In this case there is $v\in C'\cap\{v_1,v_2, v_3\}$ and since $|C'\cap\{v_1,\dots, v_4\}|=3$, $v$ has two neighbours of its own colour contradicting that $c$ is a \type{v}{*}-colouring. 
    \end{proof}
    First, we build a gadget, called the \emph{wheel gadget}, from a $6$-wheel by adding three vertices $a,b,c$, each being adjacent to two consecutive vertices of the rim of the wheel (see \cref{fig:v*wheel}). We say that the vertices  $a,b,c$ are the outlets of the wheel gadget. In a graph $G$ vertices $a,b,c$ are connected by a wheel gadget, if there are vertices $w,w_1,\dots,w_6$ for which $a,b,c,w,w_1,\dots,w_6$ induce a wheel gadget in $G$ with outlets $a,b,c$  and $w,w_1,\dots, w_6$ do not have neighbours outside the wheel gadget.  
    Observe that for a graph $G$ with \type{v}{*}-$2$-colouring and $a,b,c$ connected by a wheel gadget in $G$, the wheel gadget must be coloured by one of the three colourings given in 
    Figure~\ref{fig:v*wheel} up to permutation of the two colours. In particular, convince yourself that the central vertex of the wheel cannot have precisely one neighbour of either colour. 
\begin{figure}[hbtp]
  \hspace*{\fill}
  \begin{tikzpicture}[scale=0.6]
    \foreach \r/\c in {0/gb, 1/r, 2/b, 3/gb} \node[\c] (a\r) at ($(0,0)+(60:\r)$) {};
    \foreach \r/\c in {0/b, 1/r, 2/r} \node[\c] (b\r) at ($(1,0)+(60:\r)$) {};
    \foreach \r/\c in {0/r, 1/b} \node[\c] (c\r) at ($(2,0)+(60:\r)$) {};
    \node[gb] (d0) at (3,0) {};
    \draw[lightgray] (b0)--(a0)--(a1)  (a2)--(a3)--(b2)  (c0)--(d0)--(c1);
    \draw (b0)--(a1)--(a2)--(b2)--(c1)--(c0)--(b0);
    \foreach \n in {b0, a1, a2, b2, c1, c0} \draw (\n)--(b1);
  \end{tikzpicture}
  \hspace*{\fill}
  \begin{tikzpicture}[scale=0.6]
    \foreach \r/\c in {0/gb, 1/b, 2/b, 3/gb} \node[\c] (a\r) at ($(0,0)+(60:\r)$) {};
    \foreach \r/\c in {0/b, 1/r, 2/b} \node[\c] (b\r) at ($(1,0)+(60:\r)$) {};
    \foreach \r/\c in {0/b, 1/b} \node[\c] (c\r) at ($(2,0)+(60:\r)$) {};
    \node[gb] (d0) at (3,0) {};
    \draw[lightgray] (b0)--(a0)--(a1)  (a2)--(a3)--(b2)  (c0)--(d0)--(c1);
    \draw (b0)--(a1)--(a2)--(b2)--(c1)--(c0)--(b0);
    \foreach \n in {b0, a1, a2, b2, c1, c0} \draw (\n)--(b1);
  \end{tikzpicture}
  \hspace*{\fill}
  \begin{tikzpicture}[scale=0.6]
    \foreach \r/\c in {0/gb, 1/b, 2/r, 3/gb} \node[\c] (a\r) at ($(0,0)+(60:\r)$) {};
    \foreach \r/\c in {0/r, 1/r, 2/b} \node[\c] (b\r) at ($(1,0)+(60:\r)$) {};
    \foreach \r/\c in {0/b, 1/r} \node[\c] (c\r) at ($(2,0)+(60:\r)$) {};
    \node[gb] (d0) at (3,0) {};
    \draw[lightgray] (b0)--(a0)--(a1)  (a2)--(a3)--(b2)  (c0)--(d0)--(c1);
    \draw (b0)--(a1)--(a2)--(b2)--(c1)--(c0)--(b0);
    \foreach \n in {b0, a1, a2, b2, c1, c0} \draw (\n)--(b1);
  \end{tikzpicture}
  \hspace*{\fill}
  \caption{The \type{v}{*}-colourings of the $6$-wheel up to permutation of colours.}
  \label{fig:v*wheel}
\end{figure}
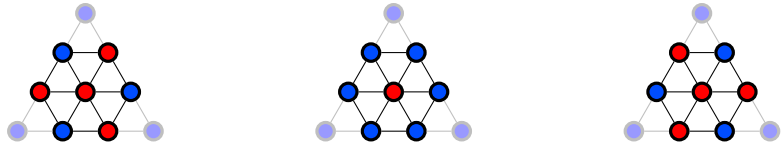

The tac of variable $x$ 
consists of one wheel gadget for each occurrence of the variable and two additional wheel gadgets with outlets added to two consecutive vertices of the corresponding $6$-wheels. Each copy of $K_4$
provides its cut vertex with an odd number of neighbours of the same colour
in each $2$-colouring of type \type{v}{*} by \cref{claim:K4v*} (note that we choose to depict $K_4$ with all vertices of the same colour, but equivalently the $K_4$ could also have two vertices of one colour and two of the other). This ensures that all $6$-wheels
receive the central $2$-colouring shown in Figure~\ref{fig:v*wheel} (or the central colouring with colours swapped). We observe that therefore the outlet and its two neighbours must receive the same colour. Hence, without loss of generality, we encode the corresponding literal to be true when its outlet plus the two neighbours have colour $1$ and false it the outlet and its two neighbours have colour $2$.
The $6$-wheels for occurrences of $x$ and $\overline{x}$ are connected by
 two negators consisting of $2K_3 \Join K_1$. By \cref{claim:K4v*} the two neighbours of the central vertex must have different colours, see
Figure~\ref{fig:v*tac}. Hence, the tac satisfies property \ref{property:tac}.

\begin{figure}[h]
  \centering
  \begin{tikzpicture}[scale=0.6]
    \node[b] (a0) at (4,7) {};
    \node[r] (r0) at (4,3) {};
    \foreach \n in {1,...,12} {
      \pgfmathtruncatemacro{\x}{\n+4}
      \pgfmathtruncatemacro{\m}{\n-1}
      \pgfmathtruncatemacro{\q}{\n/3}
      \pgfmathtruncatemacro{\r}{\n-3*\q}
      \ifcase\r{
          \node[b]  (a\n) at (\x,7)            {}; \draw (a\m)--(a\n);
          \node[b]  (b\n) at ($(a\n)+(120:1)$) {}; \draw (a\m)--(b\n)--(b\m);
          \node[b]  (c\n) at ($(a\n)+(120:2)$) {}; \draw (b\m)--(c\n)--(c\m);
          \ifnum\q=1
            \node[b] (d\n) at ($(a\n)+(120:3)$) {};
            \draw (c\m)--(d\n)--(c\n);
          \else
            \node[gb] (d\n) at ($(a\n)+(120:3)$) {};
            \draw[lightgray] (c\m)--(d\n)--(c\n);
          \fi
          \draw (a\n)--(b\n)--(c\n);
          \node[r]  (r\n) at (\x,3)            {}; \draw (r\m)--(r\n);
          \node[r]  (s\n) at ($(r\n)+(240:1)$) {}; \draw (r\m)--(s\n)--(s\m);
          \node[r]  (t\n) at ($(r\n)+(240:2)$) {}; \draw (s\m)--(t\n)--(t\m);
          \ifnum\q=4
            \node[r] (u\n) at ($(r\n)+(240:3)$) {};
            \draw (t\m)--(u\n)--(t\n);
          \else
            \node[gr] (u\n) at ($(r\n)+(240:3)$) {};
            \draw[lightgray] (t\m)--(u\n)--(t\n);
          \fi
          \draw (r\n)--(s\n)--(t\n);
        }
      \or{
          \node[b]  (a\n) at (\x,7)            {};
          \node[b]  (b\n) at ($(a\n)+(120:1)$) {};
          \draw (a\m)--(a\n)--(b\n)--(a\m);
          \node[r]  (r\n) at (\x,3)            {};
          \node[r]  (s\n) at ($(r\n)+(240:1)$) {};
          \draw (r\m)--(r\n)--(s\n)--(r\m);
        }
      \or{
          \node[b]  (a\n) at (\x,7)            {};
          \node[r]  (b\n) at ($(a\n)+(120:1)$) {};
          \node[b]  (c\n) at ($(a\n)+(120:2)$) {};
          \draw (a\m)--(a\n)--(b\n)--(c\n)--(b\m)--(b\n)--(a\m);
          \node[r]  (r\n) at (\x,3)            {};
          \node[b]  (s\n) at ($(r\n)+(240:1)$) {};
          \node[r]  (t\n) at ($(r\n)+(240:2)$) {};
          \draw (r\m)--(r\n)--(s\n)--(t\n)--(s\m)--(s\n)--(r\m);
         }
      \fi
    }
    \foreach \n in {0, 3, 6, 9, 12} {
      \pgfmathsetmacro{\x}{\n+4.0}
      \node[b] (f\n) at (\x,6.0) {}; \node[r] (w\n) at (\x,4.0) {};
      \pgfmathsetmacro{\x}{\n+3.5}
      \node[b] (e\n) at (\x,6.5) {}; \node[r] (v\n) at (\x,3.5) {};
      \pgfmathsetmacro{\x}{\n+4.5}
      \node[b] (g\n) at (\x,6.5) {}; \node[r] (x\n) at (\x,3.5) {};
      \draw (a\n)--(e\n)--(f\n)--(g\n)--(a\n)--(f\n)  (e\n)--(g\n);
      \draw (r\n)--(v\n)--(w\n)--(x\n)--(r\n)--(w\n)  (v\n)--(x\n);
    }
    \foreach \n/\x in {k/1, l/2, n/3}    \node[b] (n\n) at (\x,5) {};
    \foreach \n/\x in {q/17, r/18, s/19} \node[r] (n\n) at (\x,5) {};
    \foreach \m/\y in {l/4.5, u/5.5} {
      \foreach \n/\x in {i/1.5,  j/2.5}  \node[b] (\m\n) at (\x,\y) {};
      \foreach \n/\x in {o/17.5, p/18.5} \node[r] (\m\n) at (\x,\y) {};      
    }
    \draw (nk)--(nl)--(ui)--(nk)--(li)--(nl)  (li)--(ui);
    \draw (nn)--(nl)--(uj)--(nn)--(lj)--(nl)  (lj)--(uj);
    \draw (nq)--(nr)--(uo)--(nq)--(lo)--(nr)  (lo)--(uo);
    \draw (ns)--(nr)--(up)--(ns)--(lp)--(nr)  (lp)--(up);
    \draw[rounded corners=7mm] (r0)--  (2,3)    --(nl);
    \draw[rounded corners=7mm] (d3)--  (2,9.598)--(nl);
    \draw[rounded corners=7mm] (a12)--(18,7)    --(nr);
    \draw[rounded corners=7mm] (u12)--(18,0.402)--(nr);
  \end{tikzpicture}
  \caption{A typical $2$-colouring of type \type{v}{*} of a tac representing
    a variable that appears trice negatively and trice positively. The second colouring is obtained by swapping colours.}
  \label{fig:v*tac}
\end{figure}
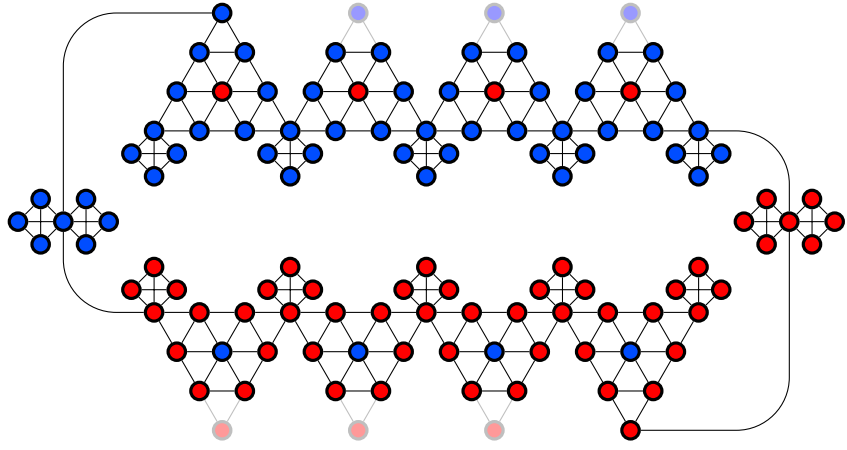

The stc is depicted in Figure~\ref{fig:v*stc}. Note that we encode that the literal corresponding to an outlet is set to true by the two outlets and their joint neighbour having colour $1$ and false by the two outlets and their joint neighbour having colour $2$. It is easy to verify that the stc satisfies property \ref{property:stc}.

\begin{figure}[h]
  \hspace*{\fill}
  \begin{tikzpicture}
        \node[b] (x) at (0,3) {};
        \node[gb] (x1) at (0.3,3.5) {};
        \node[gb] (x2) at (-0.3,3.5) {};
        \node[r] (y) at (1,0.5) {};
        \node[gr] (y1) at (1.5,0.8) {};
        \node[gr] (y2) at (1.5,0.2) {};
        \node[r] (z) at (-1,0.5) {};
        \node[gr] (z1) at (-1.5,0.8) {};
        \node[gr] (z2) at (-1.5,0.2) {};
        \node[b] (a) at (0,0) {};
        \node[r] (b) at (0,1) {};
        \node[b] (c) at (0,2) {};
        \node[b] (d1) at (0.3,2.3) {};
        \node[b] (d2) at (0.3,1.7) {};
        \node[b] (d3) at (0.6,2) {};
        \node[b] (e1) at (-0.3,2.3) {};
        \node[b] (e2) at (-0.3,1.7) {};
        \node[b] (e3) at (-0.6,2) {};
        \node[r] (f) at (1,1.5) {};
    \draw[lightgray] (x1)--(x)--(x2)--(x1);
    \draw[lightgray] (y1)--(y)--(y2)--(y1);
    \draw[lightgray] (z1)--(z)--(z2)--(z1);
    \draw (x)--(c)--(d1)--(d2)--(d3)--(d1)(d3)--(c)--(d2)(c)--(e1)--(e2)--(e3)--(e1)(e3)--(c)--(e2)(c)--(b)--(f)(b)--(a)--(y)--(b)--(z)--(a);
  \end{tikzpicture}
  \hspace*{\fill}
  \begin{tikzpicture}
        \node[b] (x) at (0,3) {};
        \node[gb] (x1) at (0.3,3.5) {};
        \node[gb] (x2) at (-0.3,3.5) {};
        \node[r] (y) at (1,0.5) {};
        \node[gr] (y1) at (1.5,0.8) {};
        \node[gr] (y2) at (1.5,0.2) {};
        \node[b] (z) at (-1,0.5) {};
        \node[gb] (z1) at (-1.5,0.8) {};
        \node[gb] (z2) at (-1.5,0.2) {};
        \node[b] (a) at (0,0) {};
        \node[r] (b) at (0,1) {};
        \node[b] (c) at (0,2) {};
        \node[b] (d1) at (0.3,2.3) {};
        \node[b] (d2) at (0.3,1.7) {};
        \node[b] (d3) at (0.6,2) {};
        \node[b] (e1) at (-0.3,2.3) {};
        \node[b] (e2) at (-0.3,1.7) {};
        \node[b] (e3) at (-0.6,2) {};
        \node[b] (f) at (1,1.5) {};
    \draw[lightgray] (x1)--(x)--(x2)--(x1);
    \draw[lightgray] (y1)--(y)--(y2)--(y1);
    \draw[lightgray] (z1)--(z)--(z2)--(z1);
    \draw (x)--(c)--(d1)--(d2)--(d3)--(d1)(d3)--(c)--(d2)(c)--(e1)--(e2)--(e3)--(e1)(e3)--(c)--(e2)(c)--(b)--(f)(b)--(a)--(y)--(b)--(z)--(a);
  \end{tikzpicture}
  \hspace*{\fill}
  \begin{tikzpicture}
        \node[b] (x) at (0,3) {};
        \node[gb] (x1) at (0.3,3.5) {};
        \node[gb] (x2) at (-0.3,3.5) {};
        \node[b] (y) at (1,0.5) {};
        \node[gb] (y1) at (1.5,0.8) {};
        \node[gb] (y2) at (1.5,0.2) {};
        \node[r] (z) at (-1,0.5) {};
        \node[gr] (z1) at (-1.5,0.8) {};
        \node[gr] (z2) at (-1.5,0.2) {};
        \node[b] (a) at (0,0) {};
        \node[r] (b) at (0,1) {};
        \node[b] (c) at (0,2) {};
        \node[b] (d1) at (0.3,2.3) {};
        \node[b] (d2) at (0.3,1.7) {};
        \node[b] (d3) at (0.6,2) {};
        \node[b] (e1) at (-0.3,2.3) {};
        \node[b] (e2) at (-0.3,1.7) {};
        \node[b] (e3) at (-0.6,2) {};
        \node[b] (f) at (1,1.5) {};
    \draw[lightgray] (x1)--(x)--(x2)--(x1);
    \draw[lightgray] (y1)--(y)--(y2)--(y1);
    \draw[lightgray] (z1)--(z)--(z2)--(z1);
    \draw (x)--(c)--(d1)--(d2)--(d3)--(d1)(d3)--(c)--(d2)(c)--(e1)--(e2)--(e3)--(e1)(e3)--(c)--(e2)(c)--(b)--(f)(b)--(a)--(y)--(b)--(z)--(a);
  \end{tikzpicture}
  \hspace*{\fill}
  \caption{Three  possible type \type{v}{*} $2$-colourings of the stc of a clause $c$ corresponding to different assignments (the other three assignments are obtained by swapping colours).}
  \label{fig:v*stc}
\end{figure}
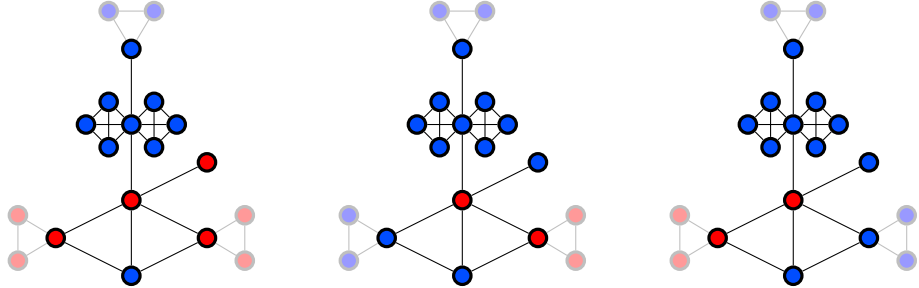

Finally, the tac and stc are compatible and the $\sigma$-$2$-colourings encoding that a literal is true/false of the tac are compatible with valid $\sigma$-$2$-colourings encoding that the literal is true/false. Since property \ref{property:combiningSTCandTAC} is also satisfied \type{v}{*}-colouring is \NP-complete by \cref{lem:reductionFromNAE3SAT}.
\end{proof}

\subsubsection{\Type1?, \Type1!, \Type{v}? and \Type{v}! for 2 colours.}

Recall that $c:V(G)\rightarrow [q]$ is a \type1?-$q$-colouring of $G$  if for every vertex $v$ we have $|N_G(v)\cap c^{-1}(c(v))|$ is odd and  $|N_G(v)\cap c^{-1}(i)|<2$ for every $i\not=c(v)$. We show the following.

\begin{theorem}\label{thm:21?}\label{link2:1?}
    For $q=2$ it is \NP-complete to decide whether a graph admits
    a \type1?-$q$-colouring.
\end{theorem}

\begin{proof}
    We provide a reduction from \NSAT. For an instance $\varphi$ of \NSAT we will construct a graph $H_\varphi$ by specifying a stc and tac for which every vertex has degree $2$ or $4$. 

    For variable $x$ that appears $t$ times (and $\lnot x$ appears $t$ times), the tac consists of a cycle $C_x$ of length $4t$ in which two pendants are attached to every second vertex of $C_x$ (see \cref{fig:1?tac} left). The two pendants attached to a vertex of $C_x$ form an outlet were along the cycles outlets for occurrences of $x$ and $\overline{x}$ alternate. 
    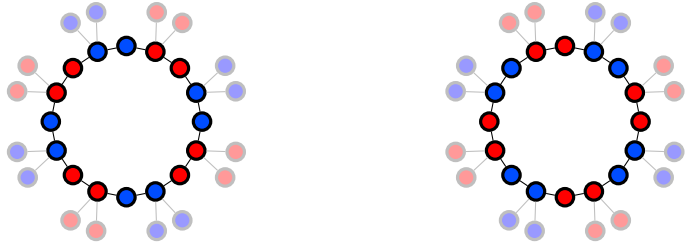
\begin{figure}[hbtp]
    \hspace*{\fill}
    \begin{tikzpicture}
    \foreach \n in {1,2,5,6,9,10,13,14} {
      \pgfmathsetmacro{\a}{(\n-3)*360/16}
      \node[r] (\n) at (\a:1) {};
    }
    \foreach \n in {3,4,7,8,11,12,15,16} {
      \pgfmathsetmacro{\a}{(\n-3)*360/16}
      \node[b] (\n) at (\a:1) {};
    }
    \foreach[count=\m] \n in {2,...,16} {
      \draw (\n) -- (\m);
    }
    \draw (16) -- (1);
    \foreach \n in {4,8,12,16} {
      \pgfmathsetmacro{\a}{(\n-3)*360/16+7}
      \pgfmathsetmacro{\b}{(\n-3)*360/16-7}
      \node[gb] (g\n) at (\a:1.5) {};
      \node[gb] (gg\n) at (\b:1.5) {};
      \draw[lightgray] (g\n)--(\n);
      \draw[lightgray] (gg\n)--(\n);
    }
    \foreach \n in {2,6,10,14} {
      \pgfmathsetmacro{\a}{(\n-3)*360/16+7}
      \pgfmathsetmacro{\b}{(\n-3)*360/16-7}
      \node[gr] (g\n) at (\a:1.5) {};
      \node[gr] (gg\n) at (\b:1.5) {};
      \draw[lightgray] (g\n)--(\n);
      \draw[lightgray] (gg\n)--(\n);
    }
  \end{tikzpicture}
  \hspace*{\fill}
  \begin{tikzpicture}
    \foreach \n in {1,2,5,6,9,10,13,14} {
      \pgfmathsetmacro{\a}{(\n-3)*360/16}
      \node[b] (\n) at (\a:1) {};
    }
    \foreach \n in {3,4,7,8,11,12,15,16} {
      \pgfmathsetmacro{\a}{(\n-3)*360/16}
      \node[r] (\n) at (\a:1) {};
    }
    \foreach[count=\m] \n in {2,...,16} {
      \draw (\n) -- (\m);
    }
    \draw (16) -- (1);
    \foreach \n in {4,8,12,16} {
      \pgfmathsetmacro{\a}{(\n-3)*360/16+7}
      \pgfmathsetmacro{\b}{(\n-3)*360/16-7}
      \node[gr] (g\n) at (\a:1.5) {};
      \node[gr] (gg\n) at (\b:1.5) {};
      \draw[lightgray] (g\n)--(\n);
      \draw[lightgray] (gg\n)--(\n);
    }
    \foreach \n in {2,6,10,14} {
      \pgfmathsetmacro{\a}{(\n-3)*360/16+7}
      \pgfmathsetmacro{\b}{(\n-3)*360/16-7}
      \node[gb] (g\n) at (\a:1.5) {};
      \node[gb] (gg\n) at (\b:1.5) {};
      \draw[lightgray] (g\n)--(\n);
      \draw[lightgray] (gg\n)--(\n);
    }
  \end{tikzpicture}
  \hspace*{\fill}
  \caption{The two colourings of the tac of a variable that appears four times positively and four times negatively.}
  \label{fig:1?tac}
\end{figure}
    Observe that every second vertex of $C_x$ has degree $2$ which implies that these vertices must have one neighbour of their own colour and one neighbour of the other colour by definition of \type1?-colouring. Therefore, the degree $2$ vertices force that restricted to the cycle each colour class induces a matching. This in particular implies that the cycle vertices with pendants attached have a neighbour coloured different to themselves on the cycle. Hence, both pendants must have the same colour as the vertex they are attached to. Furthermore, all outlets representing an occurrence of $x$ receive the same colour while all outlets representing an occurrence of $\overline{x}$ receive the other colour. Hence, we encode the corresponding literal to be true by both outlet vertices and the vertex they are attached to having colour $1$ and false by both outlet vertices and the vertex they are attached to having colour $2$ (this choice is arbitrary). 
    This ensures that property \ref{property:tac} holds.

    The stc consists of a $4$-dimensional hypercube with three edges subdivided (see \cref{fig:1?Stc} for an illustration). Here, the three subdivision vertices are the three outlets of the stc and as before the outlet plus its two neighbours having colour 1 encodes true and the outlet plus its two neighbours having colour 2 encodes false. 
   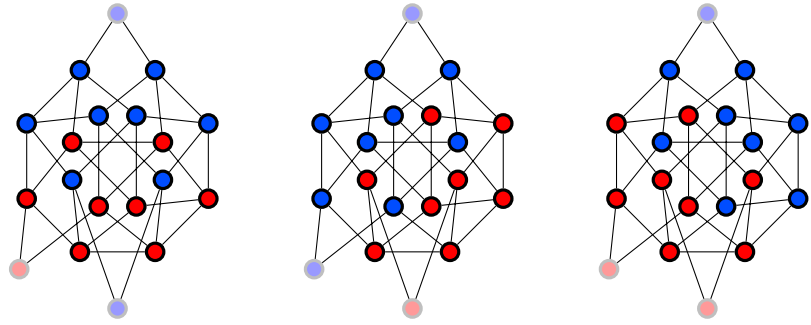
\begin{figure}[hbtp]
   \centering
  \begin{tikzpicture}[scale=0.26]
    \def \r {2.5cm} 
    \def \R {5cm} 
    \def \x {2.7cm} 
    \def \y {2cm} 
    \node[r]  (v1) at (22.5:\r) {};
    \node[b]  (v2) at (67.5:\r) {};
    \node[b]  (v3) at (112.5:\r) {};
    \node[r]  (v4) at (157.5:\r) {};
    \node[b]  (v5) at (202.5:\r) {};
    \node[r]  (v6) at (247.5:\r) {};
    \node[r]  (v7) at (292.5:\r) {};
    \node[b]  (v8) at (337.5:\r) {}; 
    \node[b]  (u1) at (22.5:\R) {};
    \node[b]  (u2) at (67.5:\R) {};
    \node[b]  (u3) at (112.5:\R) {};
    \node[b]  (u4) at (157.5:\R) {};
    \node[r]  (u5) at (202.5:\R) {};
    \node[r]  (u6) at (247.5:\R) {};
    \node[r]  (u7) at (292.5:\R) {};
    \node[r]  (u8) at (337.5:\R) {};
    \node[gb]  (x1) at (0,1.5*\R) {};
    \node[gb]  (x2) at (0,-1.5*\R) {};
    \node[gr]  (x3) at (-1*\R,-1.1*\R) {};
    \draw (u1)--(u2)--(x1)--(u3)--(u4)--(u5)--(u6)--(u7)--(u8)--(u1);
    \draw (u1)--(v2)--(u3)(u2)--(v3)--(u4)(u3)--(v4)--(u5)(u4)--(v5)--(u6)(u5)--(x3)--(v6)--(u7)(u6)--(v7)--(u8)(u7)--(v8)--(u1)(u8)--(v1)--(u2);
    \draw (v1)--(v4)--(v7)--(v2)--(v5)--(x2)--(v8)--(v3)--(v6)--(v1);

    \begin{scope}[xshift=15cm]
    \node[b]  (v1) at (22.5:\r) {};
    \node[r]  (v2) at (67.5:\r) {};
    \node[b]  (v3) at (112.5:\r) {};
    \node[b]  (v4) at (157.5:\r) {};
    \node[r]  (v5) at (202.5:\r) {};
    \node[b]  (v6) at (247.5:\r) {};
    \node[r]  (v7) at (292.5:\r) {};
    \node[r]  (v8) at (337.5:\r) {}; 
    \node[r]  (u1) at (22.5:\R) {};
    \node[b]  (u2) at (67.5:\R) {};
    \node[b]  (u3) at (112.5:\R) {};
    \node[b]  (u4) at (157.5:\R) {};
    \node[b]  (u5) at (202.5:\R) {};
    \node[r]  (u6) at (247.5:\R) {};
    \node[r]  (u7) at (292.5:\R) {};
    \node[r]  (u8) at (337.5:\R) {};
    \node[gb]  (x1) at (0,1.5*\R) {};
    \node[gr]  (x2) at (0,-1.5*\R) {};
    \node[gb]  (x3) at (-1*\R,-1.1*\R) {};
    \draw (u1)--(u2)--(x1)--(u3)--(u4)--(u5)--(u6)--(u7)--(u8)--(u1);
    \draw (u1)--(v2)--(u3)(u2)--(v3)--(u4)(u3)--(v4)--(u5)(u4)--(v5)--(u6)(u5)--(x3)--(v6)--(u7)(u6)--(v7)--(u8)(u7)--(v8)--(u1)(u8)--(v1)--(u2);
    \draw (v1)--(v4)--(v7)--(v2)--(v5)--(x2)--(v8)--(v3)--(v6)--(v1);
    \end{scope}
    \begin{scope}[xshift=30cm]
    \node[b]  (v1) at (22.5:\r) {};
    \node[b]  (v2) at (67.5:\r) {};
    \node[r]  (v3) at (112.5:\r) {};
    \node[b]  (v4) at (157.5:\r) {};
    \node[r]  (v5) at (202.5:\r) {};
    \node[r]  (v6) at (247.5:\r) {};
    \node[b]  (v7) at (292.5:\r) {};
    \node[r]  (v8) at (337.5:\r) {}; 
    \node[b]  (u1) at (22.5:\R) {};
    \node[b]  (u2) at (67.5:\R) {};
    \node[b]  (u3) at (112.5:\R) {};
    \node[r]  (u4) at (157.5:\R) {};
    \node[r]  (u5) at (202.5:\R) {};
    \node[r]  (u6) at (247.5:\R) {};
    \node[r]  (u7) at (292.5:\R) {};
    \node[b]  (u8) at (337.5:\R) {};
    \node[gb]  (x1) at (0,1.5*\R) {};
    \node[gr]  (x2) at (0,-1.5*\R) {};
    \node[gr]  (x3) at (-1*\R,-1.1*\R) {};
    \draw (u1)--(u2)--(x1)--(u3)--(u4)--(u5)--(u6)--(u7)--(u8)--(u1);
    \draw (u1)--(v2)--(u3)(u2)--(v3)--(u4)(u3)--(v4)--(u5)(u4)--(v5)--(u6)(u5)--(x3)--(v6)--(u7)(u6)--(v7)--(u8)(u7)--(v8)--(u1)(u8)--(v1)--(u2);
    \draw (v1)--(v4)--(v7)--(v2)--(v5)--(x2)--(v8)--(v3)--(v6)--(v1);
    \end{scope}

  \end{tikzpicture}
  
  \caption{
  The three colourings of the stc for which the top outlet is blue (we receive the remaining three by flipping colours). 
  }
  \label{fig:1?Stc}
\end{figure}
To verify property \ref{property:stc}, observe that all non-outlet vertices of the stc have degree 4. By definition of \type1?-colouring this implies the following.
\begin{claim}\label{claim:1?neighbourhoodInStc}
    In any \type1?, \type1!, \type{v}? or \type{v}!-colouring of the stc, every non-outlet vertex of the stc has exactly three neighbours of its own colour and one neighbour of the other colour. 
\end{claim}
Considering only valid \type1?-colouring of the stc (exactly those for which outlets and their neighbours have the same colour) and then applying \cref{claim:1?neighbourhoodInStc} exhaustively, it is easy to verify that property \cref{property:stc} holds. Also, observe that the stc and tac are compatible and the truth-assignment colourings of the tac and stc are compatible implying property \ref{property:combiningSTCandTAC}. We obtain that \type1?-$2$-colouring are \NP-complete by \cref{lem:reductionFromNAE3SAT}.
\end{proof} 
\begin{corollary}\label{cor2:1!}\label{link2:1!}\label{link2:v?}\label{link2:v!}
    For $q=2$ it is \NP-complete to decide whether a graph admits
    a  \type1!, \type{v}? or a \type{v}!-$q$-colouring.
\end{corollary}
\begin{proof}
    Recall that the graph $H_\varphi$ constructed in \cref{thm:21?} has the property that every vertex has degree $2$ or $4$. It is easy to verify that on such graphs \type1?-$2$-colouring, \type1!-$2$-colouring, \type{v}?-$2$-colouring and \type{v}!-$2$-colouring coincide. Hence, the reduction from \cref{thm:21?} shows that \type1!, \type{v}? and a \type{v}!-$2$-colouring are \NP-complete.
\end{proof}

\subsubsection{\Type?1 for \boldmath$2$ colours.}
First recall that $c:V(G)\rightarrow [q]$ is a \type{?}{1}-$q$-colouring of $G$  if for every vertex $v$ we have $|N_G(v)\cap c^{-1}(c(v))|<2$ and  $|N_G(v)\cap c^{-1}(i)|$ is odd for every $i\not=c(v)$. We show the following.

\begin{theorem}\label{link2:?1}
    For $q=2$ it is \NP-complete to decide whether a graph admits
    a \type?1-$q$-colouring.
\end{theorem}
\begin{proof}
    We provide a reduction from \NSAT. It is straightforward to observe the following essential property for our reduction.
    \begin{claim}\label{claim:K4?1}
        Let $G$ be a graph and $c_{\type?1}:V(G)\rightarrow \{1,2\}$ a \type?1-$2$-colouring of $G$.  In any induced $K_4$ of $G$, each colour appears twice. Furthermore, for any vertex $u$ in an induced $K_4$ and neighbour $v$ of $u$ not contained in the same $K_4$ it holds that $c_{\type?1}(u)\neq c_{\type?1}(v)$.
    \end{claim}
    Presume $\varphi$ is an instance of \NSAT. The tac of a variable $x$ appearing in $\varphi$ is given in \cref{fig:?1TAC}. True is encoded by the two vertices of an outlet having the same colour and their neighbours having the other colour and false is encoded by the two vertices of an outlet having different colours and their neighbours having the opposite different colours.
    \begin{figure}[hbtp]
    \centering
    \begin{tikzpicture}[scale=0.39]
        \node[r,label={[label distance=0.05cm]180:$v_x$}](V1) at (0,0){};
        \node[r](V2) at (3,0){};
        \node[b](PV2) at (4,-1){};
        \draw (V2)--(PV2);
        \node[b,label={[label distance=0.05cm]180:$u_x$}](V3) at (0,3){};
        \node[r](PV3) at (0,4.5){};
        \draw (V3)--(PV3);
        \node[b](V4) at (3,3){};
        \node[r](v1) at (1.5,2.1){};
        \node[r](v2) at (0.9,1.5){};
        \node[b](v3) at (1.5,0.9){};
        \node[b](v4) at (2.1,1.5){};
        \draw (v1)--(v2)--(v3)--(v4)--(v1)--(v3)--(v2)--(v4)(v1)--(V3)--(v2)(v3)--(V2)--(v4);
        \draw (V1)--(V2)--(V4)--(V3)--(V1)to[bend left=45](V4)(V2)to[bend left=45](V3);
        \node[b](U1) at (5,0){};
        \node[b](U2) at (8,0){};
        \node[r](U3) at (5,3){};
        \node[r](U4) at (8,3){};
        \draw (V2)--(U1)--(U2)--(U3)--(U4)--(U1)--(U3)--(V4)(U2)--(U4);
        \node[r](W1) at (10,0){};
        \node[r](W2) at (13,0){};
        \node[b](W3) at (10,3){};
        \node[b](W4) at (13,3){};
        \draw (U2)--(W1)--(W2)--(W3)--(W4)--(W1)--(W3)--(U4)(W2)--(W4);
        \node[b](X1) at (15,0){};
        \node[b](X2) at (18,0){};
        \node[r](X3) at (15,3){};
        \node[r](X4) at (18,3){};
        \draw (W2)--(X1)--(X2)--(X3)--(X4)--(X1)--(X3)--(W4)(X2)--(X4);
        \node[r](Y1) at (20,0){};
        \node[r](Y2) at (23,0){};
        \node[b](Y3) at (20,3){};
        \node[b](Y4) at (23,3){};
        \draw (X2)--(Y1)--(Y2)--(Y3)--(Y4)--(Y1)--(Y3)--(X4)(Y2)--(Y4);
        \node[b](Z1) at (25,0){};
        \node[b](Z2) at (28,0){};
        \node[r](Z3) at (25,3){};
        \node[r](Z4) at (28,3){};
        \draw (Y2)--(Z1)--(Z2)--(Z3)--(Z4)--(Z1)--(Z3)--(Y4)(Z2)--(Z4);
        \node[r](A1) at (30,0){};
        \node[r](A2) at (33,0){};
        \node[b](PA2) at (33,-1.5){};
        \draw (A2)--(PA2);
        \node[b](A3) at (30,3){};
        \node[b](A4) at (33,3){};
        \node[r](PA4) at (33,4.5){};
        \draw (A4)--(PA4);
        \draw (Z2)--(A1)--(A2)--(A3)--(A4)--(A1)--(A3)--(Z4)(A2)--(A4);
        \node[gb](O1)at (0,-3){};
        \node[gb](O2)at (3,-3){};
        \draw (O1)--(V1)(O2)--(V2);
        \node[gb](O3)at (10,-3){};
        \node[gb](O4)at (13,-3){};
        \draw (O3)--(W1)(O4)--(W2);
        \node[b](PW1) at (9,-1){};
        \draw (W1)--(PW1);
        \node[b](PW2) at (14,-1){};
        \draw (W2)--(PW2);
        \node[gb](O5)at (20,-3){};
        \node[gr](O6)at (23,-3){};
        \draw (O5)--(Y1)(O6)--(Y3);
        \node[b](PY1) at (19,-1){};
        \draw (Y1)--(PY1);
        \node[b](PY2) at (24,-1){};
        \draw (Y2)--(PY2);
        \node[gb](O7)at (30,-3){};
        \node[gr](O8)at (33,-3){};
        \draw (O7)--(A1)(O8)--(A3);
        \node[r](PA3) at (30,4.5){};
        \draw (A3)--(PA3);
        \node[b](PA1) at (29,-1){};
        \draw (A1)--(PA1);

        \begin{scope}[yshift=-9.8cm]
        \node[r,label={[label distance=0.05cm]180:$v_x$}](V1) at (0,0){};
        \node[b](V2) at (3,0){};
        \node[r](PV2) at (4,-1){};
        \draw (V2)--(PV2);
        \node[r,label={[label distance=0.05cm]180:$u_x$}](V3) at (0,3){};
        \node[b](PV3) at (0,4.5){};
        \draw (V3)--(PV3);
        \node[b](V4) at (3,3){};
        \node[b](v1) at (1.5,2.1){};
        \node[b](v2) at (0.9,1.5){};
        \node[r](v3) at (1.5,0.9){};
        \node[r](v4) at (2.1,1.5){};
        \draw (v1)--(v2)--(v3)--(v4)--(v1)--(v3)--(v2)--(v4)(v1)--(V3)--(v2)(v3)--(V2)--(v4);
        \draw (V1)--(V2)--(V4)--(V3)--(V1)to[bend left=45](V4)(V2)to[bend left=45](V3);
        \node[r](U1) at (5,0){};
        \node[b](U2) at (8,0){};
        \node[r](U3) at (5,3){};
        \node[b](U4) at (8,3){};
        \draw (V2)--(U1)--(U2)--(U3)--(U4)--(U1)--(U3)--(V4)(U2)--(U4);
        \node[r](W1) at (10,0){};
        \node[b](W2) at (13,0){};
        \node[r](W3) at (10,3){};
        \node[b](W4) at (13,3){};
        \draw (U2)--(W1)--(W2)--(W3)--(W4)--(W1)--(W3)--(U4)(W2)--(W4);
        \node[r](X1) at (15,0){};
        \node[b](X2) at (18,0){};
        \node[r](X3) at (15,3){};
        \node[b](X4) at (18,3){};
        \draw (W2)--(X1)--(X2)--(X3)--(X4)--(X1)--(X3)--(W4)(X2)--(X4);
        \node[r](Y1) at (20,0){};
        \node[b](Y2) at (23,0){};
        \node[r](Y3) at (20,3){};
        \node[b](Y4) at (23,3){};
        \draw (X2)--(Y1)--(Y2)--(Y3)--(Y4)--(Y1)--(Y3)--(X4)(Y2)--(Y4);
        \node[r](Z1) at (25,0){};
        \node[b](Z2) at (28,0){};
        \node[r](Z3) at (25,3){};
        \node[b](Z4) at (28,3){};
        \draw (Y2)--(Z1)--(Z2)--(Z3)--(Z4)--(Z1)--(Z3)--(Y4)(Z2)--(Z4);
        \node[r](A1) at (30,0){};
        \node[b](A2) at (33,0){};
        \node[r](PA2) at (33,-1.5){};
        \draw (A2)--(PA2);
        \node[r](A3) at (30,3){};
        \node[b](A4) at (33,3){};
        \node[r](PA4) at (33,4.5){};
        \draw (A4)--(PA4);
        \draw (Z2)--(A1)--(A2)--(A3)--(A4)--(A1)--(A3)--(Z4)(A2)--(A4);
        \node[gb](O1)at (0,-3){};
        \node[gr](O2)at (3,-3){};
        \draw (O1)--(V1)(O2)--(V2);
        \node[gb](O3)at (10,-3){};
        \node[gr](O4)at (13,-3){};
        \draw (O3)--(W1)(O4)--(W2);
        \node[b](PW1) at (9,-1){};
        \draw (W1)--(PW1);
        \node[r](PW2) at (14,-1){};
        \draw (W2)--(PW2);
        \node[gb](O5)at (20,-3){};
        \node[gb](O6)at (23,-3){};
        \draw (O5)--(Y1)(O6)--(Y3);
        \node[b](PY1) at (19,-1){};
        \draw (Y1)--(PY1);
        \node[r](PY2) at (24,-1){};
        \draw (Y2)--(PY2);
        \node[gb](O7)at (30,-3){};
        \node[gb](O8)at (33,-3){};
        \draw (O7)--(A1)(O8)--(A3);
        \node[b](PA3) at (30,4.5){};
        \draw (A3)--(PA3);
        \node[b](PA1) at (29,-1){};
        \draw (A1)--(PA1);
        \end{scope}
    \end{tikzpicture}
    \caption{The two (up to permutation of colours) colourings of the tac for a variable $x$ which occurs $4$ times, twice non-negated and twice negated.}\label{fig:?1TAC}
\end{figure}
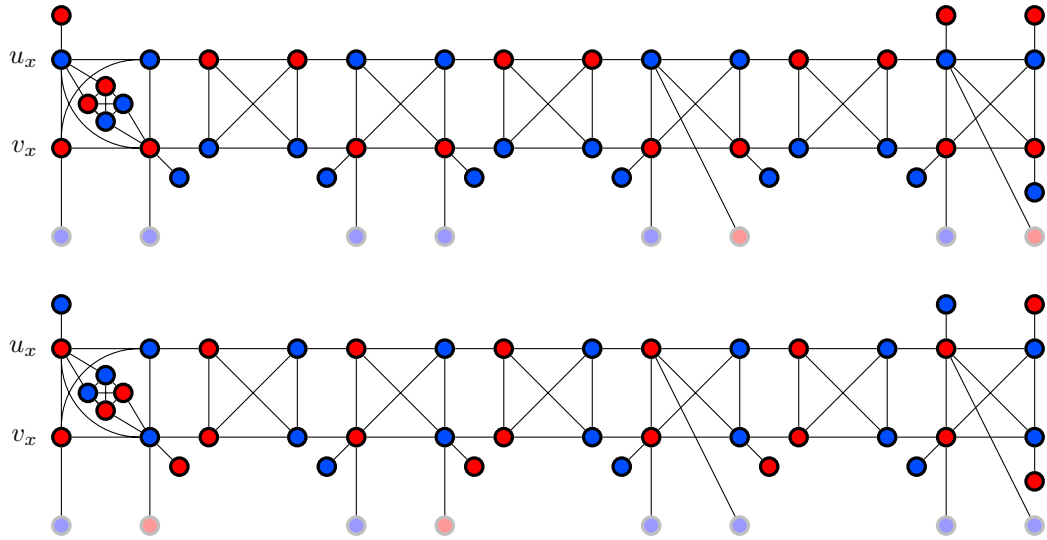
To verify that the tac of some variable $x$ satisfies property \ref{property:tac}, observe that choosing the colours for $u_x$ and $v_x$ (marked in \cref{fig:?1TAC}) to be either the same or different fixes the colours for the rest of the tac by applying \cref{claim:K4?1} successively. Hence, up to permutation of colours there are two colourings.

The stc of some clause $c$ of $\varphi$, is given in \cref{fig:?1STC}. 
\begin{figure}[hbtp]
    \centering
    \begin{tikzpicture}[scale=0.27]
        \node[r](V1) at (0,0){};
        \node[r](V2) at (3,0){};
        \node[b](V3) at (0,3){};
        \node[b](V4) at (3,3){};
        \draw (V1)--(V2)--(V4)--(V3)--(V1)--(V4)(V2)--(V3);
        \node[r](U1) at (6,0){};
        \node[r](U2) at (9,0){};
        \node[b](U3) at (6,3){};
        \node[b](U4) at (9,3){};
        \draw (U1)--(U2)--(U3)--(U4)--(U1)--(U3)(U2)--(U4);
        \node[gr](O1)at (0,6){};
        \node[gb](O2)at (1.5,6){};
        \draw (O1)--(V4)(O2)to[bend left=15](U1);
        \node[gr](O3)at (3.75,6){};
        \node[gr](O4)at (5.25,6){};
        \draw (O3)--(V4)(O4)--(U3);
        \node[gr](O5)at (7.5,6){};
        \node[gr](O6)at (9,6){};
        \draw (O5)--(V4)(O6)--(U4);
        \node[r](P1)at (-3,-3){};
        \node[b](P2)at (0,-3){};
        \node[b](P3)at (3,-3){};
        \node[b](P4)at (9,-3){};
        \draw (P1)--(V3)(P2)--(V1)(P3)--(V2)(P4)--(U2);

        \begin{scope}[xshift=17cm]
        \node[r](V1) at (0,0){};
        \node[r](V2) at (3,0){};
        \node[b](V3) at (0,3){};
        \node[b](V4) at (3,3){};
        \draw (V1)--(V2)--(V4)--(V3)--(V1)--(V4)(V2)--(V3);
        \node[r](U1) at (6,0){};
        \node[b](U2) at (9,0){};
        \node[b](U3) at (6,3){};
        \node[r](U4) at (9,3){};
        \draw (U1)--(U2)--(U3)--(U4)--(U1)--(U3)(U2)--(U4);
        \node[gr](O1)at (0,6){};
        \node[gb](O2)at (1.5,6){};
        \draw (O1)--(V4)(O2)to[bend left=15](U1);
        \node[gr](O3)at (3.75,6){};
        \node[gr](O4)at (5.25,6){};
        \draw (O3)--(V4)(O4)--(U3);
        \node[gr](O5)at (7.5,6){};
        \node[gb](O6)at (9,6){};
        \draw (O5)--(V4)(O6)--(U4);
        \node[r](P1)at (-3,-3){};
        \node[b](P2)at (0,-3){};
        \node[b](P3)at (3,-3){};
        \node[r](P4)at (9,-3){};
        \draw (P1)--(V3)(P2)--(V1)(P3)--(V2)(P4)--(U2);
        \end{scope}

        \begin{scope}[xshift=34cm]
        \node[r](V1) at (0,0){};
        \node[r](V2) at (3,0){};
        \node[b](V3) at (0,3){};
        \node[b](V4) at (3,3){};
        \draw (V1)--(V2)--(V4)--(V3)--(V1)--(V4)(V2)--(V3);
        \node[r](U1) at (6,0){};
        \node[b](U2) at (9,0){};
        \node[r](U3) at (6,3){};
        \node[b](U4) at (9,3){};
        \draw (U1)--(U2)--(U3)--(U4)--(U1)--(U3)(U2)--(U4);
        \node[gr](O1)at (0,6){};
        \node[gb](O2)at (1.5,6){};
        \draw (O1)--(V4)(O2)to[bend left=15](U1);
        \node[gr](O3)at (3.75,6){};
        \node[gb](O4)at (5.25,6){};
        \draw (O3)--(V4)(O4)--(U3);
        \node[gr](O5)at (7.5,6){};
        \node[gr](O6)at (9,6){};
        \draw (O5)--(V4)(O6)--(U4);
        \node[r](P1)at (-3,-3){};
        \node[b](P2)at (0,-3){};
        \node[b](P3)at (3,-3){};
        \node[r](P4)at (9,-3){};
        \draw (P1)--(V3)(P2)--(V1)(P3)--(V2)(P4)--(U2);
        \end{scope}

        \begin{scope}[xshift=0cm,yshift=-14cm]
        \node[r](V1) at (0,0){};
        \node[r](V2) at (3,0){};
        \node[b](V3) at (0,3){};
        \node[b](V4) at (3,3){};
        \draw (V1)--(V2)--(V4)--(V3)--(V1)--(V4)(V2)--(V3);
        \node[b](U1) at (6,0){};
        \node[b](U2) at (9,0){};
        \node[r](U3) at (6,3){};
        \node[r](U4) at (9,3){};
        \draw (U1)--(U2)--(U3)--(U4)--(U1)--(U3)(U2)--(U4);
        \node[gr](O1)at (0,6){};
        \node[gr](O2)at (1.5,6){};
        \draw (O1)--(V4)(O2)to[bend left=15](U1);
        \node[gr](O3)at (3.75,6){};
        \node[gb](O4)at (5.25,6){};
        \draw (O3)--(V4)(O4)--(U3);
        \node[gr](O5)at (7.5,6){};
        \node[gb](O6)at (9,6){};
        \draw (O5)--(V4)(O6)--(U4);
        \node[r](P1)at (-3,-3){};
        \node[b](P2)at (0,-3){};
        \node[b](P3)at (3,-3){};
        \node[r](P4)at (9,-3){};
        \draw (P1)--(V3)(P2)--(V1)(P3)--(V2)(P4)--(U2);
        \end{scope}

        \begin{scope}[xshift=17cm,yshift=-14cm]
        \node[r](V1) at (0,0){};
        \node[r](V2) at (3,0){};
        \node[b](V3) at (0,3){};
        \node[b](V4) at (3,3){};
        \draw (V1)--(V2)--(V4)--(V3)--(V1)--(V4)(V2)--(V3);
        \node[b](U1) at (6,0){};
        \node[r](U2) at (9,0){};
        \node[b](U3) at (6,3){};
        \node[r](U4) at (9,3){};
        \draw (U1)--(U2)--(U3)--(U4)--(U1)--(U3)(U2)--(U4);
        \node[gr](O1)at (0,6){};
        \node[gr](O2)at (1.5,6){};
        \draw (O1)--(V4)(O2)to[bend left=15](U1);
        \node[gr](O3)at (3.75,6){};
        \node[gr](O4)at (5.25,6){};
        \draw (O3)--(V4)(O4)--(U3);
        \node[gr](O5)at (7.5,6){};
        \node[gb](O6)at (9,6){};
        \draw (O5)--(V4)(O6)--(U4);
        \node[r](P1)at (-3,-3){};
        \node[b](P2)at (0,-3){};
        \node[b](P3)at (3,-3){};
        \node[b](P4)at (9,-3){};
        \draw (P1)--(V3)(P2)--(V1)(P3)--(V2)(P4)--(U2);
        \end{scope}

        \begin{scope}[xshift=34cm,yshift=-14cm]
        \node[r](V1) at (0,0){};
        \node[r](V2) at (3,0){};
        \node[b](V3) at (0,3){};
        \node[b](V4) at (3,3){};
        \draw (V1)--(V2)--(V4)--(V3)--(V1)--(V4)(V2)--(V3);
        \node[b](U1) at (6,0){};
        \node[r](U2) at (9,0){};
        \node[r](U3) at (6,3){};
        \node[b](U4) at (9,3){};
        \draw (U1)--(U2)--(U3)--(U4)--(U1)--(U3)(U2)--(U4);
        \node[gr](O1)at (0,6){};
        \node[gr](O2)at (1.5,6){};
        \draw (O1)--(V4)(O2)to[bend left=15](U1);
        \node[gr](O3)at (3.75,6){};
        \node[gb](O4)at (5.25,6){};
        \draw (O3)--(V4)(O4)--(U3);
        \node[gr](O5)at (7.5,6){};
        \node[gr](O6)at (9,6){};
        \draw (O5)--(V4)(O6)--(U4);
        \node[r](P1)at (-3,-3){};
        \node[b](P2)at (0,-3){};
        \node[b](P3)at (3,-3){};
        \node[b](P4)at (9,-3){};
        \draw (P1)--(V3)(P2)--(V1)(P3)--(V2)(P4)--(U2);
        \end{scope}
    \end{tikzpicture}
    \caption{The six \type?1-$2$-colouring of the stc corresponding to the assignments of literals of  $c$.}
    \label{fig:?1STC}
\end{figure}
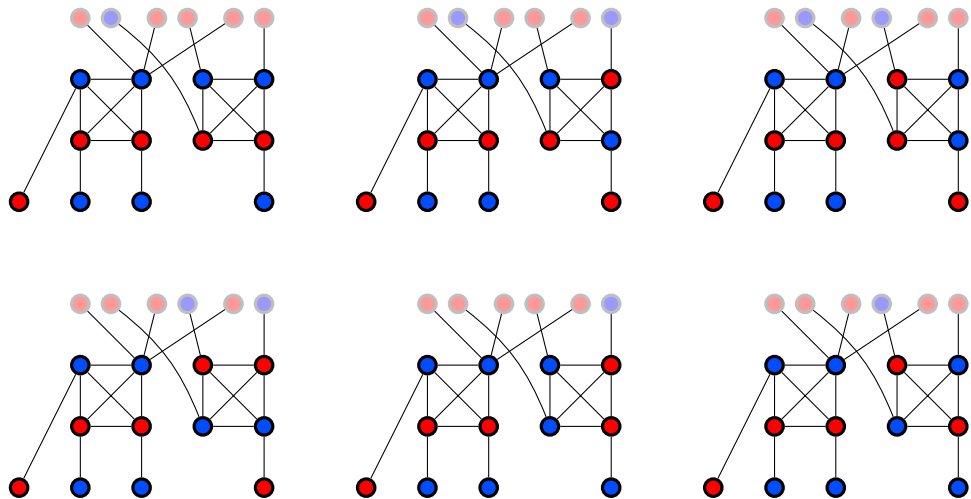
Using \cref{claim:K4?1} it follows that the three outlet vertices, one from each outlet, that share a neighbour have to receive the same colour in a \type?1-$2$-colouring. The three other outlet vertices are connected to three different vertices of a $K_4$ and therefore by \cref{claim:K4?1} cannot all receive the same colour. Therefore, there must be an outlet for which both vertices receive the same colour and one for which both vertices receive different colours. It follows that the stc satisfies property \ref{property:stc}.

Finally, the stc and tac are compatible. Furthermore, fixing colouring for true and false of the tac as in \cref{fig:?1TAC} and the $6$ colourings depicted in \cref{fig:?1STC} for the $6$ different assignments of the stc, the appropriate colourings of the two gadgets are compatible. Therefore, the stc and tac satisfy property \ref{property:combiningSTCandTAC}. By \cref{lem:reductionFromNAE3SAT} we obtain that \type?1-$2$-colouring is \NP-complete.
\end{proof}

\subsubsection{\Type{+}{+}-colouring for \boldmath$q=2$}

First recall that $c:V(G)\rightarrow [q]$ is a \type++-$q$-colouring of $G$ if for every vertex $v$ we have $|N_G(v)\cap c^{-1}(i)|>0$ for every $i\in[q]$. We show the following.

\begin{theorem}\label{link2:++}
    For $q=2$ it is \NP-complete to decide whether a graph admits
    a \type++-$q$-colouring.
\end{theorem}
\begin{proof}
We provide a reduction from \NSAT. First observe the following.
\begin{claim}\label{claim:++degree2}
    In any \type++-$2$-colouring every vertex of degree $2$ has a neighbour of each of the two colours.
\end{claim}
Using this property, we construct the tac of a variable $x$ that appears $k$ times positively and $k$ times negatively the instance $\varphi$ of \NSAT from a cycle of length $4k$ in which we attach a pendant to every second vertex. Additionally, we attach an outlet consisting of a single vertex to each of the pendant. 
    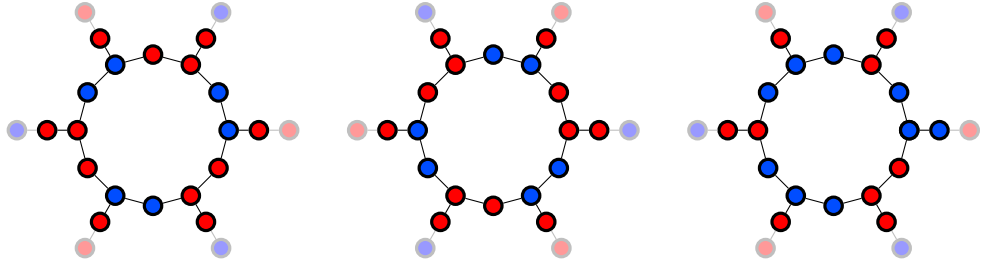
\begin{figure}[hbtp]
    \hspace*{\fill}
    \begin{tikzpicture}
    \def \r {1cm}
    \def \rr {1.4cm}
    \def \rrr {1.8cm}
    \def \rrrr {2.2cm}
    \foreach \n in {1,2,5,6,9,10} {
      \pgfmathsetmacro{\a}{(\n-3)*360/12}
      \node[r] (a\n) at (\a:\r) {};
    }
    \foreach \n in {3,4,7,8,11,12} {
      \pgfmathsetmacro{\a}{(\n-3)*360/12}
      \node[b] (a\n) at (\a:\r) {};
    }
    \foreach[count=\m] \n in {2,...,12} {
      \draw (a\n) -- (a\m);
    }
    \draw (a12) -- (a1);
    
    \foreach \n in {1,3,5,7,9,11} {
      \pgfmathsetmacro{\a}{(\n-3)*360/12}
      \node[r] (x\n) at (\a:\rr) {};
    }
    \foreach \n in {1,5,9} {
      \pgfmathsetmacro{\a}{(\n-3)*360/12}
      \node[gb] (y\n) at (\a:\rrr) {};
    }
    \foreach \n in {3,7,11} {
      \pgfmathsetmacro{\a}{(\n-3)*360/12}
      \node[gr] (y\n) at (\a:\rrr) {};
    }
    \foreach \n in {1,3,5,7,9,11} {
      \draw (a\n)--(x\n);
      \draw[lightgray](x\n)--(y\n);
    }

    \begin{scope}[xshift=45mm]
    
    \foreach \n in {1,2,5,6,9,10} {
      \pgfmathsetmacro{\a}{(\n-3)*360/12}
      \node[b] (a\n) at (\a:\r) {};
    }
    \foreach \n in {3,4,7,8,11,12} {
      \pgfmathsetmacro{\a}{(\n-3)*360/12}
      \node[r] (a\n) at (\a:\r) {};
    }
    \foreach[count=\m] \n in {2,...,12} {
      \draw (a\n) -- (a\m);
    }
    \draw (a12) -- (a1);
    
    \foreach \n in {1,3,5,7,9,11} {
      \pgfmathsetmacro{\a}{(\n-3)*360/12}
      \node[r] (x\n) at (\a:\rr) {};
    }
    \foreach \n in {1,5,9} {
      \pgfmathsetmacro{\a}{(\n-3)*360/12}
      \node[gr] (y\n) at (\a:\rrr) {};
    }
    \foreach \n in {3,7,11} {
      \pgfmathsetmacro{\a}{(\n-3)*360/12}
      \node[gb] (y\n) at (\a:\rrr) {};
    }
    \foreach \n in {1,3,5,7,9,11} {
      \draw (a\n)--(x\n);
      \draw[lightgray](x\n)--(y\n);
    }
    \end{scope}

    \begin{scope}[xshift=90mm]
      \node[a] (a0) at (0:\r) {};
      \foreach[count=\n, count=\m from 0] \c in {b,r,b,b,b,r,b,b,b,r,r,b} {
        \pgfmathtruncatemacro{\a}{30*\n}
        \node[\c] (a\n) at (\a:\r)   {}; \draw (a\n)--(a\m);
      }
      \foreach \n/\c/\g in {0/b/gr, 2/r/gb, 4/r/gr, 6/r/gb, 8/r/gr, 10/r/gb} {
        \pgfmathtruncatemacro{\a}{30*\n}
        \node[\c] (x\n) at (\a:\rr)  {}; \draw (a\n)--(x\n);
        \node[\g] (y\n) at (\a:\rrr) {}; \draw[lightgray] (x\n)--(y\n);
      }
    \end{scope}
  \end{tikzpicture}
  \hspace*{\fill}
  \caption{Three possible colourings of the tac of a variable that appears three times positively and three times negatively.}
  \label{fig:++tac}
\end{figure}
Observe that by \cref{claim:++degree2} the vertices of the cycle of degree $3$ alternate in colour. Since the pedants attached to these vertices have degree $2$ again, \cref{claim:++degree2} implies that the outlets alternate in colour as well. We associate every second outlet (along the cycle) with an occurrence of $x$ and the remaining outlets with an occurrence of $\overline{x}$. True is then encoded by the outlet having colour $1$ and false by the outlet having colour $2$. Observe that we can get a \type++-$2$-colouring of the tac e.g. as in \cref{fig:++tac}. Observe that property \ref{property:tac} holds for the tac.

We now construct the stc of a clause $c$ from a claw, by attaching an outlet consisting of a single vertex to each of the degree $1$ vertices of the claw. We encode a literal being true by the corresponding outlet having colour $1$ and its neighbour having colour $1$, while false is encoded by the outlet having colour $1$ and its neighbour having colour $2$. 
\begin{figure}[hbtp]
    \hspace*{\fill}
    \begin{tikzpicture}
    \def \r {0.5cm};
    \def \rr {1cm};
    \node[b] (m) at (0,0){};
    \foreach \n in {1,2} {
      \pgfmathsetmacro{\a}{(\n-3)*360/3}
      \node[r] (a\n) at (\a:\r) {};
      \node[gr] (x\n) at (\a:\rr) {};
    }
    \foreach \n in {3} {
      \pgfmathsetmacro{\a}{(\n-3)*360/3}
      \node[b] (a\n) at (\a:\r) {};
      \node[gr] (x\n) at (\a:\rr) {};
    }
    \foreach \n in {1,2,3} {
      \draw (a\n) -- (m);
      \draw[lightgray] (a\n)--(x\n);
    }
    
    \begin{scope}[xshift=2.1cm]
    \node[b] (m) at (0,0){};
    \foreach \n in {1} {
      \pgfmathsetmacro{\a}{(\n-3)*360/3}
      \node[r] (a\n) at (\a:\r) {};
      \node[gr] (x\n) at (\a:\rr) {};
    }
    \foreach \n in {2,3} {
      \pgfmathsetmacro{\a}{(\n-3)*360/3}
      \node[b] (a\n) at (\a:\r) {};
      \node[gr] (x\n) at (\a:\rr) {};
    }
    \foreach \n in {1,2,3} {
      \draw (a\n) -- (m);
      \draw[lightgray] (a\n)--(x\n);
    }
    \end{scope}

    \begin{scope}[xshift=4.2cm]
    \node[b] (m) at (0,0){};
    \foreach \n in {1,3} {
      \pgfmathsetmacro{\a}{(\n-3)*360/3}
      \node[r] (a\n) at (\a:\r) {};
      \node[gr] (x\n) at (\a:\rr) {};
    }
    \foreach \n in {2} {
      \pgfmathsetmacro{\a}{(\n-3)*360/3}
      \node[b] (a\n) at (\a:\r) {};
      \node[gr] (x\n) at (\a:\rr) {};
    }
    \foreach \n in {1,2,3} {
      \draw (a\n) -- (m);
      \draw[lightgray] (a\n)--(x\n);
    }
    \end{scope}

    \begin{scope}[xshift=6.3cm]
    \node[b] (m) at (0,0){};
    \foreach \n in {2,3} {
      \pgfmathsetmacro{\a}{(\n-3)*360/3}
      \node[r] (a\n) at (\a:\r) {};
      \node[gr] (x\n) at (\a:\rr) {};
    }
    \foreach \n in {1} {
      \pgfmathsetmacro{\a}{(\n-3)*360/3}
      \node[b] (a\n) at (\a:\r) {};
      \node[gr] (x\n) at (\a:\rr) {};
    }
    \foreach \n in {1,2,3} {
      \draw (a\n) -- (m);
      \draw[lightgray] (a\n)--(x\n);
    }
    \end{scope}

    \begin{scope}[xshift=8.4cm]
    \node[b] (m) at (0,0){};
    \foreach \n in {2} {
      \pgfmathsetmacro{\a}{(\n-3)*360/3}
      \node[r] (a\n) at (\a:\r) {};
      \node[gr] (x\n) at (\a:\rr) {};
    }
    \foreach \n in {1,3} {
      \pgfmathsetmacro{\a}{(\n-3)*360/3}
      \node[b] (a\n) at (\a:\r) {};
      \node[gr] (x\n) at (\a:\rr) {};
    }
    \foreach \n in {1,2,3} {
      \draw (a\n) -- (m);
      \draw[lightgray] (a\n)--(x\n);
    }
    \end{scope}

    \begin{scope}[xshift=10.5cm]
    \node[b] (m) at (0,0){};
    \foreach \n in {3} {
      \pgfmathsetmacro{\a}{(\n-3)*360/3}
      \node[r] (a\n) at (\a:\r) {};
      \node[gr] (x\n) at (\a:\rr) {};
    }
    \foreach \n in {1,2} {
      \pgfmathsetmacro{\a}{(\n-3)*360/3}
      \node[b] (a\n) at (\a:\r) {};
      \node[gr] (x\n) at (\a:\rr) {};
    }
    \foreach \n in {1,2,3} {
      \draw (a\n) -- (m);
      \draw[lightgray] (a\n)--(x\n);
    }
    \end{scope}
  \end{tikzpicture}
  \hspace*{\fill}
  \caption{The six valid \type++-$2$-colourings of the stc corresponding to the six possible assignments. }
  \label{fig:++stc}
\end{figure}
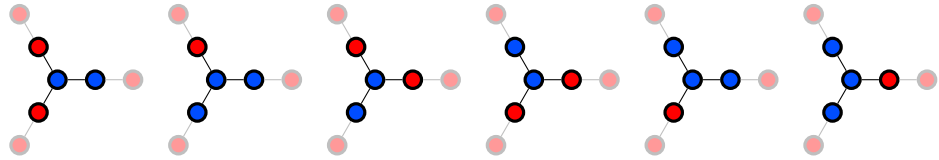
\end{proof}
Note that the central vertex of the claw can be coloured correctly in a \type++-$2$-colouring if and only if the neighbours of the outlets do not all have the same colour. Additionally, for the arbitrary choice of the outlets all having colour $1$ the centre of the claw has to receive colour $2$ (but this can be assigned differently) implying property \ref{property:stc}. After fixing the colourings of the tac as given in \cref{fig:++tac} and the colouring of the stc as given in \cref{fig:++stc}, it is easy to observe that property \ref{property:combiningSTCandTAC} holds. Hence, \cref{lem:reductionFromNAE3SAT} implies that \type++-$2$-colouring is \NP-complete.

\subsection{Reductions from other problems}

\subsubsection{\Type{?}{v}-colouring}

Recall that $c:V(G)\rightarrow [q]$ is a \type?v-$q$-colouring of $G$  if for every vertex $v$ it holds that $|N_G(v)\cap c^{-1}(c(v))|< 2$ and $|N_G(v)\cap c^{-1}(i)|$ is odd or $0$ for every $i\neq c(v)$. 

\begin{theorem}\label{thm:?v}\label{link234:?v}
    For $q \geq 2$ it is \NP-complete to decide whether a graph admits a \type?v-$q$-colouring.
\end{theorem}

\begin{proof}
    We give a reduction from defective-$q$-colouring with defect $d = 1$, which is \NP-complete \cite{cowen1997defective} for $q\geq 2$ colours.  
    Given a graph $G$, we construct a graph $H$ by attaching $q$ pendants to every vertex of even degree and $q-1$ pendants to every vertex of odd degree.
    
    First, 
    let $c:V(G) \rightarrow [q]$ be a defective colouring of $G$. We obtain a \type?v-colouring $c_{\type{?}{v}}:V(H) \rightarrow [q]$ of $H$ as follows. For each vertex $v \in V(G)$, we set $c_{\type{?}{v}}(v) = c(v)$. Now, fix a vertex $v \in V(G)$.
    If $v$ has even degree, let $v_1,\dots,v_q \in V(H)$ be the pendants attached to $v$ in $H$. 
    Since the degree of $v$ is even, there is an even number of colour classes containing an odd number of neighbours of $v$. Hence, there is an even number, if $q$ is even, or an odd number, if $q$ is odd, of colour classes containing an even number of neighbours of $v$. For each $i\in [q]$, if $v$ has an even number of neighbours of colour $i$, we set $c_{\type{?}{v}}(v_{i})=i$. The remaining even number of pendants get assigned to an arbitrary colour not equal to $c(v)$ (note that since $q\geq 2$ such a colour exists). Hence, $v$ has an odd number of neighbours of each colour in $H$. Furthermore, if $v$ had a neighbour of its own colour in $G$, then the colour class of colour $c(v)$ is odd and therefore no pendant receives colour $c(v)$ implying that $|N_H(v)\cap c_{\type?v}^{-1}(c_{\type?v}(v))|<2$. On the other hand, if $v$ has no neighbour of its own colour in $G$, then the colour class of $c(v)$ is even and one pendant of $v$ receives colour $c(v)$ implying $|N_H(v)\cap c_{\type?v}^{-1}(c_{\type?v}(v))|<2$. Finally, each of the pendants has one neighbour in some colour which satisfies the conditions of \type?v-colouring. 

    If $v$ has odd degree, let $v_1,\dots,v_{q-1} \in V(H)$ be the pendants attached to $v$ in $H$. 
    Since the degree of $v$ is odd, there is an odd number of colour classes containing an odd number of neighbours of $v$. In particular, there must be at least one colour class containing an odd number of neighbours of $v$ and, without loss of generality, assume the respective colour is $q$. We additionally obtain that there is an even number, if $q-1$ is even, or an odd number, if $q-1$ is odd,  of colour classes containing an even number of neighbours of $v$. For each $i\in [q-1]$, if $v$ has an even number of neighbours of colour $i$, we set $c_{\type{?}{v}}(v_{i})=i$. The remaining even number of pendants get assigned to an arbitrary colour not equal to $c(v)$. Similarly to before we can justify that this is a valid \type?v-colouring.

    On the other hand, 
    let $c_{\type{?}{v}}:V(H) \rightarrow [q]$ be a \type?v-colouring of $H$. We get a defective colouring $c:V(G) \rightarrow [q]$ of $G$ as follows. For each vertex $v \in V(G)$, we set $c(v) = c_{\type{?}{v}}(v)$. It follows from the definition of \type?v-colouring that $c$ is a defective colouring with defect $1$.
\end{proof}


\subsubsection{\Type{?}{+}-colouring}

Recall that $c:V(G)\rightarrow [q]$ is a $\type{?}{+}$-$q$-colouring of $G$ if for every vertex $v$ it holds that $|N_G(v)\cap c^{-1}(c(v))|< 2$ and $|N_G(v)\cap c^{-1}(i)|>0$ for every $i\neq c(v)$.

\begin{theorem}\label{link234:?+}
    For $q \geq 2$ it is \NP-complete to decide whether a graph admits a \type{?}{+}-$q$-colouring.
\end{theorem}

\begin{proof}
    We give a reduction from defective-$q$-colouring with defect $d = 1$, which is \NP-complete \cite{cowen1997defective} for $q\geq 2$ colours. Given a graph $G$, we construct a graph $H$ by adding, for each vertex $v \in V(G)$, a clique $K_{q-1}^v$ with vertices $w_1^v,\ldots,w_{q-1}^v$, and the edge $vw_i^v$ for each $i \in [q-1]$.
    
    First, assume that $G$ is a YES-instance of defective colouring, and let $c:V(G) \rightarrow [q]$ be a defective colouring of $G$. We obtain a \type{?}{+}-$q$-colouring $c_{\type{?}{+}}:V(H) \rightarrow [q]$ of $H$ as follows. For each vertex $v \in V(G)$, we set $c_{\type{?}{+}}(v) = c(v)$. Then, for every vertex $v \in V(G)$, we assign a second colour to $w_1^v$, a third colour to $w_2^v$, and so on, until we assign the last colour to $w_{q-1}^v$. It is easy to verify that $c_{\type{?}{+}}$ is a \type{?}{+}-$q$-colouring.

    On the other hand, assume $H$ is a YES-instance of \type{?}{+}-$q$-colouring and let $c_{\type{?}{+}}:V(H) \rightarrow [q]$ be a \type{?}{+}-$q$-colouring of $H$. Fix a vertex $v \in V(G)$.
    
    Observe that the vertex $w_1^v$ has $q-1$ neighbours. Therefore, by definition of \type{?}{+}-$q$-colouring, all other colours in its neighbourhood appear exactly once. In particular, for $i \in [q-1]$, the colour of $w_i^v$ is different from the colour of $v$. Moreover, the vertices $w_i^v$ and $w_j^v$ must receive different colours for each $i \neq j$.
    
    We obtain a defective colouring $c:V(G) \rightarrow [q]$ of $G$ as follows. For each vertex $v \in V(G)$, we set $c(v) = c_{\type{?}{+}}(v)$. It is easy to verify that $c$ is a defective colouring with defect 1.
\end{proof}

\begin{corollary}\label{link2:v+}
    Starting from \type{v}{*}, this reduction also works for \type{v}{+}. In particular, for $q = 2$, it follows from~\cref{thm2:v*} that it is \NP-complete to decide whether a graph admits a \type{v}{+}-colouring with $q$ colours.
\end{corollary}


\subsubsection{\Type{!}{0}-colouring}

Recall that $c:V(G)\rightarrow [q]$ is a \type{!}{0}-$q$-colouring of $G$ if for every $v\in V(G)$ we have $|N_G(v)\cap c^{-1}(c(v))|=1$ and $|N_G(v)\cap c^{-1}(i)|$ is even for every $i\not=c(v)$.

\begin{theorem}\label{link234:!0}
    For $q\geq 2$ it is \NP-complete to decide whether a graph admits
    a \type!0-$q$-colouring.
\end{theorem}

\begin{proof}
    We give a reduction from \type{!}{*}-$q$-colouring, which is \NP-complete for $2$ colours by \cite{Sch78} and for more than $2$ colours by \cite{DemaineKP25}. Given a graph $G$, we construct a graph $H$ as follows. Take two disjoint copies of $G$, called $G^1$ and $G^2$. For every vertex $v\in V(G)$ of even degree and every $i \in [2q-1]$, we add a $K_2$ with vertices $x_i^v,y_i^v$ and the edges $v^1x_i^v$ and $x_i^vv^2$, where $v^1$ and $v^2$ are the two copies of $v$. For every vertex $v\in V(G)$ of odd degree and every $i \in [2q]$, we add a $K_2$ with vertices $x_i^v,y_i^v$ and the edges $v^1x_i^v$ and $x_i^vv^2$, where $v^1$ and $v^2$ are the two copies of $v$.
    
    First, 
    let $c:V(G) \rightarrow [q]$ be a \type{!}{*}-$q$-colouring of $G$. We obtain a \type{!}{0}-$q$-colouring $c_{\type{!}{0}}:V(H) \rightarrow [q]$ of $H$ as follows. For every vertex $v \in V(G)$, we set $c_{\type{!}{0}}(v^1) = c_{\type{!}{0}}(v^2) = c(v)$. Fix a vertex $v \in V(G)$ and assume first that $v$ has even degree.
    Since the degree of $v$ is even, there is an even number of colour classes containing an odd number of neighbours of $v$. 
    For each $i \in [q] \setminus \{c(v)\}$, if $v$ has an odd number of neighbours of colour $i$, we set $c_{\type{!}{0}}(x_i^v)=c_{\type{!}{0}}(y_i^v)=i$. Fix an arbitrary colour $j$ not equal to $c(v)$ (this is always possible as $q\geq 2$). For every $i \in [2q-1]$ for which $x_i^v$ and $y_i^v$ have not yet been assigned a colour we set $c_{\type{!}{0}}(x_i^v)=c_{\type{!}{0}}(y_i^v)=j$, adding an even number of extra neighbours of $v^1$ and $v^2$ of colour $j$. Observe that $v^1$ and $v^2$ have an even number of neighbours of each colour in $H$ except their own. Furthermore, since $v$ has exactly one neighbour of colour $c(v)$ in $G$, we do not assign any $x_i^v$ the colour $c(v)$, and therefore $|N_H(v^1)\cap c_{\type{!}{0}}^{-1}(c_{\type{!}{0}}(v^1))|=|N_H(v^2)\cap c_{\type{!}{0}}^{-1}(c_{\type{!}{0}}(v^2))|=|N_G(v)\cap c_{\type{!}{0}}^{-1}(c_{\type{!}{0}}(v))|$. Since we guarantee that $c_{\type{!}{0}}(x_i^v)\neq c(v)=c_{\type{!}{0}}(v^1)=c_{\type{!}{0}}(v^2)$, $x_i^v$ has one neighbour $y_i^v$ of its own colour and $0$ or $2$ neighbours of any other colour. Finally, $y_i^v$ has exactly one neighbour which has the same colour.

    Now assume that $v$ has odd degree. 
    Since the degree of $v$ is odd, there is an odd number of colour classes containing an odd number of neighbours of $v$. For each $i \in [q] \setminus \{c(v)\}$, if $v$ has an odd number of neighbours of colour $i$, we set $c_{\type{!}{0}}(x_i^v)=c_{\type{!}{0}}(y_i^v)=i$. Fix an arbitrary colour $j$ not equal to $c(v)$. For every $i \in [2q]$ for which $x_i^v$ and $y_i^v$ have not yet been assigned a colour we set $c_{\type{!}{0}}(x_i^v)=c_{\type{!}{0}}(y_i^v)=j$, adding an even number of extra neighbours of $v^1$ and $v^2$ of colour $j$. Similarly to before we can justify that this is a valid \type{!}{0}-colouring.

    On the other hand, 
    let $c_{\type{!}{0}}:V(H) \rightarrow [q]$ be a \type{!}{0}-$q$-colouring of $H$. Fix a vertex $v \in V(G)$. 
    Since each $y_i^v$ has degree $1$, 
    by the definition of \type{!}{0}-$q$-colouring we obtain that $c_{\type{!}{0}}(x_i^v)=c_{\type{!}{0}}(y_i^v)$. Since $v^1$ is adjacent to $x_i^v$ it follows that $c_{\type{!}{0}}(x_i^v)\neq c_{\type{!}{0}}(v^1)$. Therefore, we have $|N_H(v^1)\cap c_{\type{!}{0}}^{-1}(c_{\type{!}{0}}(v^1))|=|N_{G^1}(v^1)\cap c_{\type{!}{0}}^{-1}(c_{\type{!}{0}}(v^1))|$ by definition of \type{!}{0}-$q$-colouring. We conclude that restricting the colouring $c_{\type{!}{0}}$ to $G^1 \cong G$ results in a \type{!}{*}-colouring $c:V(G) \rightarrow [q]$.
\end{proof}


\subsubsection{\Type{!}{v} and \Type{1}{v}-colouring}

Recall that $c:V(G)\rightarrow [q]$ is a $\sigma$-$q$-colouring of $G$ for $\sigma\in \{\type{!}{v},\type{1}{v}\}$ if for every vertex $v$ it holds that $|N_G(v)\cap c^{-1}(i)|$ is either odd or $0$ for every $i\not=c(v)$. Additionally, for \type{!}{v}-$q$-colouring $|N_G(v)\cap c^{-1}(c(v))|=1$ and for \type{1}{v}-$q$-colouring $|N_G(v)\cap c^{-1}(c(v))|$ is odd.

\begin{theorem}\label{thm:!v}\label{link234:!v}\label{link34:1v}
    For $q \geq 2$ it is \NP-complete to decide whether a graph admits a \type{!}{v}-$q$-colouring  and for $q\geq 3$ it is \NP-complete to decide whether a graph admits a \type{1}{v}-$q$-colouring.
\end{theorem}

\begin{proof}
    We define $f:\{\type{!}{*},\type{1}{*}\} \rightarrow \{\type{!}{v},\type{1}{v}\}$ as $f(\type{!}{*}) = \type{!}{v}$ and $f(\type{1}{*}) = \type{1}{v}$. 
    Fixing $\sigma \in \{\type{!}{*},\type{1}{*}\}$, we reduce $\sigma$-$q$-colouring to $f(\sigma)$-$q$-colouring. Note that \type{!}{*}-$q$-colouring is \NP-complete for $2$ colours by \cite{Sch78} and for more than $2$ colours by \cite{DemaineKP25}, and \type1*-colouring is \NP-complete for $q\geq 3$ colours by \cite{BelmonteS21}. Given a graph $G$, we construct a graph $H$ as follows. For every vertex $v\in V(G)$ of even degree and every $i \in [q]$ we add a $K_2$ with vertices $x_i^v,y_i^v$ and the edge $vx_i^v$ to $G$. For every vertex $v\in V(G)$ of odd degree and every $i \in [q-1]$ we add a $K_2$ with vertices $x_i^v,y_i^v$ and the edge $vx_i^v$ to $G$.

    First, 
    let $c_{\sigma}:V(G) \rightarrow [q]$ be a $\sigma$-$q$-colouring of $G$. We obtain a $f(\sigma)$-$q$-colouring $c_{f(\sigma)}:V(H) \rightarrow [q]$ of $H$ as follows. For every vertex $v \in V(G)$, we set $c_{f(\sigma)}(v) = c_{\sigma}(v)$. Fix a vertex $v \in V(G)$ and assume first that $v$ has even degree. 
    Since the degree of $v$ is even, there is an even number of colour classes containing an odd number of neighbours of $v$. 
    Hence, there is an even number, if $q$ is even, or an odd number, if $q$ is odd, of colour classes containing an even number of neighbours of $v$. For each $i \in [q]$, if $v$ has an even number of neighbours of colour $i$, we set $c_{f(\sigma)}(x_i^v)=c_{f(\sigma)}(y_i^v)=i$. Fix an arbitrary colour $j$ not equal to $c_{\sigma}(v)$ (this is always possible as $q\geq 2$). For every $i \in [q]$ for which $v$ has an odd number of neighbours of colour $i$ we set $c_{f(\sigma)}(x_i^v)=c_{f(\sigma)}(y_i^v)=j$ adding an even number of extra neighbours of $v$ of colour $j$. Observe that $v$ has an odd number of neighbours of each colour in $H$. Furthermore, since $v$ has an odd number of neighbour of colour $c(v)$ in $G$, we do not assign any $x_i^v$ the colour $c_{\sigma}(v)$, and therefore $|N_H(v)\cap c_{f(\sigma)}^{-1}(c_{f(\sigma)}(v))|=|N_G(v)\cap c_{f(\sigma)}^{-1}(c_{f(\sigma)}(v))|$. Since we guarantee that $c_{f(\sigma)}(x_i^v)\neq c_{\sigma}(v)=c_{f(\sigma)}(v)$, $x_i^v$ has one neighbour $y_i^v$ of its own colour and $0$ or $1$ neighbours of any other colour. Finally, $y_i^v$ has exactly one neighbour which has the same colour.

    Now assume that $v$ has odd degree. 
    Since the degree of $v$ is odd, there is an odd number of colour classes containing an odd number of neighbours of $v$. In particular, there must be at least one colour class containing an odd number of neighbours of $v$ and, without loss of generality, assume the respective colour is $q$. We additionally obtain that there is an even number, if $q-1$ is even, or an odd number, if $q-1$ is odd, of colour classes containing an even number of neighbours of $v$. For each $i\in [q-1]$, if $v$ has an even number of neighbours of colour $i$, we set $c_{f(\sigma)}(x_{i}^v)=c_{f(\sigma)}(y_{i}^v)=i$. Fix an arbitrary colour $j$ not equal to $c_{\sigma}(v)$. For every $i \in [q-1]$ for which $v$ has an odd number of neighbours of colour $i$ we set $c_{f(\sigma)}(x_i^v)=c_{f(\sigma)}(y_i^v)=j$ adding an even number of extra neighbours of colour $j$. Similarly to before we can justify that this is a valid $f(\sigma)$-colouring.

    On the other hand, 
    let $c_{f(\sigma)}:V(H) \rightarrow [q]$ be a $f(\sigma)$-$q$-colouring of $H$. Fix a vertex $v \in V(G)$. 
    Since each $y_i^v$ has degree $1$, 
    by the definition of $f(\sigma)$-$q$-colouring we obtain that $c_{f(\sigma)}(x_i^v)=c_{f(\sigma)}(y_i^v)$. Since $v$ is adjacent to $x_i^v$ it follows that $c_{f(\sigma)}(x_i^v)\neq c_{f(\sigma)}(v)$. Therefore, we have $|N_H(v)\cap c_{f(\sigma)}^{-1}(c_{f(\sigma)}(v))|=|N_G(v)\cap c_{f(\sigma)}^{-1}(c_{f(\sigma)}(v))|$ by definition of $f(\sigma)$-$q$-colouring. We conclude that restricting $c_{f(\sigma)}$
    to $V(G)$ yields a $\sigma$-$q$-colouring.
\end{proof}

\subsubsection{\Type{0}{!}-colouring}
Note that $c:V(G)\rightarrow [q]$ is a \type0!-$q$-colouring if for every $v\in V(G)$ we have $|N_G(v)\cap c^{-1}(c(v))|$ is even and $|N_G(v)\cap c^{-1}(i)|=2$ for every $i\not=c(v)$. We show the following.
\begin{theorem}\label{link234:0!}
    For $q\geq 2$ it is \NP-complete to decide whether a graph admits
    a \type0!-$q$-colouring.
\end{theorem}

\begin{proof}
    We give a reduction from \type{1}{!}-$q$-colouring, which is \NP-complete for $2$ colours by \cref{cor2:1!} and for more than $2$ colours by \cref{cor34:1!}. Given a graph $G$, we construct a graph $H$ as follows. Take two disjoint copies of $G$, called $G^1$ and $G^2$. For each vertex $v \in V(G)$, we introduce a $q$-clique $X^v$ with vertices $x_1^v, \dots, x_q^v$.
    Finally, if $v^1$ and $v^2$ are the two copies of $v$, we add the edges $v^1x_q^v$ and $x_q^vv^2$. This completes the construction of $H$, which is illustrated in Figure~\ref{fig:0!}.

    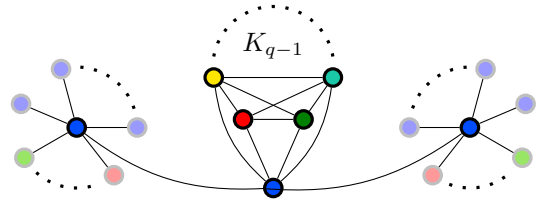
\begin{figure}[h]
  \centering
  \begin{tikzpicture}[scale=0.4]
  \def \dist {0.8cm}
  \def \r{2cm}
  \tikzstyle{lw2}=[line width=1.2]

        \draw[loosely dotted, lw2] (5:2) arc (5:100:2);
        \draw[loosely dotted, lw2] (215:2) arc (215:293:2);
      \node[b](m) at (0,0){};
      \node[gb](u) at (105:2){};
      \node[gb](v) at (157:2){};
      \node[gg](x) at (210:2){};
      \node[gr](y) at (308:2){};
      \node[gb](z) at (360:2){};
      \draw(m)--(u);
      \draw(m)--(v);
      \draw(m)--(x);
      \draw(m)--(y);
      \draw(m)--(z);

      \begin{scope}[xshift=13cm]
        \draw[loosely dotted, lw2] (80:2) arc (80:175:2);
        \draw[loosely dotted, lw2] (250:2) arc (250:310:2);
      \node[b](mm) at (0,0){};
      \node[gb](uu) at (75:2){};
      \node[gb](vv) at (23:2){};
      \node[gr](xx) at (232:2){};
      \node[gg](yy) at (330:2){};
      \node[gb](zz) at (180:2){};
      \draw(mm)--(uu);
      \draw(mm)--(vv);
      \draw(mm)--(xx);
      \draw(mm)--(yy);
      \draw(mm)--(zz);
      \end{scope}

      \begin{scope}[xshift=6.5cm,yshift=2cm]
          \draw[loosely dotted, lw2] (10:\r) arc (10:170:\r);
            \node[ye](y1) at (190:\r){};
            \node[r](x2) at (240:\r){};
            \node[g](y2) at (300:\r){};
            \node[tq](x3) at (350:\r){};
            \node[b](c) at (0,-4){};

            \draw (y1)--(x2)(y1)--(x3)(y1)--(y2);
            \draw (x2)--(x3)(x2)--(y2);
            \draw (y2)--(x3);
            \draw (c)--(x2)(c)--(y2)(c) to[bend left=20] (y1)(c) to[bend right=20] (x3);
            \draw (m) to[bend right=20] (c) to[bend right=20] (mm);
            \node at (0,0.7) {$K_{q-1}$};
            
      \end{scope}

  \end{tikzpicture}
  \caption{Construction for $q \geq 2$.}
    \label{fig:0!}
\end{figure}

    First,
    let $c:V(G) \rightarrow [q]$ be a \type{1}{!}-colouring of $G$. We obtain a \type0!-colouring $c_{\type{0}{!}}:V(H) \rightarrow [q]$ of $H$ as follows. Fix a vertex $v \in V(G)$. For each $\alpha \in \{1,2\}$, we set $c_{\type{0}{!}}(v^\alpha) = c_{\type{0}{!}}(x_q^v) = c(v)$. Let $c_1,\dots, c_{q-1}$ be the remaining colours (the colours different from $c(v)$). We assign $x_i^v$ the colour $c_i$ for each $i \in [q-1]$. It is easy to verify that $c_{\type{0}{!}}$ is a \type0!-colouring, as illustrated in Figure~\ref{fig:0!}.

    On the other hand, 
    let $c_{\type{0}{!}}:V(H) \rightarrow [q]$ be a \type0!-colouring of $H$. Fix a vertex $v \in V(G)$.
    Observe that the vertex $x_1^v$ has $q-1$ neighbours. Therefore, by definition of \type0!-$q$-colouring, all  colours apart from $c_{\type{0}{!}}(x_1^v)$  appear exactly once in its neighbourhood. In particular, the vertices $x_i^v$ and $x_j^v$ must receive different colours for each $i \neq j$. Moreover, since $c_{\type{0}{!}}(x_q^v)$ has to appear an even number of times in the neighbourhood of $x_q^v$, the vertices $v^1$ and $v^2$ must receive the same colour as $x_q^v$.

    As a consequence of the previous paragraph it follows that restricting the colouring $c_{\type{0}{!}}$ to $G^1 \cong G$ results in a \type{1}{!}-colouring $c:V(G) \rightarrow [q]$.
\end{proof}

\subsection{Polynomial cases using systems of linear equations}

In this section, we prove the polynomial-time solvability of some problems using a technique similar to that in~\cite{BelmonteS21}. In summary, we express the problem as the existence of a feasible solution to a system of linear equations over the binary field, which can be decided in polynomial time via Gaussian elimination.

\subsubsection{\Type{0}{1}-colouring}

\begin{theorem}\label{link2:01}
    For $q > 0$, \type{0}{1}-$q$-colouring is polynomial for $q = 2$.
\end{theorem}

\begin{proof}
    For $q = 2$, we build a system of linear equations over the binary field. Given a graph $G$, introduce a variable $x_v \in \{0,1\}$ for each $v \in V(G)$, where $x_v = 1$ if the vertex $v$ receives colour 1, and $x_v = 0$ otherwise. Now, introduce a variable $x_{uv} \in \{0,1\}$ for each $uv \in E(G)$, where $x_{uv} = 1$ if the vertices $u$ and $v$ receive the same colour, and $x_{uv} = 0$ otherwise. We guarantee this latter property by adding the following set of linear equations:
    \begin{equation}
        x_u + x_v + x_{uv} \equiv_2 1 \quad \text{for all } uv \in E(G).
    \end{equation}

    Next, observe that no vertex of even degree may exists. We guarantee this property by adding the linear equation $0 \equiv_2 1$ if such a vertex exists.
    
    Similarly, observe that every vertex of odd degree must share its colour with an even number of neighbours. We guarantee this property by adding the following set of linear equations:
    \begin{equation}
        \sum_{v \in N(u)} x_{uv} \equiv_2 0 \quad \text{for all } u \in V(G) \text{ such that } d(u) \text{ is odd}.
    \end{equation}

    It is easy to verify that $G$ is a YES-instance of \type01-colouring if and only if the system of linear equations above admits a feasible solution, which allows us to conclude.
\end{proof}

\subsubsection{\Type{0}{v}-colouring}

\begin{theorem}\label{link2:0v}
    For $q > 0$, \type{0}{v}-$q$-colouring is polynomial for $q = 2$.
\end{theorem}

\begin{proof}
    For $q = 2$, we build a system of linear equations over the binary field. Given a graph $G$, introduce a variable $x_v \in \{0,1\}$ for each $v \in V(G)$, where $x_v = 1$ if the vertex $v$ receives colour 1, and $x_v = 0$ otherwise. Now, introduce a variable $x_{uv} \in \{0,1\}$ for each $uv \in E(G)$, where $x_{uv} = 1$ if the vertices $u$ and $v$ receive the same colour, and $x_{uv} = 0$ otherwise. We guarantee this latter property by adding the following set of linear equations:
    \begin{equation}
        x_u + x_v + x_{uv} \equiv_2 1 \quad \text{for all } uv \in E(G).
    \end{equation}

    Next, observe that every vertex of even degree must share its colour with all its neighbours. We guarantee this property by adding the following set of linear equations:
    \begin{equation}
        x_{uv} \equiv_2 1 \quad \text{for all } u \in V(G) \text{ such that } d(u) \text{ is even}, \text{ for all } v \in N(u).
    \end{equation}
    Similarly, observe that every vertex of odd degree must share its colour with an even number of neighbours. We guarantee this property by adding the following set of linear equations:
    \begin{equation}
        \sum_{v \in N(u)} x_{uv} \equiv_2 0 \quad \text{for all } u \in V(G) \text{ such that } d(u) \text{ is odd}.
    \end{equation}

    It is easy to verify that $G$ is a YES-instance of \type0v-colouring if and only if the system of linear equations above admits a feasible solution, which allows us to conclude.
\end{proof}

\subsubsection{\Type{1}{1}-colouring}

\begin{theorem}\label{link2:11}
    For $q > 0$, \type{1}{1}-$q$-colouring is polynomial for $q = 2$.
\end{theorem}

\begin{proof}
    For $q = 2$, we build a system of linear equations over the binary field. Given a graph $G$, introduce a variable $x_v \in \{0,1\}$ for each $v \in V(G)$, where $x_v = 1$ if the vertex $v$ receives colour 1, and $x_v = 0$ otherwise. Now, introduce a variable $x_{uv} \in \{0,1\}$ for each $uv \in E(G)$, where $x_{uv} = 1$ if the vertices $u$ and $v$ receive the same colour, and $x_{uv} = 0$ otherwise. We guarantee this latter property by adding the following set of linear equations:
    \begin{equation}
        x_u + x_v + x_{uv} \equiv_2 1 \quad \text{for all } uv \in E(G).
    \end{equation}

    Next, observe that every vertex of even degree must share its colour with an odd number of neighbours. We guarantee this property by adding the following set of linear equations:
    \begin{equation}
        \sum_{v \in N(u)} x_{uv} \equiv_2 1 \quad \text{for all } u \in V(G) \text{ such that } d(u) \text{ is even}.
    \end{equation}
    Similarly, observe that no vertex of odd degree may exists. We guarantee this property by adding the linear equation $0 \equiv_2 1$ if such a vertex exists.

    It is easy to verify that $G$ is a YES-instance of \type11-colouring if and only if the system of linear equations above admits a feasible solution, which allows us to conclude.
\end{proof}

\subsubsection{\Type{1}{v}-colouring}

\begin{theorem}\label{link2:1v}
    For $q > 0$, \type{1}{v}-$q$-colouring is polynomial for $q = 2$.
\end{theorem}

\begin{proof}
    For $q = 2$, we build a system of linear equations over the binary field. Given a graph $G$, introduce a variable $x_v \in \{0,1\}$ for each $v \in V(G)$, where $x_v = 1$ if the vertex $v$ receives colour 1, and $x_v = 0$ otherwise. Now, introduce a variable $x_{uv} \in \{0,1\}$ for each $uv \in E(G)$, where $x_{uv} = 1$ if the vertices $u$ and $v$ receive the same colour, and $x_{uv} = 0$ otherwise. We guarantee this latter property by adding the following set of linear equations:
    \begin{equation}
        x_u + x_v + x_{uv} \equiv_2 1 \quad \text{for all } uv \in E(G).
    \end{equation}

    Next, observe that every vertex of even degree must share its colour with an odd number of neighbours.  We guarantee this property by adding the following set of linear equations:
    \begin{equation}
        \sum_{v \in N(u)} x_{uv} \equiv_2 1 \quad \text{for all } u \in V(G) \text{ such that } d(u) \text{ is even}.
    \end{equation}
    Similarly, observe that every vertex of odd degree must share its colour with all its neighbours. We guarantee this property by adding the following set of linear equations:
    \begin{equation}
        x_{uv} \equiv_2 1 \quad \text{for all } u \in V(G) \text{ such that } d(u) \text{ is odd}, \text{ for all } v \in N(u).
    \end{equation}

    It is easy to verify that $G$ is a YES-instance of \type1v-colouring if and only if the system of linear equations above admits a feasible solution, which allows us to conclude.
\end{proof}

\subsubsection{\Type{v}{0}-colouring}

\begin{theorem}\label{link2:v0}
    For $q > 0$, \type{v}{0}-$q$-colouring is polynomial for $q = 2$.
\end{theorem}

\begin{proof}
    For $q = 2$, we build a system of linear equations over the binary field. Given a graph $G$, introduce a variable $x_v \in \{0,1\}$ for each $v \in V(G)$, where $x_v = 1$ if the vertex $v$ receives colour 1, and $x_v = 0$ otherwise. Now, introduce a variable $x_{uv} \in \{0,1\}$ for each $uv \in E(G)$, where $x_{uv} = 1$ if the vertices $u$ and $v$ receive the same colour, and $x_{uv} = 0$ otherwise. We guarantee this latter property by adding the following set of linear equations:
    \begin{equation}
        x_u + x_v + x_{uv} \equiv_2 1 \quad \text{for all } uv \in E(G).
    \end{equation}

    Next, observe that no vertex of even degree may share its colour with any of its neighbours. We guarantee this property by adding the following set of linear equations:
    \begin{equation}
        x_{uv} \equiv_2 0 \quad \text{for all } u \in V(G) \text{ such that } d(u) \text{ is even}, \text{ for all } v \in N(u).
    \end{equation}
    Similarly, observe that every vertex of odd degree must share its colour with an odd number of neighbours. We guarantee this property by adding the following set of linear equations:
    \begin{equation}
        \sum_{v \in N(u)} x_{uv} \equiv_2 1 \quad \text{for all } u \in V(G) \text{ such that } d(u) \text{ is odd}.
    \end{equation}

    It is easy to verify that $G$ is a YES-instance of \type{v}{0}-colouring if and only if the system of linear equations above admits a feasible solution, which allows us to conclude.
\end{proof}

\subsubsection{\Type{v}{1}-colouring}

\begin{theorem}\label{link2:v1}
    For $q > 0$, \type{v}{1}-$q$-colouring is polynomial for $q = 2$.
\end{theorem}

\begin{proof}
    For $q = 2$, we build a system of linear equations over the binary field. Given a graph $G$, introduce a variable $x_v \in \{0,1\}$ for each $v \in V(G)$, where $x_v = 1$ if the vertex $v$ receives colour 1, and $x_v = 0$ otherwise. Now, introduce a variable $x_{uv} \in \{0,1\}$ for each $uv \in E(G)$, where $x_{uv} = 1$ if the vertices $u$ and $v$ receive the same colour, and $x_{uv} = 0$ otherwise. We guarantee this latter property by adding the following set of linear equations:
    \begin{equation}
        x_u + x_v + x_{uv} \equiv_2 1 \quad \text{for all } uv \in E(G).
    \end{equation}

    Next, observe that every vertex of even degree must share its colour with an odd number of neighbours. We guarantee this property by adding the following set of linear equations:
    \begin{equation}
        \sum_{v \in N(u)} x_{uv} \equiv_2 1 \quad \text{for all } u \in V(G) \text{ such that } d(u) \text{ is even}.
    \end{equation}
    Similarly, observe that no vertex of odd degree may share its colour with any of its neighbours. We guarantee this property by adding the following set of linear equations:
    \begin{equation}
        x_{uv} \equiv_2 0 \quad \text{for all } u \in V(G) \text{ such that } d(u) \text{ is odd}, \text{ for all } v \in N(u).
    \end{equation}

    It is easy to verify that $G$ is a YES-instance of \type{v}{1}-colouring if and only if the system of linear equations above admits a feasible solution, which allows us to conclude.
\end{proof}

\subsubsection{\Type$\star$1-colouring}

\begin{theorem}\label{link2:*1}
    For $q > 0$, \type{*}{1}-$q$-colouring is polynomial for $q = 2$.
\end{theorem}

\begin{proof}
    For $q = 2$, we build a system of linear equations over the binary field. Given a graph $G$, introduce a variable $x_v \in \{0,1\}$ for each $v \in V(G)$, where $x_v = 1$ if the vertex $v$ receives colour 1, and $x_v = 0$ otherwise. Now, introduce a variable $x_{uv} \in \{0,1\}$ for each $uv \in E(G)$, where $x_{uv} = 1$ if the vertices $u$ and $v$ receive the same colour, and $x_{uv} = 0$ otherwise. We guarantee this latter property by adding the following set of linear equations:
    \begin{equation}
        x_u + x_v + x_{uv} \equiv_2 1 \quad \text{for all } uv \in E(G).
    \end{equation}

    Next, observe that every vertex of even degree must share its colour with an odd number of neighbours. We guarantee this property by adding the following set of linear equations:
    \begin{equation}
        \sum_{v \in N(u)} x_{uv} \equiv_2 1 \quad \text{for all } u \in V(G) \text{ such that } d(u) \text{ is even}.
    \end{equation}
    Similarly, observe that every vertex of odd degree must share its colour with an even number of neighbours. We guarantee this property by adding the following set of linear equations:
    \begin{equation}
        \sum_{v \in N(u)} x_{uv} \equiv_2 0 \quad \text{for all } u \in V(G) \text{ such that } d(u) \text{ is odd}.
    \end{equation}

    It is easy to verify that $G$ is a YES-instance of \type*1-colouring if and only if the system of linear equations above admits a feasible solution, which allows us to conclude.
\end{proof}

\subsubsection{\Type{v}{v}-colouring}

\begin{theorem}\label{link2:vv}
    \type{v}{v}-$q$-colouring is polynomial for $q = 2$.
\end{theorem}

\begin{proof}
    We build a system of linear equations over the binary field. Given a graph $G$, introduce a variable $x_v \in \{0,1\}$ for each $v \in V(G)$, where $x_v = 1$ if the vertex $v$ receives colour 1, and $x_v = 0$ otherwise. Now, introduce a variable $x_{uv} \in \{0,1\}$ for each $uv \in E(G)$, where $x_{uv} = 1$ if the vertices $u$ and $v$ receive the same colour, and $x_{uv} = 0$ otherwise.  
    We guarantee this latter property by adding the following set of linear equations:
    \begin{equation*}
        x_u + x_v + x_{uv} \equiv_2 1 \quad \text{for all } uv \in E(G).
    \end{equation*}

    Next, observe that if the degree of $u$ is even and larger than $0$, then $u$ must be adjacent to both a vertex in its own colour and a vertex in the other colour. Hence, $u$ must be adjacent to both an odd number of vertices in its own colour and an odd number of vertices in the other colour, where the latter follows from the former. We guarantee these property by adding the following set of linear equations:
    \begin{equation*}
        \sum_{v \in N(u)} x_{uv} \equiv_2 1 \quad \text{for all } u \in V(G) \text{ such that } d(u) \text{ is even and larger than } 0.
    \end{equation*}
    Finally, for every vertex $u$ of odd degree introduce variable $z_u\in \{0,1\}$. Observe that if a vertex $u$ has odd degree then either all its neighbours must be the same colour as $u$ or all neighbours of $u$ must be the colour different from $u$. The variable $z_u=1$ indicates that all neighbours of $u$ have the same colour as $u$ and $z_u=0$ indicates that all neighbours of $u$ have a different colour from $u$. This property is ensured by adding the equations:
    \begin{equation*}
        x_{uv}+z_u\equiv 0 \quad \text{for all } u\in V(G), v\in N(u) \text{ such that } d(u) \text{ is odd.}
    \end{equation*}
    Observe that every vertex of degree $0$ can be coloured in either colour and automatically satisfies the conditions required by \type{v}{v}-colourings.

    It is easy to verify that $G$ is a YES-instance of \type{v}{v}-colouring if and only if the system of linear equations above admits a feasible solution, which allows us to conclude.
\end{proof}

\section{Complexity for three and potentially more colours}\label{sec:threeColours}

\subsection{Reductions from proper colouring}
In this subsection we present reductions from proper colouring with $q$ colours (proper $q$-colouring) which is \NP-complete for $q\geq 3$ \cite{G&J} even if $G$ has maximum degree at most $q+1$. The latter follows from the \NP-completeness of edge-colouring with $q$ colours ($q$-edge-colouring) \cite{Holyer,LeGa} by considering the line graph.


\subsubsection{\Type=0, \Type?0 and \Type{v}{0}-colouring}

First let us recall that $c\colon V(G)\rightarrow [q]$ is a $\sigma$-$q$-colouring of $G$ for $\sigma\in \{\type{=}{0},\type{?}{0},\type{v}{0}\}$ if for every vertex $v$ it holds that $|N_G(v)\cap c^{-1}(i)|$ is even for every $i\not=c(v)$. Additionally, for \type=0-$q$-colouring $|N_G(v)\cap c^{-1}(c(v))|=0$, for \type?0-$q$-colouring $|N_G(v)\cap c^{-1}(c(v))|<2$ and for \type{v}{0}-$q$-colouring $|N_G(v)\cap c^{-1}(c(v))|$ is either odd or $0$. 
\begin{theorem}\label{link34:=0}\label{link34:?0}\label{link34:v0}
    For $q\geq 3$ it is \NP-complete to decide whether a graph admits
    a \type=0, \type?0 or a \type{v}{0}-$q$-colouring.
\end{theorem}

\begin{proof}
    We give a reduction from proper $q$-colouring. Given a graph $G$, we construct a graph $H$ by adding a false twin to every vertex. 
    Fix $\sigma \in \{\type{=}{0},\type{?}{0},\type{v}{0}\}$.
    
    First, 
    let $c\colon V(G) \rightarrow [q]$  be a proper colouring of $G$. We obtain a $\sigma$-colouring $c_{\sigma}:V(H) \rightarrow [q]$ of $H$ as follows. For every vertex $v \in V(G)$, if $v_1,v_2 \in V(H)$ are the false twins corresponding to $v$ in $H$, we set $c_{\sigma}(v_1) = c_{\sigma}(v_2) = c(v)$. It is easy to verify that $c_{\sigma}$ is a $\sigma$-$q$-colouring.

    On the other hand, 
    let $c_{\sigma}:V(H) \rightarrow [q]$ be a $\sigma$-colouring of $H$. Fix a vertex $v \in V(G)$ and let $v_1,v_2 \in V(H)$ be the twins corresponding to $v$ in $H$.
    Observe that, by definition of $\sigma$-colouring, the size of $N_H(v_i) \cap c_{\sigma}^{-1}(j)$ is even for each $j \neq c_{\sigma}(v_i)$. Since by construction the size of the neighbourhood of every vertex is even, this implies that he size of $N_H(v_i) \cap c_{\sigma}^{-1}(c_{\sigma}(v_i))$ is even.
    By definition of $\sigma$-$q$-colouring, this implies that the size of $N(v_i) \cap c_{\sigma}^{-1}(c_{\sigma}(v_i))$ must be zero. In other words, $c_{\sigma}$ is a proper colouring of $H$.
    We obtain a proper colouring $c\colon V(G) \rightarrow [q]$ of $G$ as follows. For each vertex $v \in V(G)$, we set $c(v) = c_{\sigma}(v_1)$ or $c(v) = c_{\sigma}(v_2)$ where $v_1,v_2$ are the two twins corresponding to $v$. It is easy to verify that $c$ is a proper colouring.
\end{proof}


\subsubsection{\Type{=}{v}-colouring}

First let us recall that $c\colon V(G)\rightarrow [q]$ is a \type=v-$q$-colouring of $G$ if for every vertex $v$ it holds that $|N_G(v)\cap c^{-1}(c(v))|=0$ ($c$ is a proper colouring) and $|N_G(v)\cap c^{-1}(i)|$ is odd or $0$ for every $i\not=c(v)$. We prove the following.
\begin{theorem}\label{thm:=v}\label{link34:=v}
    For $q\geq 3$ it is \NP-complete to decide whether a graph admits
    a \type=v-$q$-colouring.
\end{theorem}

\begin{proof}
    We give a reduction from proper $q$-colouring. Given a graph $G$, we construct a graph $H$ by attaching $q-1$ pendants to each vertex of even degree and $q-2$ pendants to each vertex of odd degree.

    First, 
    let $c\colon V(G) \rightarrow [q]$ be a proper colouring of $G$. We obtain a \type=v-$q$-colouring $c_{\type{=}{v}}:V(H) \rightarrow [q]$ of $H$ as follows. For each vertex $v \in V(G)$, we set $c_{\type{=}{v}}(v) = c(v)$. Now, fix a vertex $v \in V(G)$, and assume without loss of generality that $c(v) = q$. 
    If $v$ has even degree, let $v_1,\dots, v_{q-1} \in V(H)$ be the pendants attached to $v$ in $H$. 
    Since $v$ has even degree, there is an even number of colour classes containing an odd number of neighbours of $v$. Furthermore, no neighbour has colour $q$. Hence, there is an even number, if $q-1$ is even, or an odd number, if $q-1$ is odd,  of colour classes among colours $[q-1]$ containing an even number of neighbours of $v$. For each $i\in [q-1]$, if $v$ has an even number of neighbours of colour $i$, we set $c_{\type{=}{v}}(v_{i})=i$. The remaining even number of pendants get assigned colour $1$. Hence, $v$ has no neighbours of colour $q$ in $H$ and an odd number of neighbours of each of the remaining colours in $H$. Furthermore, each of the pendants has no neighbour in their own colour and one neighbour in some other colour. 
    
    If $v$ has odd degree, let $v_1,\dots,v_{q-2} \in V(H)$ be the pendants attached to $v$ in $H$. 
    Since $v$ has odd degree, there is an odd number of colour classes containing an odd number of neighbours of $v$. As there in particular must be at least one colour class containing an odd number of neighbours of $v$, we may assume that $v$ has an odd number of neighbours of colour $q-1$.
    Furthermore, by assumption, colour $q$ does not appear in the neighbourhood of $v$. Hence, there is an odd number, if $q-2$ is odd, or an even number, if $q-2$ is even, of colour classes with colours in $[q-2]$ containing an even number of neighbours of $v$. For each $i\in [q-2]$, if $v$ has an even number of neighbours of colour $i$, we set $c_{\type{=}{v}}(v_{i})=i$. The remaining even number of pendants get assigned colour $1$. Again, $v$ has no neighbours of colour $q$ in $H$ and an odd number of neighbours of each of the remaining colours in $H$. Each of the pendants is adjacent to only one vertex of a different colour. Hence, $c_{\type{=}{v}}$ is a valid \type=v-colouring.

    On the other hand, 
    let $c_{\type{=}{v}}:V(H) \rightarrow [q]$ be a \type=v-colouring of $H$. We obtain a proper colouring $c\colon V(G) \rightarrow [3]$ of $G$ as follows. For each vertex $v \in V(G)$, we set $c(v) = c_{\type{=}{v}}(v)$. It follows directly from the definition of  \type=v-colouring that $c$ is a proper colouring.
\end{proof}


\subsubsection{\Type{2}{$\star$}, \Type{2}{v} and \Type{2}{?}-colouring}

First recall that $c\colon V(G)\rightarrow [q]$ is a $\sigma$-$q$-colouring of $G$ for $\sigma\in \{\type{2}{*},\type{2}{v},\type{2}{?}\}$ if for every vertex $v$ we have $|N_G(v)\cap c^{-1}(c(v))|$ even and not $0$. Additionally, for \type2*-$q$-colouring $|N_G(v)\cap c^{-1}(i)|$ is arbitrary, for \type2v-$q$-colouring $|N_G(v)\cap c^{-1}(i)|$ is odd or $0$ and for \type2?-$q$-colouring $|N_G(v)\cap c^{-1}(i)|<2$ for every $i\not=c(v)$. We show the following.

\begin{theorem}\label{link34:2*}\label{link34:2v}\label{link34:2?}
    For $q\geq 3$ it is \NP-complete to decide whether a graph admits
    a \type2*, \type2v or a \type2?-$q$-colouring.
\end{theorem}

\begin{proof}
    We give a reduction from proper $q$-colouring. Given a graph $G$, we construct a graph $H$ as follows. For each vertex $v \in V(G)$ of degree $d$, we introduce a cycle $C_v$ with vertices $x_1^v,\ldots,x_d^v$. We subdivide each edge of $C_v$ by inserting a vertex $y_i^v$ between $x_i^v$ and $x_{i+1}^v$, for each $i \in [d]$, with indices modulo $d$. 
    Now, for each vertex $v \in V(G)$ of degree $d$, let $E_v$ be the set of the edges incident to $v$. We arbitrarily choose a bijection $\pi_v: E_v \rightarrow \{x_1^v,\ldots,x_d^v\}$, assigning a label to each edge. Finally, for each edge $e = uv \in E(G)$, we add the edge $\pi_u(e)\pi_v(e)$. See Figure~\ref{fig:2starAND2vAND2question} for an illustration. Fix $\sigma \in \{\type{2}{*},\type{2}{v},\type{2}{?}\}$.

    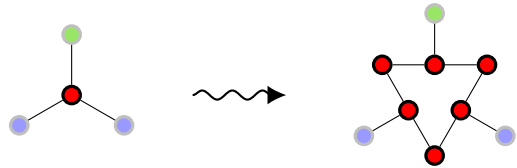
\begin{figure}[h]
        \centering
        \begin{tikzpicture}[scale=0.4]
        \begin{scope}[xshift=-12cm]
            \node[r](m) at (0,0){};
            \node[gg](x) at (90:2){};
            \node[gb](y) at (210:2){};
            \node[gb](z) at (330:2){};
            \draw (m)--(x);
            \draw (m)--(y);
            \draw (m)--(z);
            \draw[-{Latex[length=2.5mm, width=2.5mm]},decorate, decoration={snake, segment length=4mm, amplitude=0.6mm}, thick] (4,0)--(7.1,0);
        \end{scope}
        \node[r](x1) at (90:1){};
        \node[r](x2) at (210:1){};
        \node[r](x3) at (330:1){};
        \node[r](y1) at (150:2){};
        \node[r](y2) at (270:2){};
        \node[r](y3) at (30:2){};
        \node[gg](xx) at (90:2.7){};
        \node[gb](yy) at (210:2.7){};
        \node[gb](zz) at (330:2.7){};
        \draw (x1)--(y1)--(x2)--(y2)--(x3)--(y3)--(x1);
        \draw (x1)--(xx);
        \draw (x2)--(yy);
        \draw (x3)--(zz);
    \end{tikzpicture}
    \caption{Gadget representing vertex $v$ of $G$ in $H$, showing the cycle and the subdivision vertices.}
    \label{fig:2starAND2vAND2question}
\end{figure}
    
    First, 
    let $c\colon V(G) \rightarrow [q]$ be a proper colouring of $G$. We obtain a $\sigma$-colouring $c_{\sigma}:V(H) \rightarrow [q]$ of $H$ as follows. We set $c_{\sigma}(x_i^v) = c_{\sigma}(y_i^v) = c(v)$ for each $i \in [\deg_G(v)]$. It is easy to verify that $c_{\sigma}$ is a $\sigma$-colouring, as illustrated in Figure~\ref{fig:2starAND2vAND2question}.

    On the other hand, 
    let $c_{\sigma}:V(H) \rightarrow [q]$ be a $\sigma$-colouring of $H$. Fix a vertex $v \in V(G)$, and let $d$ denote its degree. 
    Observe that since $y_i^v$ has degree $2$ both its neighbours $x_i^v$ and $x_{i+1}^v$ must receive the same colour as $y_i^v$ in $c_{\sigma}$. Hence, the vertex set $\{x_1^v,\dots, x_d^v, y_1^v,\dots, y_d^v\}$ must be monochromatic. 
    Next, given $i \in [d]$, observe that $x_i^v$ has exactly one neighbour besides $y_{i-1}^v$ and $y_i^v$, namely $\pi_u(e) = x_j^u$ for some vertex $u$ and index $j$.
    As a consequence, $x_j^u$ must have a different colour than $x_i^v$. We conclude that, for every edge $e = uv \in E(G)$, the vertices $\pi_u(e)$ and $\pi_v(e)$ receive different colours. Hence, we can construct a proper colouring 
    $c\colon V(G) \rightarrow [q]$ of $G$ as follows. For each vertex $v \in V(G)$, we set $c(v) = c_{\sigma}(x_1^v)$. Isolated vertices (who do not have a corresponding subdivided cycle in $H$) can be coloured arbitrarily. 
\end{proof}


\subsubsection{\Type{0}{?}-colouring}

Note that $c\colon V(G)\rightarrow [q]$ is a \type0?-$q$-colouring if for every $v\in V(G)$ we have $|N_G(v)\cap c^{-1}(c(v))|$ is even and $|N_G(v)\cap c^{-1}(i)|<2$ for every $i\not=c(v)$. We show the following.
\begin{theorem}\label{link34:0?}
    For $q\geq 3$ it is \NP-complete to decide whether a graph admits
    a \type0?-$q$-colouring.
\end{theorem}

\begin{proof}
    We give a reduction from proper $q$-colouring on graphs of degree at most $q+1$. Given a graph $G$ of degree at most $q+1$, we construct a graph $H$ as follows. 
    We set $p=q+1$ in case $q$ is even and $q+2$ in case $q$ is odd.
    For each vertex $v \in V(G)$  we create a $p$-clique $X^v$ with vertices $x_1^{v},\dots, x_p^v$.  
    For each vertex $v \in V(G)$, let $E_v$ be the set of edges incident to $v$. We arbitrarily choose a bijection $\pi_v: E_v \rightarrow \{x_1^v,\dots,x_{d}^v\}$ where $d$ is the degree of $v$ in $G$. Finally, for each edge $e = uv \in E(G)$, we add the edge $\pi_u(e)\pi_v(e)$. For an illustration see Figure~\ref{fig:0question}.
    
    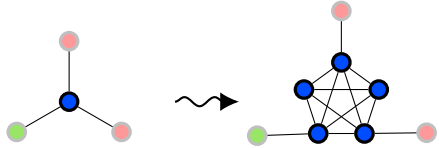
\begin{figure}[h]
  \centering
  \begin{tikzpicture}[scale=0.4]
  \def \dist {0.8cm}
  \begin{scope}[xshift=0]
      \node[b](m) at (0,0){};
      \node[gr](x) at (90:2){};
      \node[gg](y) at (210:2){};
      \node[gr](z) at (330:2){};
      \draw(m)--(x);
      \draw(m)--(y);
      \draw(m)--(z);
      \draw[-{Latex[length=2.5mm, width=2.5mm]},decorate, decoration={snake, segment length=4mm, amplitude=0.6mm}, thick] (3.5,0)--(5.62,0);
  \end{scope}
  \begin{scope}[xshift=9cm]
      \node[gr](x) at (90:3){};
      \node[gg](y) at (202:3){};
      \node[gr](z) at (340:3){};
      \node[b](w1) at (90:1.3){};
      \node[b](w2) at (162:1.3){};
      \node[b](w3) at (234:1.3){};
      \node[b](w4) at (306:1.3){};
      \node[b](w5) at (18:1.3){};
      \draw(w1)--(w2)(w1)--(w3)(w1)--(w4)(w1)--(w5)(w1)--(x);
      \draw(w5)--(w3)(w5)--(w4)(w5)--(w2)(w3)--(y);
      \draw(w4)--(w3)(w4)--(w2)(w4)--(z);
      \draw(w3)--(w2);
  \end{scope}

  \end{tikzpicture}
  \caption{Gadget representing a vertex $v$ of degree $3$ of $G$ in $H$ for $3$ colours.}
    \label{fig:0question}
\end{figure}

    First, 
    let $c\colon V(G) \rightarrow [q]$ be a proper colouring of $G$. We obtain a \type0?-colouring $c_{\type{0}{?}}:V(H) \rightarrow [q]$ of $H$ as follows. For each $v \in V(G)$ and $i\in [p]$, we set $c_{\type{0}{?}}(x_i^v) = c(v)$. 
    It is easy to verify that $c_{\type{0}{?}}$ is a \type0?-colouring, as illustrated in Figure~\ref{fig:0question}.

    On the other hand, 
    let $c_{\type{0}{?}}:V(H) \rightarrow [q]$ be a \type0?-colouring of $H$. 
    First observe that for each $v\in V(G)$ one colour must appear at least twice in $X^v$, say colour $j$. In case not all vertices of $X^v$ have colour $j$, say $c_{\type0?}(x_i^v)\not= j$, vertex $x_i^v$ has two neighbours of colour $j$ which is different from its own colour, a contradiction. Hence, $X^v$ is monochromatic. Given each vertex $x_i^v$ has at most one neighbour not contained in $X^v$, namely $\pi_u(e) = x_j^u$ for some vertex $u$ and index $j$, and an even number of neighbours inside $X^v$, $c_{\type0?}(x_i^v)\not=c_{\type0?}(x_j^u)$. Hence, for every edge $e = uv \in E(G)$, the vertices $x_1^u$ and $x_1^v$ receive different colours.
    Therefore, we obtain a proper colouring $c\colon V(G) \rightarrow [q]$ by setting $c(v) = c_{\type{0}{?}}(x_1^v)$ for each $v\in V(G)$. 
\end{proof}


\subsubsection{\Type{v}{v}-colouring}

First let us recall that $c\colon V(G)\rightarrow [q]$ is a \type{v}{v}-$q$-colouring of $G$ if for every vertex $v$ and every $i\in [q]$ we have $|N_G(v)\cap c^{-1}(i)|$ is odd or $0$. We prove the following.
\begin{theorem}\label{link34:vv}
    For $q\geq 3$ it is \NP-complete to decide whether a graph admits
    a \type{v}{v}-$q$-colouring.
\end{theorem}

\begin{proof}
    We give a reduction from proper $q$-colouring. Given a graph $G$, we construct a graph $H$ as follows. We subdivide each edge $e \in E(G)$ with a vertex  $x^{e}$. 
    Additionally, we attach a pendant $y^v$ to each non-subdivision vertex $v$ of even degree. For an illustration, see Figure~\ref{fig:vv}.

    \begin{figure}[h]
  \centering
  \begin{tikzpicture}[scale=0.4]
  \def \dist {0.8cm}
  \begin{scope}[xshift=-0.5cm]
      \node[b](m) at (0,0){};
      \node[gr](x) at (90:2){};
      \node[gg](y) at (210:2){};
      \node[gr](z) at (330:2){};
      \draw(m)--(x);
      \draw(m)--(y);
      \draw(m)--(z);
      \draw[-{Latex[length=2.5mm, width=2.5mm]},decorate, decoration={snake, segment length=4mm, amplitude=0.6mm}, thick] (2.5,0)--(4.62,0);
  \end{scope}
  \begin{scope}[xshift=7cm]
      \node[b](m) at (0,0){};
      \node[gr](x) at (90:3){};
      \node[gg](y) at (210:3){};
      \node[gr](z) at (330:3){};
      \node[g](xx) at (90:1.5){};
      \node[g](yy) at (210:1.5){};
      \node[g](zz) at (330:1.5){};
      \draw(m)--(xx)--(x);
      \draw(m)--(yy)--(y);
      \draw(m)--(zz)--(z);
  \end{scope}

  \begin{scope}[xshift=17cm]
      \node[b](m) at (0,0){};
      \node[gr](u) at (90:2){};
      \node[gg](x) at (180:2){};
      \node[gr](y) at (270:2){};
      \node[gp](z) at (0:2){};
      \draw(m)--(u);
      \draw(m)--(x);
      \draw(m)--(y);
      \draw(m)--(z);
      \draw[-{Latex[length=2.5mm, width=2.5mm]},decorate, decoration={snake, segment length=4mm, amplitude=0.6mm}, thick] (3,0)--(5.12,0);
  \end{scope}
  \begin{scope}[xshift=26cm]
      \node[b](m) at (0,0){};
      \node[gr](u) at (90:3){};
      \node[gg](x) at (180:3){};
      \node[gr](y) at (270:3){};
      \node[gp](z) at (0:3){};
      \node[g](uu) at (90:1.5){};
      \node[g](xx) at (180:1.5){};
      \node[g](yy) at (270:1.5){};
      \node[g](zz) at (0:1.5){};
      \node[g](uuu) at (45:2.3){};
      \draw(m)--(uu)--(u);
      \draw(m)--(xx)--(x);
      \draw(m)--(yy)--(y);
      \draw(m)--(zz)--(z);
      \draw(uuu)--(m);
  \end{scope}

  \end{tikzpicture}
  \caption{Gadget representing vertex $v$ of $G$ in $H$, showing subdivision vertices and pendants.}
    \label{fig:vv}
\end{figure}
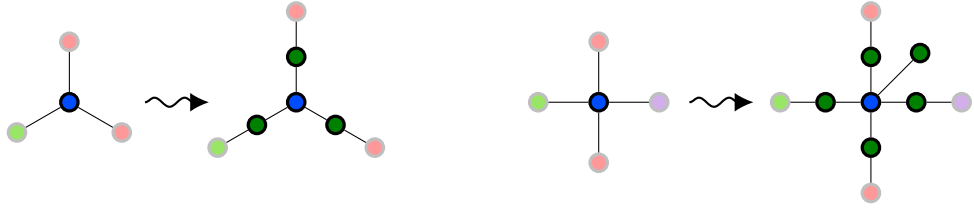

    First, 
    let $c\colon V(G) \rightarrow [q]$ be a proper colouring of $G$. We obtain a \type{v}{v}-colouring $c_{\type{v}{v}}:V(H) \rightarrow [q]$ of $H$ as follows. For each vertex $v \in V(G)$, we set $c_{\type{v}{v}}(v) = c(v)$. For each edge $e \in E(G)$, we set $c_{\type{v}{v}}(x^{e})=1$ and for each $v$ of even degree we set  $c_{\type{v}{v}}(y^{v})=1$. 
    Now each original vertex is adjacent to an odd number of vertices in colour $1$ and no further vertices. Each pendant is adjacent to one vertex of some colour (potentially its own). Each subdivision vertex is adjacent to two vertices of different colours and hence has $1$ or $0$ neighbours of each colour. Hence, $c_{\type{v}{v}}$ is a valid \type{v}{v}-colouring.
    
    On the other hand, 
    let $c_{\type{v}{v}}:V(H) \rightarrow [q]$ be a \type{v}{v}-colouring of $H$. Fix any edge $uv \in E(G)$. Observe that, by definition of \type{v}{v}-colouring, the vertex $x^{uv}$ forces $u$ and $v$ to have different colours. 
    Hence, we obtain a proper colouring $c\colon V(G) \rightarrow [q]$ of $G$ by setting $c(v) = c_{\type{v}{v}}(v)$ for every $v \in V(G)$.
\end{proof}


\subsection{Reductions from edge colouring}
In this section we present reductions from edge-$q$-colouring which is \NP-complete even on $q$-regular graphs \cite{Holyer,LeGa}.

\subsubsection{\IR colouring}\label{ss:improperRainbow}

Many colouring such as types \type{+}1, \type11 and \type1{+} 
coincide in the case when the considered graph is $q$-regular, and we call them
\emph{\iR colourings}. Formally, $c\colon V(G)\rightarrow [q]$ is an \iR $q$-colouring of a $q$-regular graph $G$ if for every $v\in V(G)$ it holds that $|N_G(v)\cap c^{-1}(i)|=1$ for every $i\in [q]$.
An example is given in Figure~\ref{fig:torwall}.

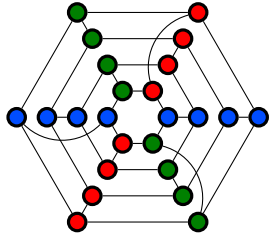
\begin{figure}[h]
  \centering
  \begin{tikzpicture}[scale=0.4]
    \foreach \r in {1,2,3,4} {
      \foreach[count=\n] \c in {r,g,b,r,g,b} {
        \pgfmathtruncatemacro{\a}{60*\n}
        \node[\c] (\n\r) at (\a:\r) {};
      }
      \draw (1\r)--(2\r)--(3\r)--(4\r)--(5\r)--(6\r)--(1\r);
    }
    \foreach \n in {2,4,6} \draw (\n1)--(\n2) (\n3)--(\n4);
    \foreach \n in {1,3,5} {
      \draw (\n2)--(\n3);
      \draw (\n1) to[bend left=45] (\n4);
    }
  \end{tikzpicture}
  \caption[rainbow colouring]{This toroidal wall is a $C_4$-free bipartite
    cubic graph that comes with an \iR colouring.}
  \label{fig:torwall}
\end{figure}

It is easy to see that a graph admits an \iR colouring if and only if
all its connected components admit one. Every $1$-regular graph has an \iR
colouring, because every component is a $K_2$ and the monochromatic colouring
satisfies the rainbow condition.

For a graph $G=(V,E)$ let $D = \{uv \in E \mid N(u) \cap N(v) \ne \es\}$
be the set of edges contained in a triangle of $G$.

\begin{lemma}
  A $q$-regular graph $G=(V,E)$ has an \iR $q$-colouring only if
  $(V, E \sm D)$ has a perfect matching.
\end{lemma}

\begin{proof}
  Let $c\colon V \to [q]$ be an \iR colouring of $G$.
  The set $M = \{uv \in E \mid c(u)=c(v)\}$ is a perfect matching of $G$
  since every vertex $v \in V$ has exactly one neighbour $u$ of the same
  colour. Moreover, $D \cap M = \es$ because each colour appears in each
 neighbourhood exactly once.
\end{proof}

\begin{corollary}
  A $2$-regular graph has an \iR colouring if and only if for every connected
  component there is an integer $k>0$ such that this component is a $C_{4k}$.
\end{corollary}

Let $G=(V,E)$ be a graph. The \emph{neighbourhood graph} $\NC(G)$ is the graph
obtained from $G$ by removing all edges in $E$ and completing the neighbourhood
of every vertex in $V$. Formally,
$\NC(G) = (V, \bigcup\limits_{v \in V} \! N(v)^{(2)})$.

\begin{lemma} \label{l:N(G)}
  A $q$-regular graph $G=(V,E)$ has an \iR $q$-colouring only if
  $\NC(G)$ has a proper $q$-colouring.  
\end{lemma}

\begin{proof}
  An \iR colouring of $G$ is a proper $q$-colouring of $\NC(G)$.
\end{proof}

\begin{lemma}
  For every \iR colouring of a regular graph, all colour classes have
  the same size. Moreover, its size is even.
\end{lemma}

\begin{proof}
  We apply an inductive argument on $q$, the number of colours for
  \iR colourings of $q$-regular graphs. For $q=2$, every \iR colouring
  of $C_{4k}$ has $2k$ vertices of either colour. For $q>2$ consider an
  \iR colouring $c \colon V \to [q]$ of $G=(V,E)$. For every pair of
  different colours $i,j \in [q]$ the subgraph of $G$ induced by the vertices
  of colours $i$ and $j$ is $2$-regular, and $c$ restricted to these vertices
  becomes an \iR colouring of the subgraph. Hence we have $|c^{-1}(i)| =
  |c^{-1}(j)|$. The claim follows by transitivity. 
\end{proof}

\begin{corollary}
  A $q$-regular graph on $n$ vertices has an \iR colouring only if
  $n$ is divisible by $2q$.
\end{corollary}

The \emph{inflation} $G^{*}$ of a graph $G=(V,E)$ has the vertex set
$\{(v,e) \colon v \in e \text{ and } e \in E\}$. Its vertices $(v,e)$
and $(w,f)$ are adjacent if $v=w$ or $e=f$, see~\cite{Chvatal}.

\begin{theorem}
    For $q\geq 3$ it is \NP-complete to decide whether a $q$-regular graph admits
    an \iR colouring
\end{theorem}

\begin{proof}
  We show that for all integers $q \ge 3$, the edges of a $q$-regular graph $G$ can
  be properly coloured by $q$ colours if and only if $G^{*}$ has an
  \iR colouring. 
  
  First let $c \colon E \to [q]$ be a proper edge colouring of $G=(V,E)$.
  For an edge $e=\{u,v\}$ of $G$ we colour the vertices $(u,e)$ and $(v,e)$
  by colour $c(e)$. This gives an \iR colouring of $G^{*}$.

  Conversely, if $c^{*}$ is an \iR colouring of $G^{*}$ then, for every
  vertex $v \in V$, on the $q$-clique $C$ induced by
  $\{(v,e) \colon e \ni v\}$ in $G^{*}$, each of the $q$ colours appears
  exactly once. Consequently, for
  every edge $e=\{u,v\}$ of $G$ we have $c^{*}(u,e)=c^{*}(v,e)$. Therefore,
  using these common colours on the edges of $G$ results in a proper edge
  colouring of $G$.
\end{proof}

\begin{corollary}\label{cor34:1!}\label{link34:11}\label{link34:1+}\label{link34:1!}\label{link34:v1}\label{link34:v!}\label{link34:*1}\label{link34:*!}\label{link34:+1}\label{link34:++}\label{link34:+!}\label{link34:?1}\label{link34:??}\label{link34:?!}\label{link34:!1}\label{link34:!+}\label{link34:!?}\label{link34:!!}
    For every $q \geq 3$, it is \NP-complete to decide whether a graph admits a \hypertarget{113}{\type11}, \type1+, \type1!, \type{v}{1}, \type{v}{!}, \type*1, \type*!, \type+1, \type++, \type+!, \type?1, \type??, \type?!, \type!1, \type!+, \type!?, or \type!!-colouring with $q$ colours.
\end{corollary}

\begin{corollary}\label{c8}
    For $q = 3$, it is \NP-complete to decide whether a graph admits a \type{!}{v}-colouring.
\end{corollary}

\subsubsection{Proper rainbow colouring}\label{ss:proper-rainbow}
In this section we consider proper rainbow colourings which are defined similarly to improper rainbow colourings with the difference that proper rainbow colourings are proper colouring. We define a more general version in the following ($\PRC(q,1)$ on $q$-regular graphs is proper rainbow colouring).
Fix two integers $q,k \geq 1$. A proper colouring $c\colon V(G) \rightarrow [q]$ is a $\PRC(q,k)$ of a graph $G$, if $|N(v) \cap c^{-1}(i)| \geq k$ for all $v \in V(G)$ and all $i \neq c(v)$.

\begin{remark}
    If $G$ is $(k(q-1))$-regular, then necessarily $|N(v) \cap c^{-1}(i)| = k$ for all $v \in V$ and all $i \neq c(v)$.
\end{remark}

\begin{theorem}
    For all $q \geq 3$ and $k \geq 2$ even, it is \NP-complete to decide whether a graph admits a $\PRC(q,k)$ even for $(k(q-1))$-regular graphs.
\end{theorem}

\begin{proof}
    We give a reduction from $q$-edge-colouring on $q$-regular graphs. Given a $q$-regular graph $G$, we construct a graph $H$ as follows. Take $k/2$ disjoint copies of the line graph $L(G)$. For each edge $e = uv \in E(G)$, let $e_1,\ldots,e_{k/2}$ be the $k/2$ copies of the vertex of $L(G)$ corresponding to $e$. We add an edge between $e_i$ and every neighbour of $e_j$ within its own copy of $L(G)$, for each $i \neq j$.

    First, 
    let $c\colon E(G) \rightarrow [q]$ be a $q$-edge-colouring of $G$. We obtain a $\PRC(q,k)$-colouring $c':V(H) \rightarrow [q]$ of $H$ as follows. For each edge $e = uv \in E(G)$, we set $c'(e_i) = c(e)$ for each $i \in [k/2]$. It is easy to verify that $c'$ is a $\PRC(q,k)$-colouring.

    On the other hand, 
    let $c':V(H) \rightarrow [q]$ be a $\PRC(q,k)$-colouring of $H$. Fix an edge $e = uv \in V(G)$.
    
    \begin{claim}
        Each copy of the vertex of $L(G)$ corresponding to $e$ has at most two neighbours within its own copy for each colour different from its own. Moreover, for each such colour, one of these neighbours corresponds to a vertex representing an edge in $G$ incident to $u$, and the other corresponds to a vertex representing an edge in $G$ incident to $v$.
    \end{claim}
    
    \begin{claimproof}
        Suppose, by contradiction, that this is not the case. Let $e_i$ be one of the $k/2$ copies of the vertex of $L(G)$ corresponding to $e$. If $e_i$ has at least three neighbours in its own copy sharing the same colour, then, by the pigeonhole principle, at least two of them correspond to vertices representing edges incident to the same endpoint, either $u$ or $v$. Without loss of generality, assume that two of these neighbours correspond to edges incident to $u$. By definition of the line graph, this means that these vertices are adjacent, a contradiction.
    \end{claimproof}

    We conclude that each copy of the vertex of $L(G)$ corresponding to $e$ has at most two neighbours within its own copy for each colour different from its own. Since there are $k/2$ copies of the line graph, it follows that $e_i$ has at most $k$ neighbours of each colour. By definition of $\PRC(q,k)$-colouring, this quantity must be exact. Therefore, each copy of the vertex of $L(G)$ corresponding to $e$ has exactly two neighbours in its own copy, for each colour except its own. Moreover, for each such colour, these two neighbours correspond to edges in $G$ incident to distinct endpoints $u$ and $v$.
    
    We obtain a $q$-edge-colouring $c\colon E(G) \rightarrow [q]$ of $G$ as follows. For each edge $e = uv \in E(G)$, we set $c(e) = c'(e_1)$. It is easy to verify that $c$ is a $q$-edge-colouring.
\end{proof}

\begin{corollary}\label{link34:02}\label{link34:v2}\label{link34:*2}\label{link34:=2}\label{link34:?2}
    For every $q \geq 3$, it is \NP-complete to decide whether a graph admits a \type02, \type{v}{2}, \type*2, \type=2 or \type?2-colouring with $q$ colours.
\end{corollary}

\begin{proof}
    For $(2(q-1))$-regular graphs, these problems are exactly $\PRC(q,2)$.
\end{proof}

Recall that $c\colon V(G)\rightarrow [q]$ is a \type!2-$q$-colouring of $G$ if for every vertex $v$ it holds that $|N_G(v)\cap c^{-1}(c(v))|=1$ and $|N_G(v)\cap c^{-1}(i)|$ is even and larger than $0$ for every $i\not=c(v)$. We prove the following.

\begin{theorem}\label{link34:!2}
    For $q\geq 3$ it is \NP-complete to decide whether a graph admits
    a \type!2-$q$-colouring.
\end{theorem}

\begin{proof}
    We give a reduction from $\PRC(q,2)$-colouring on $2(q-1)$-regular graphs. Given a $2(q-1)$-regular graph $G$, we construct a graph $H$ as follows. Take two disjoint copies of $G$, called $G^1$ and $G^2$. For each vertex $v \in V(G)$, we introduce two  $K_{2q-1}$, denoted $K^{v,1}$ and $K^{v,2}$. For $\alpha \in \{1,2\}$, the vertices of the $K_{2q-1}$ $K^{v,\alpha}$ are denoted by $w_1^{v,\alpha}, \dots, w_{2q-1}^{v,\alpha}$.
    Additionally, for every $v\in V(G)$ we introduce a $K_2$ with vertices $x^{v},y^{v}$. Finally, if $v^1$ and $v^2$ are the two copies of $v$, we add the edges $v^{\alpha}w_{2q-1}^{v,\alpha}$ for each $\alpha \in \{1,2\}$, $w_i^{v,\alpha}x^{v}$ for each $i \in \{1,\dots,q-1\}$ and each $\alpha \in \{1,2\}$, and $w_i^{v,\alpha}y^v$ for each $i \in \{q,\dots,2q-2\}$ and each $\alpha \in \{1,2\}$. This completes the construction of $H$, which is illustrated in Figure~\ref{fig:!2}.

    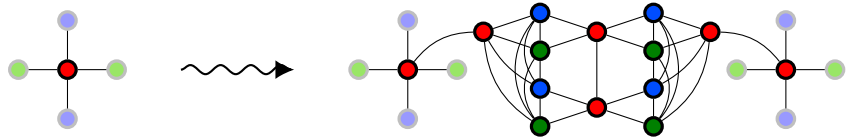
\begin{figure}[h]
        \centering
        \begin{tikzpicture}[scale=0.5]
        \begin{scope}[xshift=-14cm]
            \node[r](m) at (0,0){};
            \node[gg](u) at (-1.3,0){};
            \node[gb](x) at (0,-1.3){};
            \node[gg](y) at (1.3,0){};
            \node[gb](z) at (0,1.3){};
            \draw (m)--(u);
            \draw (m)--(x);
            \draw (m)--(y);
            \draw (m)--(z);
            \draw[-{Latex[length=2.5mm, width=2.5mm]},decorate, decoration={snake, segment length=4mm, amplitude=0.6mm}, thick] (3,0)--(6,0);
        \end{scope}
            \node[r](m1) at (-5,0){};
            \node[gg](u1) at (-6.3,0){};
            \node[gb](x1) at (-5,-1.3){};
            \node[gg](y1) at (-3.7,0){};
            \node[gb](z1) at (-5,1.3){};

            \node[b](a1) at (-1.5,1.5){};
            \node[g](b1) at (-1.5,0.5){};
            \node[b](c1) at (-1.5,-0.5){};
            \node[g](d1) at (-1.5,-1.5){};
            \node[r](e1) at (-3,1){};

            \node[r](v) at (0,1){};
            \node[r](w) at (0,-1){};

            \node[b](a2) at (1.5,1.5){};
            \node[g](b2) at (1.5,0.5){};
            \node[b](c2) at (1.5,-0.5){};
            \node[g](d2) at (1.5,-1.5){};
            \node[r](e2) at (3,1){};
            
            \node[r](m2) at (5,0){};
            \node[gg](u2) at (6.3,0){};
            \node[gb](x2) at (5,-1.3){};
            \node[gg](y2) at (3.7,0){};
            \node[gb](z2) at (5,1.3){};
            
            \draw (m1)--(u1)(m1)--(x1)(m1)--(y1)(m1)--(z1);
            \draw (e1)--(a1)(e1)--(b1)(e1)to[bend right=15](c1)(e1)to[bend right=25](d1);
            \draw (e1) to[bend right=25](m1);
            \draw (a1)--(v)(b1)--(v)(c1)--(w)(d1)--(w);
            \draw (v)--(w);
            \draw (a2)--(v)(b2)--(v)(c2)--(w)(d2)--(w);
            \draw (e2) to[bend left=25](m2);
            \draw (e2)--(a2)(e2)--(b2)(e2)to[bend left=15](c2)(e2)to[bend left=25](d2);
            \draw (m2)--(u2)(m2)--(x2)(m2)--(y2)(m2)--(z2);
            \draw (a1)--(b1)--(c1)--(d1)(a1)to[bend right=35](c1)(b1)to[bend right=35](d1)(a1)to[bend right=40](d1);
            \draw (a2)--(b2)--(c2)--(d2)(a2)to[bend left=35](c2)(b2)to[bend left=35](d2)(a2)to[bend left=40](d2);
        \end{tikzpicture}
        \caption{Gadget representing vertex $v$ of $G$ in $H$, showing the two copies of $G$, the two $K_5$ (here $q=3$), vertices $x^v$ and $y^v$ and all additional edges.}
        \label{fig:!2}
    \end{figure}

    First, 
    let $c\colon V(G) \rightarrow [q]$ be a $\PRC(3,2)$-colouring of $G$. We obtain a \type!2-colouring $c_{\type{!}{2}}:V(H) \rightarrow [q]$ of $H$ as follows. Fix a vertex $v \in V(G)$. For each $\alpha \in \{1,2\}$, we set $c_{\type{!}{2}}(v^\alpha) = c_{\type{!}{2}}(w_1^{v,\alpha}) = c_{\type{!}{2}}(x^v) = c_{\type{!}{2}}(y^v) = c(v)$. Let $c_1,\dots, c_{q-1}$ be the remaining colours (the colours different from $c(v)$). We assign $w_i^{v,\alpha}$ and $w_{q-1+i}^{v,\alpha}$ the colour $c_i$ for each $i \in [q-1]$ and for each $\alpha \in \{1,2\}$. It is easy to verify that $c_{\type{!}{2}}$ is a \type!2-colouring, as illustrated in Figure~\ref{fig:!2}.

    On the other hand, 
    let $c_{\type{!}{2}}:V(H) \rightarrow [q]$ be a \type!2-colouring of $H$. 
    
    \begin{claim}
        The edge $v^1 w_{2q-1}^{v,1}$ is monochromatic for every $v\in V(G)$.
    \end{claim}
    
    \begin{claimproof}
        Suppose, by contradiction, that this is not the case. Without loss of generality, assume $c_{\type{!}{2}}(v^1) = 1$ and $c_{\type{!}{2}}(w_{2q-1}^{v,1}) = 2$. By definition of \type!2-colouring, out of the remaining vertices of the $K_{2q-1}$ $K^{v,1}$ exactly one vertex, say $w_i^{v,1}$, receives colour $1$, exactly one vertex, say $w_j^{v,1}$ receives colour $2$, and there are precisely two vertices of each of the remaining colours.  In this configuration, each vertex of $K^{v,1}$ still requires another neighbour of colour $1$ in their neighbourhood, which forces both $x^v$ and $y^v$ to be coloured $1$. But since $w_i^{v,1}$, $x^v$ and $y^v$ have colour $1$, either $x^v$ or $y^v$ (depending on whether $i\leq q-1$ or $q\leq i\leq q-2$) has two neighbours  of its own colour, a contradiction.
    \end{claimproof}
    As a consequence of the claim it follows that restricting the colouring $c_{\type{!}{2}}$  to $G^1\cong G$ results in a 
     $\PRC(3,2)$-colouring $c\colon V(G) \rightarrow [q]$.
\end{proof}

In order to prove the \NP-Completeness of $PRC(q,k)$ for each $q \geq 3$ and each $k \geq 2$ even, we gave a reduction from $q$-edge-colouring on $q$-regular graphs. This approach does not appear to readily work in order to prove the \NP-Completeness of $PRC(q,1)$, which corresponds to proper rainbow $q$-colouring. A different approach, relying on the so called graph covering problems, turns out to be fruitful.

For a fixed graph $H$, the $H$-cover problem consist in, given a graph $G$, finding a ``local isomorphism'': a labeling of vertices of $G$ by vertices of $H$ so that the label set of the neighbourhood of every $v \in V(G)$ is equal to the neighbourhood (in $H$) of the label of $v$ and each neighbour of $v$ is labeled by a different neighbour of the label of $v$~\cite{kratochvil1997covering}. In other words, we want a graph homomorphism $f:V(G) \rightarrow V(H)$ such that $f$ is bijective in the neighbourhood of each vertex, which is called a covering projection from $G$ to $H$~\cite{bok2024list}.

An interesting case is when $H$ is $k$-regular for a given $k$. It is known that the $H$-cover problem is \NP-Complete for $k \geq 3$ and polynomial time solvable for $k \leq 2$~\cite{fiala2008locally}. In particular, it is \NP-Complete for $H = K_k$ with $k \geq 4$. This result provides exactly what we need to prove the following.

\begin{theorem}
    For all $q \geq 4$ it is \NP-complete to decide whether a $(q-1)$-graph admits a proper rainbow $q$-colouring.
\end{theorem}

\begin{proof}
    Note that asking whether $G$ has a proper rainbow $q$-colouring is equivalent to asking whether $G$ is a covering graph of $K_q$. Thus, the result follows from~\cite[Theorem 4.1]{kratochvil1994complexity}.
\end{proof}

\begin{corollary}\label{link4:01}\label{link4:0+}\label{link4:0!}\label{link4:v1}\label{link4:v+}\label{link4:v!}\label{link4:*1}\label{link4:*+}\label{link4:*!}\label{link4:=1}\label{link4:=+}\label{link4:=?}\label{link4:=!}\label{link4:?1}\label{link4:?+}\label{link4:?!}
    For every $q \geq 4$, it is \NP-complete to decide whether a graph admits a \type01, \type0+, \type0!, \type{v}{1}, \type{v}{+}, \type{v}{!}, \type*1, \type*+, \type*!, \type=1, \type=+, \type=?, \type=!, \type?1, \type?+, or \type?!-colouring with $q$ colours.
\end{corollary}

\subsubsection{\Type{2}{2} and \type{+}{2}-colouring}

Recall that $c\colon V(G)\rightarrow [q]$ is a $\sigma$-$q$-colouring for $\sigma\in \{\type22, \type+2\}$ of $G$  if for every vertex $v$ it holds that $|N_G(v)\cap c^{-1}(i)|$ is even and larger than $0$ for every $i\neq c(v)$. Additionally, for \type22-$q$-colouring  $|N_G(v)\cap c^{-1}(c(v))|$ is even and larger than $0$ and for \type+2-$q$-colouring $|N_G(v)\cap c^{-1}(i)|>0$.
\begin{theorem}\label{thm:22}\label{link34:22}\label{link34:+2}
    For $q \geq 3$ it is \NP-complete to decide whether a graph admits a  \type22 or a \type+2-$q$-colouring.
\end{theorem}

\begin{proof}
    We give a reduction from $q$-edge-colouring on $q$-regular graphs. Given a $q$-regular graph $G$, we construct a graph $H$ as follows. For each vertex $v \in V(G)$, we introduce a $K_{2q}$ $X^v$ with vertices $x_1^v,\ldots,x_q^v$ and $y_1^v,\ldots,y_q^v$. We also introduce, for each $i \in [q]$, a $K_{2q,2q}$ $Z^{v,i}$ with bipartition $(\{z_1^{v,i},\ldots,z_{2q}^{v,i}\},\{w_1^{v,i},\ldots,w_{2q}^{v,i}\})$. We modify each $Z^{v,i}$ by removing edges $z_1^{v,i}w_1^{v,i}$ and $z_1^{v,i}w_2^{v,i}$ and adding the edge $w_1^{v,i}w_2^{v,i}$. Additionally, for every $i \in [q]$, we replace the edge $x_i^vy_i^v$ by the two edges $x_i^vz_1^{v,i}$ and $y_i^vz_1^{v,i}$. For each vertex $v \in V(G)$, let $E_v$ be the set of the $q$ edges incident to $v$. We arbitrarily choose a bijection $\pi_v: E_v \rightarrow [q]$, assigning a label to each edge. For each $\alpha \in \{x,y\}$, we define $\pi_{v,\alpha}: E_v \rightarrow \{\alpha_1^v,\ldots,\alpha_q^v\}$, where $\pi_{v,\alpha}(e) = \alpha_{\pi_v(e)}^v$. Finally, for each edge $e = uv \in E(G)$ and each $\alpha \in \{x,y\}$, we add the edge $\pi_{u,\alpha}(e)\pi_{v,\alpha}(e)$. The construction of $H$ is illustrated in Figure~\ref{fig:22}. Fix $\sigma\in \{\type22,\type+2\}$. Note that because $H$ is $2q$-regular, the two colouring versions coincide.

    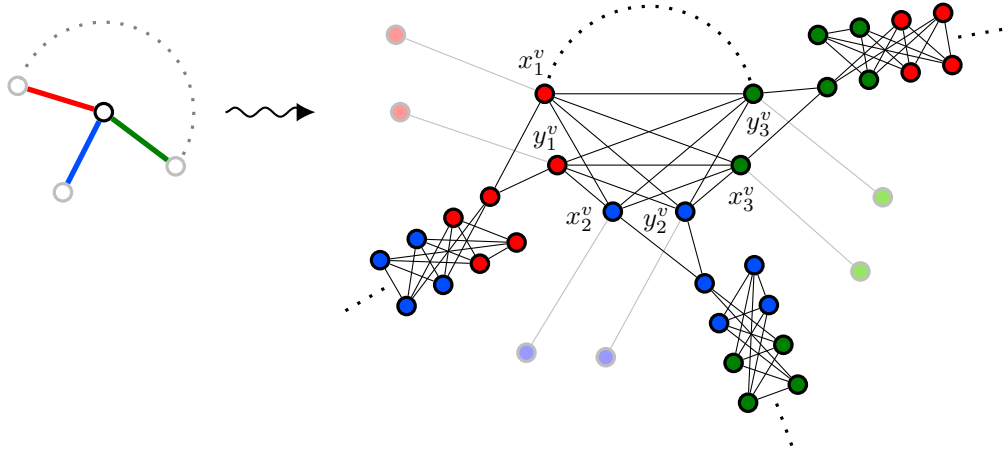
\begin{figure}[h]
        \centering
        \begin{tikzpicture}[scale=0.7]
            \tikzstyle{lw2}=[line width=1.2]
            \def \xdist {2cm}
            \def \r {2cm}
            \def \rr {4.7cm}
            \def \rrr {3.4cm}

            \begin{scope}[xshift=-10.3cm]
                \draw[loosely dotted, lw2,gray] (323:1.7) arc (323:523:1.7);
                \node[a](m) at (0,0){};
                \node[l](x) at (163:1.7){};
                \node[l](y) at (243:1.7){};
                \node[l,fill=white](z) at (323:1.7){};
                \draw[c1,lw1] (m)--(x);
                \draw[c3,lw1] (m)--(y);
                \draw[c2,lw1] (m)--(z);
                \draw[-{Latex[length=2.5mm, width=2.5mm]},decorate, decoration={snake, segment length=4mm, amplitude=0.6mm}, thick] (2.3,0)--(4,0);
            \end{scope}

            \draw[loosely dotted, lw2] (10:\r) arc (10:170:\r);
            \node[r](x1) at (170:\r){};
            \node[r](y1) at (210:\r){};
            \node[b](x2) at (250:\r){};
            \node[b](y2) at (290:\r){};
            \node[g](x3) at (330:\r){};
            \node[g](y3) at (370:\r){};
            \node at (157:\r+0.4cm){$x_1^v$};
            \node at (194:\r){$y_1^v$};
            \node at (237:\r+0.4cm){$x_2^v$};
            \node at (274:\r+0.05cm){$y_2^v$};
            \node at (317:\r+0.4cm){$x_3^v$};
            \node at (354:\r+0.1cm){$y_3^v$};
            \node[gr](xx1) at (163:\rr+0.3cm){};
            \node[gr](yy1) at (180:\rr){};
            \node[gb](xx2) at (243:\rr+0.4cm){};
            \node[gb](yy2) at (260:\rr){};
            \node[gg](xx3) at (323:\rr+0.3cm){};
            \node[gg](yy3) at (340:\rr){};
            \draw (x1)--(x2)(x1)--(x3)(x1)--(y2)(x1)--(y3);
            \draw (y1)--(x2)(y1)--(x3)(y1)--(y2)(y1)--(y3);
            \draw (x2)--(x3)(x2)--(y3);
            \draw (y2)--(x3)(y2)--(y3);
            \draw[lightgray](x1)--(xx1)(x2)--(xx2)(x3)--(xx3)(y1)--(yy1)(y2)--(yy2)(y3)--(yy3);
            \begin{scope}[shift=(208:\rrr),rotate=210]
                \node[r](z1) at (0,0){};
                \node[r](z2) at (0,1){};
                \node[r](z3) at (0.8,0){};
                \node[r](z4) at (0.8,1){};
                \node[b](z5) at (1.6,0){};
                \node[b](z6) at (1.6,1){};
                \node[b](z7) at (2.4,0){};
                \node[b](z8) at (2.4,1){};
                \draw (x1)--(z1)--(y1);
                \draw 
                (z2)--(z4)
                (z1)--(z6)(z1)--(z8);
                \draw (z3)--(z2)(z3)--(z4)(z3)--(z6)(z3)--(z8);
                \draw (z5)--(z2)(z5)--(z4)(z5)--(z6)(z5)--(z8);
                \draw (z7)--(z2)(z7)--(z4)(z7)--(z6)(z7)--(z8);
                \draw[loosely dotted, lw2] (2.6,0.5) -- (3.6,0.5);
            \end{scope}
            \begin{scope}[shift=(288:\rrr),rotate=290]
                \node[b](zz1) at (0,0){};
                \node[b](zz2) at (0,1){};
                \node[b](zz3) at (0.8,0){};
                \node[b](zz4) at (0.8,1){};
                \node[g](zz5) at (1.6,0){};
                \node[g](zz6) at (1.6,1){};
                \node[g](zz7) at (2.4,0){};
                \node[g](zz8) at (2.4,1){};
                \draw (x2)--(zz1)--(y2);
                \draw 
                (zz2)--(zz4)
                (zz1)--(zz6)(zz1)--(zz8);
                \draw (zz3)--(zz2)(zz3)--(zz4)(zz3)--(zz6)(zz3)--(zz8);
                \draw (zz5)--(zz2)(zz5)--(zz4)(zz5)--(zz6)(zz5)--(zz8);
                \draw (zz7)--(zz2)(zz7)--(zz4)(zz7)--(zz6)(zz7)--(zz8);
                \draw[loosely dotted, lw2] (2.6,0.5) -- (3.6,0.5);
            \end{scope}
            \begin{scope}[shift=(368:\rrr),rotate=370]
                \node[g](zzz1) at (0,0){};
                \node[g](zzz2) at (0,1){};
                \node[g](zzz3) at (0.8,0){};
                \node[g](zzz4) at (0.8,1){};
                \node[r](zzz5) at (1.6,0){};
                \node[r](zzz6) at (1.6,1){};
                \node[r](zzz7) at (2.4,0){};
                \node[r](zzz8) at (2.4,1){};
                \draw (x3)--(zzz1)--(y3);
                \draw 
                (zzz2)--(zzz4)
                (zzz1)--(zzz6)(zzz1)--(zzz8);
                \draw (zzz3)--(zzz2)(zzz3)--(zzz4)(zzz3)--(zzz6)(zzz3)--(zzz8);
                \draw (zzz5)--(zzz2)(zzz5)--(zzz4)(zzz5)--(zzz6)(zzz5)--(zzz8);
                \draw (zzz7)--(zzz2)(zzz7)--(zzz4)(zzz7)--(zzz6)(zzz7)--(zzz8);
                \draw[loosely dotted, lw2] (2.6,0.5) -- (3.6,0.5);
            \end{scope}
                
        \end{tikzpicture}
        \caption{Construction of $H$ in the proof of \cref{thm:22}.}
        \label{fig:22}
    \end{figure}

    First, 
    let $c\colon E(G) \rightarrow [q]$ be a $q$-edge-colouring of $G$. We obtain a $\sigma$-colouring $c_{\sigma}:V(H) \rightarrow [q]$ of $H$ as follows. For each edge $e = uv \in E(G)$ and each $\alpha \in \{x,y\}$, we set $c_{\sigma}(\pi_{u,\alpha}(e)) = c_{\sigma}(\pi_{v,\alpha}(e)) = c(e)$. Next, for every vertex $v \in V(G)$ and each $i \in [q]$, we assign to the vertices $z_1^{v,i}$, $z_2^{v,i}$, $w_1^{v,i}$, and $w_2^{v,i}$, the same colour as $x_i^v$ and $y_i^v$. Then, we assign a second colour to $z_3^{v,i}$, $z_4^{v,i}$, $w_3^{v,i}$, and $w_4^{v,i}$, a third colour to $z_5^{v,i}$, $z_6^{v,i}$, $w_5^{v,i}$, and $w_6^{v,i}$, and so on, until we assign the last colour to $z_{2q-1}^{v,i}$, $z_{2q}^{v,i}$, $w_{2q-1}^{v,i}$, and $w_{2q}^{v,i}$. It is easy to verify that $c_{\sigma}$ is a $\sigma$-colouring, as illustrated in Figure~\ref{fig:22}.

    On the other hand, 
    let $c_{\sigma}:V(H) \rightarrow [q]$ be a $\sigma$-colouring of $H$. Fix a vertex $v \in V(G)$.
    
    \begin{claim}
        For each $i \in [q]$, the vertices $x_i^v$, $y_i^v$ and $z_1^{v,i}$ must share the same colour.
    \end{claim}
    
    \begin{claimproof}
        First note that the neighbourhood of $w_3^{v,i}$ is precisely $\{z_1^{v,i},\dots, z_{2q}^{v,i}\}$ and therefore, by definition, every colour $i\neq c_{\sigma}(w_3^{v,i})$ has to appear at least twice among $\{z_1^{v,i},\dots, z_{2q}^{v,i}\}$ and an even number of times. Additionally, $c_{\sigma}(w_3^{v,i})$ has to appear at least once implying that every colour appears exactly twice among $\{z_1^{v,i},\dots, z_{2q}^{v,i}\}$. Since the neighbourhood of $w_1^{v,i}$ is the set $\{w_2^{v,i},z_2^{v,i},\dots, z_{2q}^{v,i}\}$ we conclude that $w_2^{v,i}$ must have colour $ c_\sigma(z_1^{v,i})$, because otherwise $c_\sigma(z_1^{v,i})$ and $c_\sigma(w_2^{v,i})$ would appear an odd number of times in the neighbourhood of $w_1^{v,i}$ contradicting the definition of $\sigma$-colouring. With a symmetric argument we obtain that $c_\sigma(z_1^{v,i})=c_\sigma(w_2^{v,i})=c_\sigma(w_1^{v,i})$. 

        As the neighbourhood of $z_2^{v,i}$ is precisely $\{w_1^{v,i},\dots, w_{2q}^{v,i}\}$, similarly to before, we know that every colour appears exactly twice among $\{w_1^{v,i},\dots, w_{2q}^{v,i}\}$. Hence, restricted to $Z^{v,i}$ the neighbourhood of $z_1^{v,i}$ contains every colour apart from $c_\sigma(z_1^{v,i})$ exactly twice. This implies that at least one of the two remaining neighbours of $z_1^{v,i}$ must have colour $c_\sigma(z_1^{v,i})$. If the other of the two neighbours has colour $k\neq c_\sigma(z_1^{v,i})$ then clearly $z_1^{v,i}$ has $3$ neighbours of colour $k\neq c_\sigma(z_1^{v,i})$, a contradiction. Therefore, $c_\sigma(z_1^{v,i})=c_\sigma(x_1^v)=c_\sigma(y_1^v)$ as claimed.
    \end{claimproof}

    Since $x_i^v$ has $2q$ neighbours, all colours in its neighbourhood appear exactly twice. The same argument applies to the neighbourhood of $y_i^v$. $c_\sigma(z_1^{v,i})=c_\sigma(x_1^v)=c_\sigma(y_1^v)$, for $i \neq j$, the colour of $x_i^v$ and $y_i^v$ is different from the colour of $x_j^v$ and $y_j^v$ for $j\neq i$. 
    This implies that 
    for every edge $e = uv \in E(G)$, the four vertices $\pi_{u,x}(e)$, $\pi_{v,x}(e)$, $\pi_{u,y}(e)$ and $\pi_{v,y}(e)$ receive the same colour. Therefore, we obtain a $q$-edge-colouring $c\colon E(G) \rightarrow [q]$ of $G$ as follows. For each edge $e = uv \in E(G)$, we set $c(e) = c_{\type{2}{2}}(\pi_{u,x}(e))$.
\end{proof}

\subsubsection{\Type{1}{2}-colouring}
Recall that $c\colon V(G)\rightarrow [q]$ is a $\type12$-$q$-colouring  of $G$  if for every vertex $v$ it holds that $|N_G(v)\cap c^{-1}(c(v))|$ is odd and $|N_G(v)\cap c^{-1}(i)|$ is even and larger than $0$ for every $i\neq c(v)$. 
\begin{theorem}\label{thm:12}\label{link34:12}
    For $q \geq 3$ it is \NP-complete to decide whether a graph admits a \type12-$q$-colouring.
\end{theorem}

\begin{proof}
    We give a reduction from $q$-edge-colouring on $q$-regular graphs. For the reduction we use an equal gadget. An equal gadget between vertices $u$ and $v$ consists of a $K_{2q}$ plus edges $ua$, $va$ for an arbitrarily chosen vertex $a$ of the $K_{2q}$.  We call $u$ and $v$ the outlets of the equal gadget and $a$ the anchor of the equal gadget (see \cref{fig:12equal}). It is easy to show the following.
    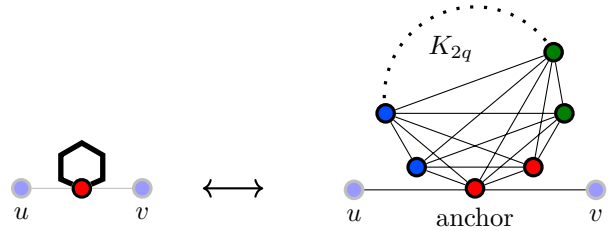
\begin{figure}[h]
        \centering
        \begin{tikzpicture}[scale=0.4]
            \tikzstyle{lw2}=[line width=1.2]
            \def \xdist {4cm}
            \def \r {3cm}
            \def \size {0.8cm}
            \begin{scope}[xshift=-15cm,yshift=-3cm]
                \node[r](b) at (2,0){};
                \draw[line width=2pt] (b)--($(2,0.7)+( 330:\size)$)--($(2,0.7)+(30:\size)$)--($(2,0.7)+( 90:\size)$)--($(2,0.7)+( 150:\size)$)--($(2,0.7)+( 210:\size)$)--(b);
                \node[gb,label=below:$u$](u) at (0,0){};
                \node[gb,label=below:$v$](v) at (4,0){};
                \draw[lightgray] (u)--(b)--(v);
                \draw[<->,thick] (6,0)--(8,0);
            \end{scope}
            \node[gb,label=below:$u$](uu) at (-\xdist,-3.05){};
            \node[gb,label=below:$v$](vv) at (\xdist,-3.05){};
            \draw (uu)--(vv);
            \draw[loosely dotted, lw2] (30:\r) arc (30:190:\r);
            \node[b](c1) at (190:\r){};
            \node[b](c2) at (230:\r){};
            \node[r,label=below:anchor](c3) at (270:\r){};
            \node[r](c4) at (310:\r){};
            \node[g](c5) at (350:\r){};
            \node[g](c6) at (30:\r){};
            \draw (c1)--(c2)(c1)--(c3)(c1)--(c4)(c1)--(c5)(c1)--(c6)(c2)--(c3)(c2)--(c4)(c2)--(c5)(c2)--(c6)(c3)--(c4)(c3)--(c5)(c3)--(c6)(c4)--(c5)(c4)--(c6)(c5)--(c6);
            \node at (-0.8,1.6){$K_{2q}$};
        \end{tikzpicture}
        \caption{Equal gadget between vertices $u$ and $v$. }
        \label{fig:12equal}
    \end{figure}
    \begin{claim}\label{claim:12equalGadget}
        The outlets of any equal gadget have to receive the same colour while the anchor can have any colour in any $\type12$-$q$-colouring. 
    \end{claim}
    \begin{claimproof}
        Let $u,v$ be the outlets of the equal gadget, $a$ the anchor and $b$ with $b\not=a$ another vertex of the $K_{2q}$ of the equal gadget. Since $b$ has exactly $2q-1$ neighbours, in any \type12-$q$-colouring $b$ must have exactly $2$ neighbours of every colour apart from its own and $1$ neighbour of its own colour. Hence, in the $K_{2q}$ every colour appears exactly twice. This implies that the anchor $a$ has $2$ neighbours of every colour apart from its own and $1$ neighbour of its own colour within the $K_{2q}$. To satisfy the conditions of \type12-colouring, $u$ and $v$ must have the same colour.
    \end{claimproof}
    We now describe the reduction. Given a $q$-regular graph $G$, we construct a graph $H$ as follows. For each vertex $v \in V(G)$, we introduce a $K_{2q,2q}$ with bipartitions $\{w_1^v,\dots, w_q^v, x_1^v,\dots, x_q^v\}$ and $\{y_1^v,\dots, y_q^v, z_1^v,\dots, z_q^v\}$. For every $i\in [q]$ we remove edges $w_i^vz_i^v$, $w_i^vy_i^v$, $x_i^vy_i^v$ and $x_i^vz_i^v$. We add an equal gadget with anchor $a_i^v$ between $y_i^v$ and $w_i^v$, an equal gadget with anchor $b_i^v$ between $w_i^v$ and $x_i^v$ and an equal gadget with anchor $c_i^v$ between $x_i^v$ and $z_i^v$. For every vertex $v \in V(G)$, let $E_v$ be the set of the $q$ edges incident to $v$. We arbitrarily choose a bijection $\pi_v: E_v \rightarrow [q]$, assigning a label to each edge. Finally, for each edge $e = uv \in E(G)$, we add an equal gadget between $w_{\pi(u)}^u$ and $w_{\pi(v)}^v$ with anchor $\overline{a}^e$ and an equal gadget between $x_{\pi(u)}^u$ and $x_{\pi(v)}^v$ with anchor $\overline{b}^e$. See \cref{fig:12} for an illustration.

     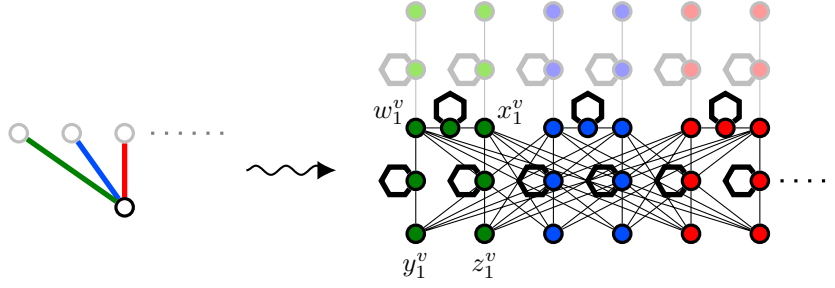
\begin{figure}[h]
        \centering
        \begin{tikzpicture}[scale=0.7]
            \tikzstyle{lw2}=[line width=1.2]
            \def \xdist {1.3cm}
            \def \ydist {2cm}
            \def \yydist {2.2cm}
            \def \r {2cm}
            \def \rr {4.7cm}
            \def \rrr {3.4cm}
            \def \size {0.3cm}

            \begin{scope}[xshift=-5.5cm,yshift=0.5cm]
                \draw[loosely dotted, lw2,gray] (0.5,1.4) --(2,1.4);
                \node[a](m) at (0,0){};
                \node[l](x) at (0,1.4){};
                \node[l](y) at (-1,1.4){};
                \node[l,fill=white](z) at (-2,1.4){};
                \draw[c1,lw1] (m)--(x);
                \draw[c3,lw1] (m)--(y);
                \draw[c2,lw1] (m)--(z);
                \draw[-{Latex[length=2.5mm, width=2.5mm]},decorate, decoration={snake, segment length=4mm, amplitude=0.6mm}, thick] (2.3,0.7)--(4,0.7);
            \end{scope}

            \draw[loosely dotted, lw2] (5.3*\xdist,0.5*\ydist) -- (6.3*\xdist,0.5*\ydist);
            \node[g](w1) at (0,\ydist){};
            \node[g](x1) at (\xdist,\ydist){};
            \node[g](y1) at (0,0){};
            \node[g](z1) at (\xdist,0){};
            \node[b](w2) at (2*\xdist,\ydist){};
            \node[b](x2) at (3*\xdist,\ydist){};
            \node[b](y2) at (2*\xdist,0){};
            \node[b](z2) at (3*\xdist,0){};
            \node[r](w3) at (4*\xdist,\ydist){};
            \node[r](x3) at (5*\xdist,\ydist){};
            \node[r](y3) at (4*\xdist,0){};
            \node[r](z3) at (5*\xdist,0){};
            \node at (-0.5,\ydist+0.3cm){$w_1^v$};
            \node at (\xdist+0.5cm,\ydist+0.3cm){$x_1^v$};
            \node at (0,-0.6){$y_1^v$};
            \node at (\xdist,-0.6){$z_1^v$};
            \node[gg](ww1) at (0,\ydist+\yydist){};
            \node[gg](xx1) at (\xdist,\ydist+\yydist){};
            \node[gb](ww2) at (2*\xdist,\ydist+\yydist){};
            \node[gb](xx2) at (3*\xdist,\ydist+\yydist){};
            \node[gr](ww3) at (4*\xdist,\ydist+\yydist){};
            \node[gr](xx3) at (5*\xdist,\ydist+\yydist){};
             \draw[lightgray](w1)--(ww1)(x1)--(xx1)(w2)--(ww2)(x2)--(xx2)(w3)--(ww3)(x3)--(xx3);  
             \draw (w2)--(y1)--(x2)(w2)--(z1)--(x2)(w3)--(y1)--(x3)(w3)--(z1)--(x3)(w3)--(y2)--(x3)(w3)--(z2)--(x3);
             \draw (y2)--(w1)--(z2)(y2)--(x1)--(z2)(y3)--(w1)--(z3)(y3)--(x1)--(z3)(y3)--(w2)--(z3)(y3)--(x2)--(z3);
             \node[g](a1) at (0,0.5*\ydist){};
             \node[g](b1) at (0.5*\xdist,\ydist){};
            \node[g](c1) at (\xdist,0.5*\ydist){};
            \node[gg](d1) at (0,\ydist+0.5*\yydist){};
            \node[gg](e1) at (\xdist,\ydist+0.5*\yydist){};
            \draw (y1)--(a1)--(w1)--(b1)--(x1)--(c1)--(z1);
            \draw[line width=2pt] (a1)--($(-0.35,0.5*\ydist)+( 60:\size)$)--($(-0.35,0.5*\ydist)+(120:\size)$)--($(-0.35,0.5*\ydist)+( 180:\size)$)--($(-0.35,0.5*\ydist)+( 240:\size)$)--($(-0.35,0.5*\ydist)+( 300:\size)$)--(a1);
            \draw[line width=2pt] (b1)--($(0.5*\xdist,\ydist+0.35cm)+( 330:\size)$)--($(0.5*\xdist,\ydist+0.35cm)+(30:\size)$)--($(0.5*\xdist,\ydist+0.35cm)+( 90:\size)$)--($(0.5*\xdist,\ydist+0.35cm)+( 150:\size)$)--($(0.5*\xdist,\ydist+0.35cm)+( 210:\size)$)--(b1);
            \draw[line width=2pt] (c1)--($(-0.35cm+\xdist,0.5*\ydist)+( 60:\size)$)--($(-0.35cm+\xdist,0.5*\ydist)+(120:\size)$)--($(-0.35cm+\xdist,0.5*\ydist)+( 180:\size)$)--($(-0.35cm+\xdist,0.5*\ydist)+( 240:\size)$)--($(-0.35cm+\xdist,0.5*\ydist)+( 300:\size)$)--(c1);
            \draw[line width=2pt,lightgray] (d1)--($(-0.35,\ydist+0.5*\yydist)+( 60:\size)$)--($(-0.35,\ydist+0.5*\yydist)+(120:\size)$)--($(-0.35,\ydist+0.5*\yydist)+( 180:\size)$)--($(-0.35,\ydist+0.5*\yydist)+( 240:\size)$)--($(-0.35,\ydist+0.5*\yydist)+( 300:\size)$)--(d1);
            \draw[line width=2pt,lightgray] (e1)--($(-0.35cm+\xdist,\ydist+0.5*\yydist)+( 60:\size)$)--($(-0.35cm+\xdist,\ydist+0.5*\yydist)+(120:\size)$)--($(-0.35cm+\xdist,\ydist+0.5*\yydist)+( 180:\size)$)--($(-0.35cm+\xdist,\ydist+0.5*\yydist)+( 240:\size)$)--($(-0.35cm+\xdist,\ydist+0.5*\yydist)+( 300:\size)$)--(e1);

            \node[b](a2) at (2*\xdist,0.5*\ydist){};
            \node[b](b2) at (2.5*\xdist,\ydist){};
            \node[b](c2) at (3*\xdist,0.5*\ydist){};
            \node[gb](d2) at (2*\xdist,\ydist+0.5*\yydist){};
            \node[gb](e2) at (3*\xdist,\ydist+0.5*\yydist){};
            \draw (y2)--(a2)--(w2)--(b2)--(x2)--(c2)--(z2);
            \draw[line width=2pt] (a2)--($(-0.35cm+2*\xdist,0.5*\ydist)+( 60:\size)$)--($(-0.35cm+2*\xdist,0.5*\ydist)+(120:\size)$)--($(-0.35cm+2*\xdist,0.5*\ydist)+( 180:\size)$)--($(-0.35cm+2*\xdist,0.5*\ydist)+( 240:\size)$)--($(-0.35cm+2*\xdist,0.5*\ydist)+( 300:\size)$)--(a2);
            \draw[line width=2pt] (b2)--($(2.5*\xdist,\ydist+0.35cm)+( 330:\size)$)--($(2.5*\xdist,\ydist+0.35cm)+(30:\size)$)--($(2.5*\xdist,\ydist+0.35cm)+( 90:\size)$)--($(2.5*\xdist,\ydist+0.35cm)+( 150:\size)$)--($(2.5*\xdist,\ydist+0.35cm)+( 210:\size)$)--(b2);
            \draw[line width=2pt] (c2)--($(-0.35cm+3*\xdist,0.5*\ydist)+( 60:\size)$)--($(-0.35cm+3*\xdist,0.5*\ydist)+(120:\size)$)--($(-0.35cm+3*\xdist,0.5*\ydist)+( 180:\size)$)--($(-0.35cm+3*\xdist,0.5*\ydist)+( 240:\size)$)--($(-0.35cm+3*\xdist,0.5*\ydist)+( 300:\size)$)--(c2);
            \draw[line width=2pt,lightgray] (d2)--($(-0.35cm+2*\xdist,\ydist+0.5*\yydist)+( 60:\size)$)--($(-0.35cm+2*\xdist,\ydist+0.5*\yydist)+(120:\size)$)--($(-0.35cm+2*\xdist,\ydist+0.5*\yydist)+( 180:\size)$)--($(-0.35cm+2*\xdist,\ydist+0.5*\yydist)+( 240:\size)$)--($(-0.35cm+2*\xdist,\ydist+0.5*\yydist)+( 300:\size)$)--(d2);
            \draw[line width=2pt,lightgray] (e2)--($(-0.35cm+3*\xdist,\ydist+0.5*\yydist)+( 60:\size)$)--($(-0.35cm+3*\xdist,\ydist+0.5*\yydist)+(120:\size)$)--($(-0.35cm+3*\xdist,\ydist+0.5*\yydist)+( 180:\size)$)--($(-0.35cm+3*\xdist,\ydist+0.5*\yydist)+( 240:\size)$)--($(-0.35cm+3*\xdist,\ydist+0.5*\yydist)+( 300:\size)$)--(e2);

            \node[r](a3) at (4*\xdist,0.5*\ydist){};
            \node[r](b3) at (4.5*\xdist,\ydist){};
            \node[r](c3) at (5*\xdist,0.5*\ydist){};
            \node[gr](d3) at (4*\xdist,\ydist+0.5*\yydist){};
            \node[gr](e3) at (5*\xdist,\ydist+0.5*\yydist){};
            \draw (y3)--(a3)--(w3)--(b3)--(x3)--(c3)--(z3);
            \draw[line width=2pt] (a3)--($(-0.35cm+4*\xdist,0.5*\ydist)+( 60:\size)$)--($(-0.35cm+4*\xdist,0.5*\ydist)+(120:\size)$)--($(-0.35cm+4*\xdist,0.5*\ydist)+( 180:\size)$)--($(-0.35cm+4*\xdist,0.5*\ydist)+( 240:\size)$)--($(-0.35cm+4*\xdist,0.5*\ydist)+( 300:\size)$)--(a3);
            \draw[line width=2pt] (b3)--($(4.5*\xdist,\ydist+0.35cm)+( 330:\size)$)--($(4.5*\xdist,\ydist+0.35cm)+(30:\size)$)--($(4.5*\xdist,\ydist+0.35cm)+( 90:\size)$)--($(4.5*\xdist,\ydist+0.35cm)+( 150:\size)$)--($(4.5*\xdist,\ydist+0.35cm)+( 210:\size)$)--(b3);
            \draw[line width=2pt] (c3)--($(-0.35cm+5*\xdist,0.5*\ydist)+( 60:\size)$)--($(-0.35cm+5*\xdist,0.5*\ydist)+(120:\size)$)--($(-0.35cm+5*\xdist,0.5*\ydist)+( 180:\size)$)--($(-0.35cm+5*\xdist,0.5*\ydist)+( 240:\size)$)--($(-0.35cm+5*\xdist,0.5*\ydist)+( 300:\size)$)--(c3);
            \draw[line width=2pt,lightgray] (d3)--($(-0.35cm+4*\xdist,\ydist+0.5*\yydist)+( 60:\size)$)--($(-0.35cm+4*\xdist,\ydist+0.5*\yydist)+(120:\size)$)--($(-0.35cm+4*\xdist,\ydist+0.5*\yydist)+( 180:\size)$)--($(-0.35cm+4*\xdist,\ydist+0.5*\yydist)+( 240:\size)$)--($(-0.35cm+4*\xdist,\ydist+0.5*\yydist)+( 300:\size)$)--(d3);
            \draw[line width=2pt,lightgray] (e3)--($(-0.35cm+5*\xdist,\ydist+0.5*\yydist)+( 60:\size)$)--($(-0.35cm+5*\xdist,\ydist+0.5*\yydist)+(120:\size)$)--($(-0.35cm+5*\xdist,\ydist+0.5*\yydist)+( 180:\size)$)--($(-0.35cm+5*\xdist,\ydist+0.5*\yydist)+( 240:\size)$)--($(-0.35cm+5*\xdist,\ydist+0.5*\yydist)+( 300:\size)$)--(e3);
            
        \end{tikzpicture}
        \caption{Construction of $H$ in the proof of \cref{thm:12}.}
        \label{fig:12}
    \end{figure}

    First, 
    let $c\colon E(G) \rightarrow [q]$ be a $q$-edge-colouring of $G$. We obtain a \type12-colouring $c_{\type12}:V(H) \rightarrow [q]$ of $H$ as follows. For every $uv\in E(G)$ we colour the two $K_{2q}$ of the two equal gadgets such that every colour appears twice and the anchors $\overline{a}^e$ and $\overline{b}^e$ have colour $c(e)$. We further set $c_{\type12}(w_{\pi(u)}^u)=c_{\type12}(x_{\pi(u)}^u)=c_{\type12}(y_{\pi(u)}^u)=c_{\type12}(z_{\pi(u)}^u)=c(e)$ and similarly for $v$. For every $i\in [q]$ and $v\in V(G)$ we colour the $3$ $K_{2q}$ of the equal gadgets such that each colour appears exactly twice and such that the anchors $a_i^v, b_i^v, c_i^v$ receive colour $c_{\type12}(w_i^v)$ (which was already fixed before). It is not difficult to verify that $c_{\type12}$ is a valid \type12-$q$-colouring (see also \cref{fig:12} for illustration).

    On the other hand, 
    let $c_{\type12}:V(H) \rightarrow [q]$ be a \type12-colouring of $H$. Fix a vertex $v \in V(G)$. We show the following.
    \begin{claim}
        For $i,j\in [q]$, $i\neq j$ it holds that $c_{\type12}(w_i^u)\neq c_{\type12}(w_j^u)$.
    \end{claim}
    \begin{claimproof}
        First observe that $y_i^v$ has degree $2q-1$ and must have exactly $2$ neighbours of every colour apart from its own and $1$ neighbour of its own colour. Since there is an equal gadget between $w_j^v$ and $x_j^v$ it must hold by \cref{claim:12equalGadget} that $c_{\type12}(w_j^v)=c_{\type12}(x_j^v)$. As $y_i^v$ is adjacent to both $x_j^v$ and $w_j^v$ and $y_i^v$ cannot have $2$ neighbours of its own colour we get that  $c_{\type12}(y_i^v)\neq c_{\type12}(w_j^v)$. As there is also an equal gadget between $y_i^v$ and $w_i^v$, by \cref{claim:12equalGadget} we conclude that $c_{\type12}(y_i^v)=c_{\type12}(w_i^v)\neq c_{\type12}(w_j^v)$, as claimed.
    \end{claimproof}
    Observe additionally that by \cref{claim:12equalGadget}, for every edge $e\in E(G)$ it holds that $c_{\type12}(w_{\pi(u)}^u)=c_{\type12}(w_{\pi(v)}^v)$. Using the claim above we conclude that the colouring $c\colon E(G)\rightarrow [q]$ defined by $c(uv)=c_{\type12}(w_{\pi(u)}^u)$ for every $uv\in E(G)$ is well defined and a valid edge-colouring of $G$.
\end{proof}

\subsubsection{\Type{1}{?} and \Type{v}{?}-colouring}
Recall that $c\colon V(G)\rightarrow [q]$ is a $\sigma$-$q$-colouring of $G$ for $\sigma\in \{\type{1}{?},\type{v}{?}\}$ if for every vertex $v$ it holds that $|N_G(v)\cap c^{-1}(i)|<2$ for every $i\not=c(v)$. Additionally, for \Type1?-$q$-colouring $|N_G(v)\cap c^{-1}(i)|$ is odd and for \Type{v}{?}-$q$-colouring $|N_G(v)\cap c^{-1}(i)|$ is odd or $0$.
\begin{theorem}\label{link3:1?}\label{link3:v?}
    For $q = 3$ it is \NP-complete to decide whether a graph admits a  \type1? and \type{v}{?}-$q$-colouring.
\end{theorem}
\begin{proof} 
    We give a reduction from $3$-edge-colouring on $3$-regular graphs. Given a $3$-regular graph $G$, we construct a graph $H$ as follows. For each vertex $v \in V(G)$, we introduce a triangle with vertices $x_1^v,x_2^v,x_3^v$. We subdivide each edge of the triangle by inserting a vertex $y_i^v$ between $x_i^v$ and $x_{i+1}^v$, for each $i \in [3]$, with indices modulo $3$. Additionally, for every vertex $v \in V(G)$, let $E_v$ be the set of the three edges incident to $v$. We arbitrarily choose a bijection $\pi_v: E_v \rightarrow \{x_1^v,x_2^v,x_3^v\}$, assigning a label to each edge. Finally, for each edge $e = uv \in E(G)$, we add the edge $\pi_u(e)\pi_v(e)$. This completes the construction of $H$, which is illustrated in Figure~\ref{fig:1questionANDvquestion}. Fix $\sigma \in \{\type{1}{?},\type{v}{?}\}$.

    \begin{figure}[h]
        \centering
        \begin{tikzpicture}[scale=0.4]
        \begin{scope}[xshift=-12cm]
            \node[a](m) at (0,0){};
            \node[l](x) at (90:2){};
            \node[l](y) at (210:2){};
            \node[l](z) at (330:2){};
            \draw[c1,lw1] (m)--(x);
            \draw[c2,lw1] (m)--(y);
            \draw[c3,lw1] (m)--(z);
            \draw[-{Latex[length=2.5mm, width=2.5mm]},decorate, decoration={snake, segment length=4mm, amplitude=0.6mm}, thick] (4,0)--(7.1,0);
        \end{scope}
        \node[r](x1) at (90:1){};
        \node[g](x2) at (210:1){};
        \node[b](x3) at (330:1){};
        \node[r](y1) at (150:2){};
        \node[g](y2) at (270:2){};
        \node[b](y3) at (30:2){};
        \node[r](p1) at (115:2.2){};
        \node[g](p2) at (235:2.2){};
        \node[b](p3) at (355:2.2){};
        \node[gr](xx) at (90:3.3){};
        \node[gg](yy) at (210:3.3){};
        \node[gb](zz) at (330:3.3){};
        \draw (x1)--(y1)--(x2)--(y2)--(x3)--(y3)--(x1);
        \draw (p1)--(x1)--(xx);
        \draw (p2)--(x2)--(yy);
        \draw (p3)--(x3)--(zz);
    \end{tikzpicture}
    \caption{Gadget representing vertex $v$ of $G$ in $H$, showing the triangle, the subdivision vertices, and the pendants.}
    \label{fig:1questionANDvquestion}
\end{figure}
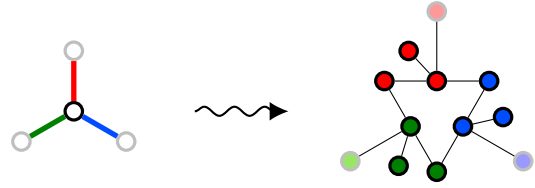

    First, 
    let $c\colon E(G) \rightarrow [3]$ be a 3-edge-colouring of $G$. We obtain a $\sigma$-colouring $c_{\sigma}:V(H) \rightarrow [3]$ of $H$ as follows. For every edge $e = uv \in E(G)$, we set $c_{\sigma}(\pi_u(e)) = c_{\sigma}(\pi_v(e)) = c(e)$. For each vertex $v \in V(G)$ and for each $i \in [3]$, we set $c_{\sigma}(p_i^v) = c_{\sigma}(y_i^v) = c(x_i^v)$. It is easy to verify that $c_{\sigma}$ is a $\sigma$-colouring, as illustrated in Figure~\ref{fig:1questionANDvquestion}.
    
    On the other hand, 
    let $c_{\sigma}:V(H) \rightarrow [3]$ be a $\sigma$-colouring of $H$. Fix a vertex $v \in V(G)$. First observe that the vertices $x_i^v$ have pairwise different colours as the subdivision vertices cannot be adjacent to precisely two vertices of the same colour. 
    As every $x_i^v$ has degree $4$, there is one colour appearing at least twice in the neighbourhood of $x_i^v$ which directly implies that $x_i^v$ is adjacent to precisely three vertices of its own colour, by the definition of $\sigma$-colouring.  Therefore, at least one of the subdivision vertices adjacent to $x_i^v$ has to be the same colour as $x_i^v$. This directly implies that either $c_\sigma(y_i^v)=c_\sigma(x_i^v)$ for all $i\in [3]$ or $c_\sigma(y_i^v)=c_\sigma(x_{i+1}^v)$ for all $i\in [3]$ with indices modulo $3$. In particular, one of the subdivision vertices adjacent of $x_i^v$ has a different colour to $x_i^v$ which implies that $c_\sigma(x_i^v)=c_\sigma(p_i^v)=c_\sigma(x_j^u)$ where $uv\in E(G)$ with $\pi_v(uv)=i$ and $\pi_u(uv)=j$. We conclude that 
    the   colouring $c\colon E(G) \rightarrow [3]$ of $G$ defined by setting $c(e) = c_{\sigma}(\pi_u(e)) = c_{\sigma}(\pi_v(e))$ is a valid $3$-edge colouring. 
\end{proof}

\begin{theorem}\label{link4:1?}\label{link4:v?}
    For $q \geq 4$ it is \NP-complete to decide whether a graph admits a  \type1? and \type{v}{?}-$q$-colouring.
\end{theorem}

\begin{proof} 
    We give a reduction from $q$-edge-colouring on $q$-regular graphs. Given a $q$-regular graph $G$, we construct a graph $H$ as follows. For each vertex $v \in V(G)$, we introduce a $K_q$ $X^v$ with vertices $x_1^v,\dots,x_q^v$. We subdivide the edges $x_1^vx_2^v$ with vertex $y^v$ and  attach a pendant $p^v$ to vertex $x_1^v$.
    For each vertex $v \in V(G)$, let $E_v$ be the set of the $q$ edges incident to $v$. We arbitrarily choose a bijection $\pi_v: E_v \rightarrow \{x_1^v,\dots,x_q^v\}$, assigning a label to each edge. Finally, for each edge $e = uv \in E(G)$, we add the edge $\pi_u(e)\pi_v(e)$. For illustration see Figure~\ref{fig:1questionANDvquestion}. Fix $\sigma \in \{\type{1}{?},\type{v}{?}\}$.

    First, 
    let $c\colon E(G) \rightarrow [q]$ be an edge-colouring of $G$. We obtain a $\sigma$-colouring $c_{\sigma}:V(H) \rightarrow [q]$ of $H$ as follows. For every edge $e = uv \in E(G)$, we set $c_{\sigma}(\pi_u(e)) = c_{\sigma}(\pi_v(e)) = c(e)$. Additionally, for each vertex $v \in V(G)$ we set $c_{\sigma}(p^v) = c_{\sigma}(y^v) = c(x_1^v)$. It is easy to verify that $c_{\sigma}$ is a $\sigma$-colouring.
    
    On the other hand, 
    let $c_{\sigma}:V(H) \rightarrow [q]$ be a $\sigma$-colouring of $H$. Fix a vertex $v \in V(G)$. 
    
    \begin{claim}
        Every vertex $x_i^v$ has at least three neighbours of its own colour.
    \end{claim}
    \begin{claimproof}
       Since $x_i^v$ has degree $q+1$, there must be some colour $j\in [q]$ appearing at least twice in the neighbourhood of $x_i^v$. By definition of $\sigma$-colouring this implies $c_\sigma(x_i^v)=j$. Since $x_i^v$ must have an odd number of neighbours of colour $j$, we conclude that a third neighbour of $x_i^v$ must have colour $j$. 
    \end{claimproof}
    \begin{claim}\label{claim:different1?v?}
        The vertices $x_1^v,\dots, x_q^v$ receive pairwise different colours.
    \end{claim}
    \begin{claimproof}
        First observe that $x_1^v$ and $x_2^v$  receive different colours as they are adjacent to $y^v$ of degree $2$ which forces them to have different colours. 
        First assume that $c_\sigma(x_i^v)=c_\sigma(x_j^v)$ for $i\not=j$, $i,j\notin \{1,2\}$. In this case both $x_1^v$ and $x_2^v$ are adjacent to $x_i^v$ and $x_j^v$. Since $x_1^v$ and $x_2^v$ have different colours, this contradicts the definition of $\sigma$-colouring. Next presume $c_\sigma(x_i^v)=c_\sigma(x_j^v)$ for $i\in \{1,2\}$, $j\notin \{1,2\}$. Then there is $x_k^v$, $k\geq 3$ with $k\not=j$. Since $x_j^v$ and $x_k^v$ have different colours by our previous argument and $x_k^v$ is adjacent to $x_i^v$ and $x_j^v$, this contradicts the definition of $\sigma$-colouring. 
    \end{claimproof}
    Since the degree of $x_1^v$ is $q+1$, one colour must appear at least twice in the neighbourhood of $x_1^v$. Combining this with \cref{claim:different1?v?} implies that $c_\sigma(x_1^v)=c_\sigma(p^v)=c_\sigma(x_j^u)$ where $uv\in E(G)$ with $\pi_v(uv)=1$ and $\pi_u(uv)=j$. Similarly, since $_i^v$ for $i\geq 2$ has degree $q$, \cref{claim:different1?v?} implies that $c_\sigma(x_1^v)=c_\sigma(x_j^u)$ where $uv\in E(G)$ with $\pi_v(uv)=i$ and $\pi_u(uv)=j$. Hence, the colouring $c\colon E(G) \rightarrow [q]$ defined by setting $c(e) = c_{\sigma}(\pi_u(e)) = c_{\sigma}(\pi_v(e))$ for every edge $e\in E(G)$ is a valid edge-colouring of $G$.
\end{proof}


\subsubsection{\Type{0}{v}-colouring}
In this section we consider \type0v-colouring for at least $3$ colours. We remind the reader that colouring $c\colon V(G)\rightarrow [q]$ of a graph $G$ is a \type0v-colouring if every vertex $v$ of $G$ is adjacent to an even number of vertices that receive the same colour as $v$ and for every other colour in $[q]$, $v$ is either adjacent to none or an odd number of vertices of that colour.
\begin{theorem}\label{thm:0v}\label{link34:0v}
    \type0v-$q$-colouring is \NP-complete for every $q \geq 3$.
\end{theorem}
We first observe the following simple consequence of the definition of \type0v-colouring.
\begin{observation}\label{obs:0v-clique}
    Let $A=\{a_1,\dots, a_t\}$ be the vertices of a clique in $G$ and $c\colon V(G)\rightarrow [q]$ a \type0v-colouring of $G$. If $a_1$ has degree $t-1$ in $G$, then every colour class of $c$ restricted to $A$ has either odd size or is empty.  
\end{observation}
The following observation is an easy consequence of \cref{obs:0v-clique}.
\begin{observation}\label{obs:0v-cliqueNeighbours}
    Let $A=\{a_1,\dots, a_t\}$ be the vertices of a clique in $G$ and $c\colon V(G)\rightarrow [q]$ a \type0v-colouring of $G$. If $a_1$ has degree $t-1$ in $G$ and $a_2$ has degree $t$ and $u$ is the neighbour of $a_2$ not contained in $A$ then $u$ has a colour not appearing in $c(A)$.  
\end{observation}
With these observation we can now prove \cref{thm:0v}.
\begin{proof}[Proof of \cref{thm:0v}]
    We first build an equal and a non-equal gadget. The non-equal gadget between vertices $u$ and $v$ is depicted in \cref{fig:0v-first}. In particular, we say that in a graph $G$, vertex $u\in V(G)$ is connected to vertex $v\in V(G)$ by a non-equal gadget (on vertices $u,v,w,w_1,w_2,w_3$) if vertices $w_1,w_2,w_4,w$ do not have neighbours outside the non-equal gadget while $u$ and $v$ may have neighbours outside the non-equal gadget. Note that the non-equal gadget is not symmetric which will be of little relevance in the construction later.
    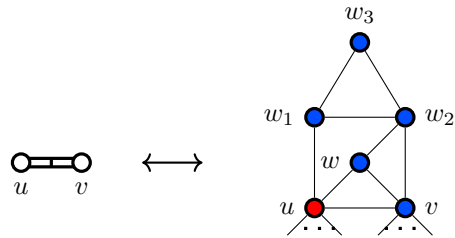
\begin{figure}[h]
        \centering
        \begin{tikzpicture}[scale=0.4]
            \tikzstyle{lw2}=[line width=1.2]
            \def \xdist {1.5cm}
            \def \ydist {1.2cm}
            \begin{scope}[xshift=-11.2cm]
                \draw[lw2] (0,0.14)--(2,0.14);
                \draw[lw2] (0,-0.14)--(2,-0.14);
                \draw[lw2] (1,0.14)--(1,-0.14);
                \node[a,fill=white,label=below:$u$](u) at (0,0){};
                \node[a,fill=white,label=below:$v$](v) at (2,0){};
                \draw[<->,thick] (4,0)--(6,0);
            \end{scope}
            \node[r,label=left:$u$](uu) at (-1*\xdist,-1*\xdist){};
            \node[b,label=right:$v$](vv) at (1*\xdist,-1*\xdist){};
            \node[b, label=left:$w$](w1) at (0,0){};
            \node[b, label=left:$w_1$](w2) at (-1*\xdist,1*\xdist){};
            \node[b, label=right:$w_2$](w3) at (1*\xdist,1*\xdist){};
            \node[b, label=above:$w_3$](w4) at (0,2.65*\xdist){};
            \draw (uu)--(w1)--(vv)--(uu)--(w2)--(w3)--(w1);
            \draw (vv)--(w3)--(w4)--(w2);
            \draw (-0.4*\xdist,-1.6*\xdist)--(uu)--(-1.6*\xdist,-1.6*\xdist);
            \draw[loosely dotted, lw2] (-0.55*\xdist,-1.45*\xdist)--(-1.45*\xdist,-1.45*\xdist);
            \draw (0.4*\xdist,-1.6*\xdist)--(vv)--(1.6*\xdist,-1.6*\xdist);
            \draw[loosely dotted, lw2] (0.55*\xdist,-1.45*\xdist)--(1.45*\xdist,-1.45*\xdist);
    
        \end{tikzpicture}
        \caption{Non-equal gadget connecting vertex $u$ to vertex $v$. In a \type0v-colouring, the vertices $u$ and $v$ receive different colours.}
        \label{fig:0v-first}
    \end{figure}
    \begin{claim}\label{claim:0v-first}
        Consider a non-equal gadget connecting vertex $u$ to vertex $v$. In a \type0v-colouring, the vertices $u$ and $v$ receive different colours. Furthermore, given different colours for $u$ and $v$ there is a \type0v-$2$-colouring of the non-equal gadget in which $u$ has three neighbours of the colour of $v$ and $v$ has no neighbours of the colour of $u$.
    \end{claim}
    \begin{claimproof}
        Let $G$ be a graph in which vertex $u$ is connected to vertex $v$  by a non-equal gadget and $c\colon V(G)\rightarrow [q]$ a \type0v-colouring of $G$.
        Since $w_3$ has degree $2$, by \cref{obs:0v-cliqueNeighbours} $u$ must have a colour which does not appear in  $c(\{w_1,w_2,w_3\})$. Since $c(w_2)\not=c(u)$ we obtain from considering the neighbourhood of $w$ that $c(u)=c(v)$ implies that $c(w)=c(u)=c(v)$. But in this case, $w_2$ has two neighbours in colour $c(u)$, a contradiction. Hence, $c(u)\not= c(v)$. The colouring in the second statement is depicted in \cref{fig:0v-first}.
    \end{claimproof}
    Now fix $q\geq 3$. We build the equal gadgets as follows. We first introduce vertices $u,v,w_1,\dots, w_{q+1}$. 
    We add a non-equal gadget connecting every pair $w_i,w_j$ with $4\leq i<j\leq q+1$. The orientation of non-equal gadgets is arbitrary. Additionally, we add edges $w_iw_j$ for $1\leq i\leq 3$, $i<j\leq q+1$ and the edges $w_2u$ and $w_3v$ (see \cref{fig:0v-second} for an illustration). We say that $u,v$ are connected by an equal gadget (on vertices $u,v,w_1,\dots, w_{q+1}$) in a graph $G$ if $w_1,\dots, w_{q+1}$ have no neighbours outside the equal gadget while $u$ and $v$ might have an arbitrary number of neighbours outside the gadget. 
    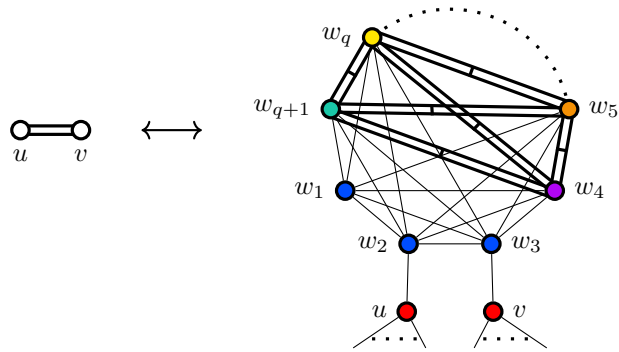
\begin{figure}[h]
        \centering
        \begin{tikzpicture}[scale=0.4]
            \tikzstyle{lw2}=[line width=1.2]
            \def \xdist {2cm}
            \def \r {4cm}
            \def \rr {4.15cm}
            \def \rrr {3.85cm}
            \begin{scope}[xshift=-14.2cm]
                \draw[lw2] (0,0.14)--(2,0.14);
                \draw[lw2] (0,-0.14)--(2,-0.14);
                \node[a,fill=white,label=below:$u$](u) at (0,0){};
                \node[a,fill=white,label=below:$v$](v) at (2,0){};
                \draw[<->,thick] (4,0)--(6,0);
            \end{scope}
            \draw[lw2] (130:1.04*\r)--(170:1.04*\r);
            \draw[lw2] (130:0.96*\r)--(170:0.96*\r);
            \draw[lw2] (150:0.89*\r)--(150:0.98*\r);
            \draw[lw2] (330:1.04*\r)--(370:1.04*\r);
            \draw[lw2] (330:0.96*\r)--(370:0.96*\r);
            \draw[lw2] (350:0.89*\r)--(350:0.98*\r);
            \draw[lw2] (128:1.05*\r)--(372:1.05*\r);
            \draw[lw2] (130:0.96*\r)--(370:0.96*\r);
            \draw[lw2] (250:0.24*\r)--(250:0.16*\r);
            \draw[lw2] (327:1.05*\r)--(173:1.05*\r);
            \draw[lw2] (330:0.96*\r)--(170:0.96*\r);
            \draw[lw2] (70:0.47*\r)--(70:0.55*\r);
            \draw[lw2] (130:1.05*\r)--(332:1.05*\r);
            \draw[lw2] (133:0.95*\r)--(329:0.95*\r);
            \draw[lw2] (0:0.19*\r)--(13:0.25*\r);
            \draw[lw2] (370:1.05*\r)--(168:1.05*\r);
            \draw[lw2] (367:0.95*\r)--(171:0.95*\r);
            \draw[lw2] (140:0.19*\r)--(127:0.25*\r);
            \draw[loosely dotted, lw2] (10:\r) arc (10:130:\r);
            \node[r,label=left:$u$](uu) at (-1.43,-6){};
            \node[r,label=right:$v$](vv) at (1.43,-6){};
            \node[ye, label=left:$w_{q}$](wq) at (130:\r){};
            \node[tq, label=left:$w_{q+1}$](wq1) at (170:\r){};
            \node[b, label=left:$w_1$](w1) at (210:\r){};
            \node[b, label=left:$w_2$](w2) at (250:\r){};
            \node[b, label=right:$w_3$](w3) at (290:\r){};
            \node[p, label=right:$w_4$](w4) at (330:\r){};
            \node[o, label=right:$w_5$](w5) at (370:\r){};
            \draw (uu)--(w2)--(w1)--(w3)--(w2);
            \draw (vv)--(w3);
            \draw (w1) -- (wq1);
            \draw (w1) -- (wq);
            \draw (w1) -- (w4);
            \draw (w1) -- (w5);
            \draw (w2) -- (wq1);
            \draw (w2) -- (wq);
            \draw (w2) -- (w4);
            \draw (w2) -- (w5);
            \draw (w3) -- (wq1);
            \draw (w3) -- (wq);
            \draw (w3) -- (w4);
            \draw (w3) -- (w5);
            \draw (-0.4*\xdist,-3.6*\xdist)--(uu)--(-1.6*\xdist,-3.6*\xdist);
            \draw[loosely dotted, lw2] (-0.55*\xdist,-3.45*\xdist)--(-1.3*\xdist,-3.45*\xdist);
            \draw (0.4*\xdist,-3.6*\xdist)--(vv)--(1.6*\xdist,-3.6*\xdist);
            \draw[loosely dotted, lw2] (0.55*\xdist,-3.45*\xdist)--(1.3*\xdist,-3.45*\xdist);
    
        \end{tikzpicture}
        \caption{Equal gadget between vertices $u$ and $v$. In a \type0v-colouring, the vertices $u$ and $v$ receive the same colours.}
        \label{fig:0v-second}
    \end{figure}
    \begin{claim}\label{claim:0v-second}
        Consider an equal gadget connecting two vertices $u$ and $v$. In a \type0v-colouring, the vertices $u$ and $v$ receive the same colour. 
        Furthermore, given a colour for $u$ and $v$ and a different colour for the neighbours of $u$ and $v$ in the equal gadget, we can complete this colouring into a \type0v-$q$-colouring of the non-equal gadget.
    \end{claim}
    \begin{claimproof}
        Let $G$ be a graph with vertices $u, v$ which are connected by an equal gadget and $c\colon V(G)\rightarrow [q]$ a \type0v-colouring of $G$.
        Since the non-equal gadget between two vertices $w_i$ and $w_j$ implies that $w_i$ and $w_j$ are adjacent, we observe that $\{w_1,\dots, w_{q+1}\}$ is a clique. We now consider the size of $c(\{w_1,\dots, w_{q+1}\})$. Since we pairwise connected $w_4,\dots, w_{q+1}$ by non-equal gadgets we know that $c(w_i)\not= c(w_j)$ for $4\leq i<j\leq q+1$. As $w_1$ has degree $q$, we can apply \cref{obs:0v-clique} and observe that there must be an additional colour appearing among vertices $w_1,w_2$ and $w_3$. Hence, $c(\{w_1,\dots, w_{q+1}\})$ contains at least $q-1$ colours. Since both $w_2$ and $w_3$ have precisely one neighbour not in the clique $\{w_1,\dots, w_{q+1}\}$ we can apply \cref{obs:0v-cliqueNeighbours} to conclude that $u$ and $v$ both have a colour not contained in  $c(\{w_1,\dots, w_{q+1}\})$ for which there is at most one option and hence $c(u)=c(v)$. 
        
        Finally,  we construct a \type0v-$q$-colouring $c\colon \{u,v,w_1,\dots,w_{q+1}\}\rightarrow [q]$ as follows. Without loss of generality, presume that we fixed colour $1$ for $u$ and $v$. First, set $c(w_1)=c(w_2)=c(w_3)=2$, and $c(w_i)=i-1$ for every $4\leq i\leq q+1$. If in the construction there is a non-equal gadget connecting $w_i$ to $w_j$, then we colour all remaining vertices in the gadget with $c(w_j)$. It is easy to verify that the condition for a valid \type0v-colouring is satisfied for $u,v,w_1,w_2,w_3$. Hence, consider $w_i$, $4\leq i\leq q+1$. Observe that in each non-equal gadget which connects some vertex $w_j$ to $w_i$, $w_i$ has two neighbours coloured $c(w_i)$. No other neighbours of $w_i$ receive colour $c(w_i)$ and hence $w_i$ has an even number of neighbours in its own colour. Furthermore, $w_i$ has no neighbours in colour $1$ and $3$-neighbours in colour $2$. Finally, $w_i$ has either one or three neighbours for every other colour depending on whether the respective non-equal gadget is oriented toward or away from $w_i$. Hence, $c$ is a \type0v-colouring of the non-equal gadget.
    \end{claimproof}
    We give a reduction from $q$-edge-colouring of $q$-regular graph. Given an instance $G$ of $q$-edge colouring, we construct a graph $H$ as follows. For each vertex $v\in V(G)$ we introduce $q$ vertices $x_1^v,\dots x_q^v$ which we pairwise connect with a non-equal gadget (we may assume that the non-equal gadget connects $x_i^v$ to $x_j^v$ with $i<j$). Now, for every vertex $v \in V(G)$, let $E_v$ be the set of the $q$ edges incident to $v$. We arbitrarily choose a bijection $\pi_v: E_v \rightarrow \{x_1^v,\dots, x_q^v\}$, assigning a label to each edge. For every edge $e = uv \in E(G)$ we introduce a new vertex $y_e$. We add equal gadget between  $\pi_u(e)$ and $\pi_v(e)$, between  $y_e$ and $\pi_u(e)$, and between $y_e$ and $\pi_v(e)$. Finally, for every edge $e\in E(G)$ we introduce another vertex $z_e$ which is just made to be a pendant of $y_e$. This completes the construction of $H$, which is illustrated in \cref{fig:0v-third}. 
    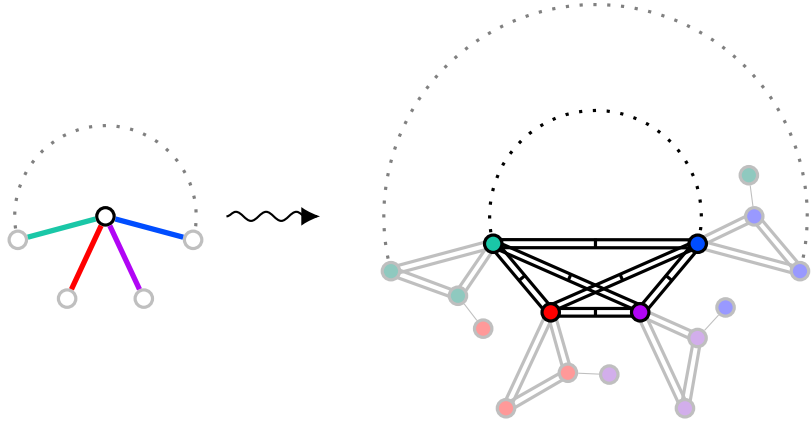
\begin{figure}[h]
        \centering
        \begin{tikzpicture}[scale=0.4]
            \tikzstyle{lw2}=[line width=1.2]
            \def \xdist {2cm}
            \def \r {3.5cm}
            \def \rr {7cm}
            \def \rrr {5.25cm}
            \begin{scope}[xshift=-16.2cm]
                \draw[loosely dotted, lw2,gray] (345:3) arc (345:555:3);
                \node[a](m) at (0,0){};
                \node[l, fill=white](w) at (195:3){};
                \node[l](x) at (245:3){};
                \node[l](y) at (295:3){};
                \node[l,fill=white](z) at (345:3){};
                \draw[c5,lw1] (m)--(w);
                \draw[c1,lw1] (m)--(x);
                \draw[c6,lw1] (m)--(y);
                \draw[c3,lw1] (m)--(z);
                \draw[-{Latex[length=2.5mm, width=2.5mm]},decorate, decoration={snake, segment length=4mm, amplitude=0.6mm}, thick] (4,0)--(7.1,0);
            \end{scope}
            \draw[lw2] (297:1.04*\r)--(195.5:1.04*\r);
            \draw[lw2] (298:0.96*\r)--(194.5:0.96*\r);
            \draw[lw2] (247:0.66*\r)--(247:0.59*\r);

            \draw[lw2] (242:0.96*\r)--(345.5:0.96*\r);
            \draw[lw2] (243:1.04*\r)--(344.5:1.04*\r);
            \draw[lw2] (293:0.66*\r)--(293:0.59*\r);
            
            \draw[lw2] (192.7:\r)--(347.2:\r);
            \draw[lw2] (197.2:\r)--(342.7:\r);
            \draw[lw2] (270:0.21*\r)--(270:0.3*\r);
            
            \draw[lw2] (245:1.04*\r)--(295:1.04*\r);
            \draw[lw2] (245:0.96*\r)--(295:0.96*\r);
            \draw[lw2] (270:0.86*\r)--(270:0.95*\r);

            \draw[lw2] (245:1.04*\r)--(195:1.04*\r);
            \draw[lw2] (245:0.96*\r)--(195:0.96*\r);
            \draw[lw2] (220:0.86*\r)--(220:0.95*\r);

            \draw[lw2] (345:1.04*\r)--(295:1.04*\r);
            \draw[lw2] (345:0.96*\r)--(295:0.96*\r);
            \draw[lw2] (320:0.86*\r)--(320:0.95*\r);

            \draw[lw2,lightgray] (193:\r)--(194:\rr);
            \draw[lw2,lightgray] (197:\r)--(196:\rr);

            \draw[lw2,lightgray] (243:\r)--(244:\rr);
            \draw[lw2,lightgray] (247:\r)--(246:\rr);

            \draw[lw2,lightgray] (293:\r)--(294:\rr);
            \draw[lw2,lightgray] (297:\r)--(296:\rr);

            \draw[lw2,lightgray] (343:\r)--(344:\rr);
            \draw[lw2,lightgray] (347:\r)--(346:\rr);

            \draw[lw2,lightgray] (192.5:\r)--(208.5:\rrr);
            \draw[lw2,lightgray] (197.5:\r)--(211.5:\rrr);

            \draw[lw2,lightgray] (242.5:\r)--(258.5:\rrr);
            \draw[lw2,lightgray] (247.5:\r)--(261.5:\rrr);

            \draw[lw2,lightgray] (292.5:\r)--(308.5:\rrr);
            \draw[lw2,lightgray] (297.5:\r)--(311.5:\rrr);

            \draw[lw2,lightgray] (342.5:\r)--(358.5:\rrr);
            \draw[lw2,lightgray] (347.5:\r)--(361.5:\rrr);

            \draw[lw2,lightgray] (193.7:\rr)--(208:\rrr);
            \draw[lw2,lightgray] (196.3:\rr)--(212:\rrr);

            \draw[lw2,lightgray] (243.7:\rr)--(258:\rrr);
            \draw[lw2,lightgray] (246.3:\rr)--(262:\rrr);

            \draw[lw2,lightgray] (293.7:\rr)--(308:\rrr);
            \draw[lw2,lightgray] (296.3:\rr)--(312:\rrr);

            \draw[lw2,lightgray] (343.7:\rr)--(358:\rrr);
            \draw[lw2,lightgray] (346.3:\rr)--(362:\rrr);

            \draw[loosely dotted, lw2] (345:\r) arc (345:555:\r);
            \draw[loosely dotted, lw2,gray] (345:\rr) arc (345:555:\rr);
            
            \node[tq](w1) at (195:\r){};
            \node[r](w2) at (245:\r){};
            \node[p](w3) at (295:\r){};
            \node[b](w4) at (345:\r){};

            \node[gtq](v1) at (195:\rr){};
            \node[gr](v2) at (245:\rr){};
            \node[gp](v3) at (295:\rr){};
            \node[gb](v4) at (345:\rr){};

            \node[gtq](u1) at (210:\rrr){};
            \node[gr](u2) at (260:\rrr){};
            \node[gp](u3) at (310:\rrr){};
            \node[gb](u4) at (360:\rrr){};

            \node[gr](x1) at (225:\rrr){};
            \node[gp](x2) at (275:\rrr){};
            \node[gb](x3) at (325:\rrr){};
            \node[gtq](x4) at (375:\rrr){};

            \draw[lightgray] (u1) -- (x1);
            \draw[lightgray] (u2) -- (x2);
            \draw[lightgray] (u3) -- (x3);
            \draw[lightgray] (u4) -- (x4);
        \end{tikzpicture}
        \caption{Construction of $H$ in the proof of \cref{thm:0v}.}
        \label{fig:0v-third}
    \end{figure}

    We now argue that the reduction is correct. Given a \type0v-$q$-colouring $c$ of $H$ we define a $q$-edge-colouring by setting the colour of edge $e\in E(G)$ be $c(y_e)$. By construction, this is a valid edge colouring. Hence, assume that $c\colon E(G)\rightarrow [q]$ is a $q$-edge-colouring of $G$. We obtain a \type0v-colouring $c_{\type0v}:V(H)\rightarrow [q]$ of $H$ as follows. For every edge $e=uv\in E(G)$ we set $c_{\type0v}(y_e)=c_{\type0v}(\pi_u(e))=c_{\type0v}(\pi_v(e))=c(e)$. Choose a colour $k_e\not=c(y_e)$ for the neighbours of $y_e, \pi_u(e)$ and $\pi_v(e)$ in the three equal gadgets connecting them. We colour the remaining vertices of the equal gadgets between $y_e, \pi_u(e)$ and $\pi_v(e)$ using the \type0v-colouring from \cref{claim:0v-second}. Additionally, we set $c(z_e)=k_e$. Since $c$ is an edge colouring, for every $v\in V(G)$ it holds that $c_{\type0v}(x_i^v)\not=c_{\type0v}(x_j^v)$ for $i\not=j$. Hence, using \cref{claim:0v-first}, we can \type0v-colour the  non-equal gadgets between pairs  $x_i^v$ and $x_j^v$ for $i\not=j$ with two colours.
    
    We are left to verify that $c_{\type0v}$ is a valid \type0v-colouring. Clearly, the property is satisfied for every vertex only appearing in an equal or non-equal gadget by \cref{claim:0v-first} and \cref{claim:0v-second}. Hence, consider a vertex $x_i^v$. For every non-equal gadget directed towards $x_i^v$ there are two vertices of colour $c_{\type0v}(x_i^v)$ and otherwise $x_i^v$ has no neighbours in the same colour. Additionally, depending on the direction of the non-equal gadgets, $x_i^v$ is adjacent to either one or three vertices of every other colour within the non-equal gadgets. Furthermore, for the edge $e$ with $\pi_v(e)=x_i^v$, $x_i^v$ has two neighbours (one per equal gadget) of colour $k_e\not=c_{\type0v}(x_i^v)$. Hence, $x_i^v$ is adjacent to an odd number of vertices in every colour apart from  $c_{\type0v}(x_i^v)$. Next, consider  vertex $y_e$. Since $y_e$ has three neighbours ($z_e$ and the one vertex per equal gadget), and these three vertices received colour $k_e\not=c_{\type0v}(y_e)$, the property is satisfied for $y_e$. Finally, every vertex $z_e$ has only one neighbour $y_e$ and $c_{\type0v}(z_e)=k_e\not=c_{\type0v}(y_e)$ concluding the proof.
\end{proof}

\subsubsection{\Type{2}{+}, \Type{2}{1} and \type{2}{!}-colouring}
Recall that $c\colon V(G)\rightarrow [q]$ is a $\sigma$-$q$-colouring of $G$ for $\sigma\in \{\type{2}{+},\type{2}{1},\type{2}{!}\}$ if for every vertex $v$ it holds that $|N_G(v)\cap c^{-1}(c(v))|$ is even and greater than $0$. Additionally, for \Type2+-$q$-colouring $|N_G(v)\cap c^{-1}(i)|>0$, for \Type{2}{1}-$q$-colouring $|N_G(v)\cap c^{-1}(i)|$ is odd, and for \Type{2}{!}-$q$-colouring $|N_G(v)\cap c^{-1}(i)|=1$ for every $i\not=c(v)$.
\begin{theorem}\label{link34:2+}\label{link34:21}\label{link34:2!}
    For $q \geq  3$ it is \NP-complete to decide whether a graph admits a  \type2+, a \type21, or a \type2!-$q$-colouring.
\end{theorem}

\begin{proof}
    We give a reduction from $q$-edge-colouring on $q$-regular graphs. Given a $q$-regular graph $G$, we construct a graph $H$ as follows. For each vertex $v \in V(G)$, we introduce two disjoint $K_q$ $X^v$ and $Y^v$ with vertices $x_1^v,\dots,x_q^v$ and $y_1^v,\dots, y_q^v$. We add the edges $x_i^vy_i^v$ for every $i\in [q]$.
    For each vertex $v \in V(G)$, let $E_v$ be the set of $q$ edges incident to $v$. We choose a bijection $\pi_v: E_v \rightarrow [q]$, assigning a label to each edge. For each edge $e = uv \in E(G)$ we add the edges $x_{\pi_u(e)}^ux_{\pi_v(e)}^v$ and $y_{\pi_u(e)}^uy_{\pi_v(e)}^v$. For an illustrated see Figure~\ref{fig:2+and21and2!}. 

    \begin{figure}[h]
  \centering
  \begin{tikzpicture}[scale=0.3]
  \def \dist {0.8cm}
  \begin{scope}[xshift=-14cm]
      \node[a](m) at (0,0){};
      \node[l](x) at (180:3){};
      \node[l](y) at (60:3){};
      \node[l](z) at (300:3){};
      \draw[c1,lw1](m)--(x);
      \draw[c3,lw1](m)--(y);
      \draw[c2,lw1](m)--(z);
      \draw[-{Latex[length=2.5mm, width=2.5mm]},decorate, decoration={snake, segment length=4mm, amplitude=0.6mm}, thick] (2,0)--(5.1,0);
  \end{scope}
    
    \foreach \r in {1,3} {
      \foreach[count=\n] \c in {r,g,b} {
        \pgfmathtruncatemacro{\a}{90+120*\n}
        \node[\c] (\n\r) at (\a:1.3*\dist*\r) {};
      }
      \draw (1\r)--(2\r)--(3\r)--(1\r);
    }
    \draw (11)--(13) (21)--(23) (31)--(33);
    \node[gb, xshift=\dist, yshift=1.5*\dist] (32) at (90:0.2) {};
    \node[gb, xshift=\dist, yshift=1.9*\dist] (34) at (90:1.4) {};
    \draw (31)--(32) (33)--(34) (32)--(34);
    \node[gg, xshift=\dist, yshift=-1.5*\dist] (22) at (330:0.2) {};
    \node[gg, xshift=1.5*\dist, yshift=-1.9*\dist] (24) at (330:1.4) {};
    \draw (21)--(22) (23)--(24) (22)--(24);
    \node[gr, xshift=-1.8*\dist, yshift=0.1*\dist] (12) at (210:0.2) {};
    \node[gr, xshift=-2.3*\dist, yshift=-0.1*\dist] (14) at (210:1.4) {};
    \draw (11)--(12) (13)--(14) (12)--(14);
    
  \end{tikzpicture}
  \caption{Gadget representing vertex $v$ of $G$ in $H$, showing triangles and the extra edge.}
    \label{fig:2+and21and2!}
\end{figure}
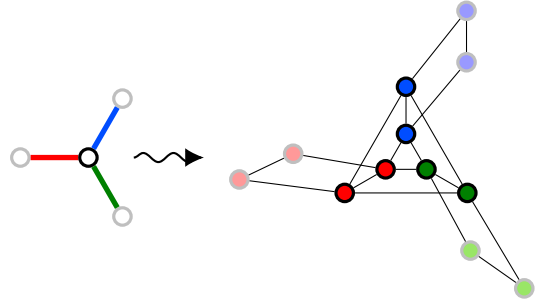

    Fix $\sigma \in \{\type{2}{+},\type{2}{1},\type{2}{!}\}$. Note that the graph $H$ is $(q+1)$-regular implying the following.
    \begin{claim}\label{claim:regImplication2xcol}
        In any $\sigma$-colouring every vertex is adjacent to exactly two vertices of its own colour and exactly one vertex of every remaining colour.
    \end{claim}
    First, 
    let $c\colon E(G) \rightarrow [q]$ be a $q$-edge-colouring of $G$. We obtain a $\sigma$-colouring $c_{\sigma}:V(H) \rightarrow [q]$ of $H$ as follows. For each edge $e = uv \in E(G)$ we set $c_{\sigma}(x_{\pi_u(e)}^u) = c_{\sigma}(y_{\pi_u(e)}^u)=c_{\sigma}(x_{\pi_v(e)}^v) = c_{\sigma}(y_{\pi_v(e)}^v) = c(e)$. It is easy to verify that $c_{\sigma}$ is a $\sigma$-colouring, as illustrated in Figure~\ref{fig:2+and21and2!}.
    On the other hand,
    let $c_{\sigma}:V(H) \rightarrow [q]$ be a $\sigma$-colouring of $H$. Fix a vertex $v \in V(G)$.
    \begin{claim}
        For all $q\geq 4$ no colour appears twice in $X^v$ and in $Y^v$.
    \end{claim}
    \begin{claimproof}
        Towards a contradiction, assume, without loss of generality, that $c_\sigma(x_1^v)=c_\sigma(x_2^v)=1$. Since $x_3^v$ is adjacent to both $x_1^v$ and $x_2^v$, we get $c_\sigma(x_3^v)=1$ by \cref{claim:regImplication2xcol}. But this yields a contradiction to \cref{claim:regImplication2xcol}, as $x_4^v$ is adjacent to three vertices of colour $1$. The argument for $Y^v$ is analogous.
    \end{claimproof}
    For three colours the argument is a little more involved, as there is a valid $\sigma$-colouring for one vertex gadget, but this colouring cannot be combined with any valid colouring of neighbouring vertex gadgets.
    \begin{claim}
        For all $q=3$ no colour appears twice in $X^v$ and in $Y^v$.
    \end{claim}
    \begin{claimproof}
        Towards a contradiction, we assume again that $c_\sigma(x_1^v)=c_\sigma(x_2^v)=1$ which implies $c_\sigma(x_3^v)=1$ by \cref{claim:regImplication2xcol}. Since $x_1^v$ has two neighbours of the same colour $y_1^v$ must receive a different colour, say $c_\sigma(y_1^v)=2$. By \cref{claim:regImplication2xcol} there must be two vertices of colour $2$ among the remaining neighbours $y_2^v$, $y_3^v$ and $y_i^u$ for some $i\in [3]$, $u\in V(G)$. Hence, either $y_2^v$ or  $y_3^v$ must receive colour $2$, say $c_\sigma(y_2^v)=2$. Since $y_3^v$ has two neighbours of colour $2$, we obtain $c_\sigma(y_3^v)=2$ by \cref{claim:regImplication2xcol}. By \cref{claim:regImplication2xcol}, this further implies that $c_\sigma(x_i^u)=c_\sigma(y_i^u)$ where $uv\in E(G)$ with $\pi_v(uv)=1$ and $\pi_u(uv)=i\in [3]$. Since $x_i^u$ has colour $3$, has one neighbour $x_1^v$ in colour $2$ and one neighbour $y_i^u$ in colour $3$, the remaining two vertices of $X^y$ must have colour $1$ and $3$. Hence, $X^u$ has two vertices of colour $3$ and one of colour $1$, yielding a contradiction to \cref{claim:regImplication2xcol}.
    \end{claimproof}
    As a direct consequence of the two claim we obtain that for every edge $e=uv\in E(G)$ it must hold $c_{\sigma}(x_{\pi_u(e)}^u)=c_{\sigma}(y_{\pi_u(e)}^u)=c_{\sigma}(x_{\pi_v(e)}^v)=c_{\sigma}(y_{\pi_v(e)}^v)$ using \cref{claim:regImplication2xcol} and for two edges $uv,vw\in E(G)$ it holds that $c_\sigma(x_{\pi_v(uv)}^v)\not=c_\sigma(x_{\pi_v(vw)}^v)$. Therefore, we obtain a valid $q$-edge colouring $c\colon E(G)\rightarrow [q]$ by setting $c(uv)=c_\sigma(x_{\pi_v(uv)}^v)=c_\sigma(x_{\pi_u(uv)}^u)$
\end{proof}

\subsection{Reductions from other colouring problems}
In this section we provide reductions from different colouring problems.

\subsubsection{\Type{=}{+} and \Type{=}{1}-colouring}
Recall that $c\colon V(G)\rightarrow [q]$ is a $\sigma$-$q$-colouring of $G$ for $\sigma\in \{\type=+,\type=1\}$ if for every vertex $v$ it holds that $|N_G(v)\cap c^{-1}(c(v))|=0$. Additionally, for \type=+-$q$-colouring $|N_G(v)\cap c^{-1}(i)|>0$ and for \type=1-$q$-colouring $|N_G(v)\cap c^{-1}(i)|$ is odd for every $i\not=c(v)$.

\begin{theorem}\label{link34:=+}\label{link34:=1}
    For $q \geq 3$ it is \NP-complete to decide whether a graph admits a \type=+ or \type=1-$q$-colouring.
\end{theorem}

\begin{proof}
    For $q \geq 3$, we give a reduction for \type{=}{+}-$q$-colouring from \type{=}{*}-$q$-colouring and for \type{=}{1}-$q$-colouring from \type{=}{0}-$q$-colouring. We therefore define $f:\{\type{=}{*},\type{=}{0}\} \rightarrow \{\type{=}{+},\type{=}{1}\}$ as $f(\type{=}{*}) = \type{=}{+}$ and $f(\type{=}{0}) = \type{=}{1}$. Given a graph $G$, we construct a graph $H$ by adding, for each vertex $v \in V(G)$, a $K_{q-1}$ called $W^v$ with vertices $w_1^v,\ldots,w_{q-1}^v$, and the edge $vw_i^v$ for each $i \in [q-1]$. 
    
    First, 
    let $c_{\sigma}:V(G) \rightarrow [q]$ be a $\sigma$-$q$-colouring of $G$. We obtain a $f(\sigma)$-$q$-colouring $c_{f(\sigma)}:V(H) \rightarrow [q]$ of $H$ as follows. For each vertex $v \in V(G)$, we set $c_{f(\sigma)}(v) = c_{\sigma}(v)$. We assign each vertex in $W^v$ a unique colour among the remaining colour in $[q]\setminus \{c_{\sigma}(v)\}$. It is easy to verify that $c_{f(\sigma)}$ is a $f(\sigma)$-$q$-colouring.

    On the other hand, 
    let $c_{f(\sigma)}:V(H) \rightarrow [q]$ be a $f(\sigma)$-$q$-colouring of $H$. Fix a vertex $v \in V(G)$.
    Observe that the vertex $w_1^v$ has $q-1$ neighbours. Therefore, by definition of $f(\sigma)$-$q$-colouring, all colours apart from $c_{f(\sigma)}(w_1^v)$  appear in its neighbourhood exactly once. In particular, for $i \in [q-1]$, the colour of $w_i^v$ is different from the colour of $v$. Moreover, the vertices $w_i^v$ and $w_j^v$ must receive different colours for each $i \neq j$.
    We obtain a $\sigma$-$q$-colouring $c_{\sigma}:V(G) \rightarrow [q]$ of $G$ as follows. For each vertex $v \in V(G)$, we set $c_{\sigma}(v) = c_{f(\sigma)}(v)$. It is easy to verify that $c_{\sigma}$ is a $\sigma$-$q$-colouring.
\end{proof}

\subsubsection{\Type{0}{+} and \Type{0}{1}-colouring}

Recall that $c\colon V(G)\rightarrow [q]$ is a $\sigma$-$q$-colouring of $G$ for $\sigma\in \{\type0+,\type01\}$ if for every vertex $v$ it holds that $|N_G(v)\cap c^{-1}(c(v))|$ is even. Additionally, for \type0+-$q$-colouring $|N_G(v)\cap c^{-1}(i)|>0$ and for \type01-$q$-colouring $|N_G(v)\cap c^{-1}(i)|$ is odd for every $i\not=c(v)$.

\begin{theorem}\label{link3:0+}\label{link3:01}
    For $q = 3$ it is \NP-complete to decide whether a graph admits a \type0+ or \type01-$q$-colouring.
\end{theorem}

\begin{proof}
    For $q = 3$, we give a reduction for \type{0}{+}-$q$-colouring from \type{=}{+}-$q$-colouring and for \type{0}{1}-$q$-colouring from \type{=}{1}-$q$-colouring. We therefore define $f:\{\type{=}{+},\type{=}{1}\} \rightarrow \{\type{0}{+},\type{0}{1}\}$ as $f(\type{=}{+}) = \type{0}{+}$ and $f(\type{=}{1}) = \type{0}{1}$. Given a graph $G$, we construct a graph $H$ by subdividing each edge $e \in E(G)$ with a vertex $x^{e}$.
    
    First, 
    let $c_{\sigma}:V(G) \rightarrow [q]$ be a $\sigma$-$q$-colouring of $G$. We obtain a $f(\sigma)$-$q$-colouring $c_{f(\sigma)}:V(H) \rightarrow [q]$ of $H$ as follows. For each vertex $v \in V(G)$, we set $c_{f(\sigma)}(v) = c_{\sigma}(v)$. For each edge $uv \in E(G)$, we assign to the vertex $x^{uv}$ the unique colour in $[q] \setminus \{c_{\sigma}(u),c_{\sigma}(v)\}$. It is easy to verify that $c_{f(\sigma)}$ is a $f(\sigma)$-$q$-colouring.

    On the other hand, 
    let $c_{f(\sigma)}:V(H) \rightarrow [q]$ be a $f(\sigma)$-$q$-colouring of $H$. Fix any edge $uv \in E(G)$.
    Observe that, by definition of $f(\sigma)$-$q$-colouring, the vertex $x^{uv}$ forces $u$ and $v$ to have different colours. Moreover, the vertices $u$ and $v$ must receive the two colours in $[q] \setminus \{c_{f(\sigma)}(x^{uv})\}$. Hence, we obtain a $\sigma$-$q$-colouring $c_{\sigma}:V(G) \rightarrow [q]$ of $G$ by setting $c_{\sigma}(v) = c_{f(\sigma)}(v)$ for every $v \in V(G)$.
\end{proof}

\begin{corollary}\label{link3:v+}\label{link3:*+}\label{link3:?+}\label{link3:v1}\label{link3:*1}\label{link3:?1}
    Starting from \type=+, this reduction also works for \type{v}+, \type*+, \type?+. Analogously, starting from \type=1, this reduction also works for \type{v}1, \type*1, \type?1.
\end{corollary}

\section{Conclusions}
We systematically study the computational complexity of graph colouring with different combinations of parity constraints. We leave \type{v}{*}-colouring (\textsl{i.e.} every vertex is required to have either $0$ or an odd number of neighbours of its own colour) for $q\geq 3$ colours open. We believe that this variant is of particular interest. It is known that every graph admits an even colouring with at most $2$-colours \cite{Lovasz93}. However, if you restrict this variant to forbid a vertex to have zero neighbours of its own colour, it becomes \NP-complete (\cref{link34:2*} for $q\geq 3$ colours and \cref{link2:2*} for two colours). On the other hand, odd colouring is \NP-complete \cite{BelmonteS21}. However,  we were not able to even find an example of a graph requiring more than three colours to be \type{v}{*}-coloured. We therefore believe that three colours might always be sufficient which would imply that allowing vertices to have no neighbours of their own colour in odd colouring makes the problem easier. Finally, relaxing the odd number constraint to allow $0$ as well also features in strong odd colourings in the literature and we therefore believe that the problem is of interest.

Another direction of research that we believe may be of interest is to consider maximization of the number of colours rather then minimization. In many considered variants the monochromatic colouring is already permissible. Its not however clear, whether there also exists colourings which use strictly more than one colour. For such problems considering the maximization problem may be of interest. In the literature, the maximum number of colours in some colouring variant are referred to as achromatic numbers (while the classical minimum number of colours are referred to as chromatic numbers).

\bibliography{bib}

@article{Chvatal,
  author       = {Va\v{s}ek Chv\'{a}tal},
  title        = {Tough graphs and hamiltonian circuits},
  journal      = {Discret.\ Math.},
  volume       = {5},
  number       = {3},
  pages        = {215--228},
  year         = {1973},
  url          = {https://doi.org/10.1016/0012-365X(73)90138-6},
  doi          = {10.1016/0012-365X(73)90138-6},
  bibsource    = {dblp computer science bibliography, https://dblp.org}
}

@book{G&J,
  author       = {M.\,R.~Garey and.
                  David S.~Johnson},
  title        = {Computers and Intractability: {A} Guide to the Theory of {NP}-Completeness},
  publisher    = {W. H. Freeman},
  year         = {1979},
  isbn         = {0-7167-1044-7},
  bibsource    = {dblp computer science bibliography, https://dblp.org}
}

@article{Holyer,
  author       = {Ian Holyer},
  title        = {The {NP}-Completeness of Edge-Coloring},
  journal      = {{SIAM} J.\ Comput.},
  volume       = {10},
  number       = {4},
  pages        = {718--720},
  year         = {1981},
  url          = {https://doi.org/10.1137/0210055},
  doi          = {10.1137/0210055},
  bibsource    = {dblp computer science bibliography, https://dblp.org}
}

@article{LeGa,
  author       = {Daniel Leven and
                  Zvi Galil},
  title        = {{NP}-completeness of Finding the Chromatic Index of Regular Graphs},
  journal      = {J.\ Algorithms},
  volume       = {4},
  number       = {1},
  pages        = {35--44},
  year         = {1983},
  url          = {https://doi.org/10.1016/0196-6774(83)90032-9},
  doi          = {10.1016/0196-6774(83)90032-9},
  bibsource    = {dblp computer science bibliography, https://dblp.org}
}

@book{Lovasz93,
  author    = {L{\'{a}}szl{\'{o}} Lov{\'{a}}sz},
  title     = {Combinatorial Problems and Exercises},
  publisher = {North-Holland},
  year      = {1993},
  isbn      = {9780821869475}
}

@article{BelmonteS21,
  author    = {R{\'{e}}my Belmonte and
               Ignasi Sau},
  title     = {On the Complexity of Finding Large Odd Induced Subgraphs and Odd Colorings},
  journal   = {Algorithmica},
  volume    = {83},
  number    = {8},
  pages     = {2351--2373},
  year      = {2021},
  url       = {https://doi.org/10.1007/s00453-021-00830-x},
  doi       = {10.1007/s00453-021-00830-x},
  bibsource = {dblp computer science bibliography, https://dblp.org}
}

@article{Scott01,
  author    = {Alex D.\ Scott},
  title     = {On Induced Subgraphs with All Degrees Odd},
  journal   = {Graphs Comb.},
  volume    = {17},
  number    = {3},
  pages     = {539--553},
  year      = {2001},
  url       = {https://doi.org/10.1007/s003730170028},
  doi       = {10.1007/s003730170028},
  bibsource = {dblp computer science bibliography, https://dblp.org}
}

@misc{PetrS21,
  author    = {Petrusevski, Mirko and Skrekovski, Riste},
  title     = {Colorings with neighborhood parity condition},
  publisher = {arXiv},
  doi       = {10.48550/ARXIV.2112.13710},
  url       = {https://arxiv.org/abs/2112.13710},
  year      = {2021}, 
  copyright = {Creative Commons Zero v1.0 Universal}
}

@article{cowen1997defective,
  title={Defective coloring revisited},
  author={Cowen, Lenore and Goddard, Wayne and Jesurum, C.~Esther}, 
  journal={Journal of Graph Theory},
  volume={24},
  number={3},
  pages={205--219},
  year={1997},
  publisher={Wiley Online Library},
  doi = {https://doi.org/10.1002/(SICI)1097-0118(199703)24:3<205::AID-JGT2>3.0.CO;2-T},
  url = {https://onlinelibrary.wiley.com/doi/abs/10.1002/\%28SICI\%291097-0118\%28199703\%2924\%3A3\%3C205\%3A\%3AAID-JGT2\%3E3.0.CO\%3B2-T},
  eprint = {https://onlinelibrary.wiley.com/doi/pdf/10.1002/\%28SICI\%291097-0118\%28199703\%2924\%3A3\%3C205\%3A\%3AAID-JGT2\%3E3.0.CO\%3B2-T}
}

@article{kratochvil1997covering,
  title={Covering regular graphs},
  author={Kratochv\'{\i}l, Jan and Proskurowski, Andrzej and Telle, Jan Arne},
  journal={Journal of Combinatorial Theory, Series B},
  volume={71},
  number={1},
  pages={1--16},
  year={1997},
  publisher={Elsevier}
}

@article{bok2024list,
  title={List covering of regular multigraphs with semi-edges},
  author={Bok, Jan and Fiala, Ji{\v{r}}{\'\i} and Jedli{\v{c}}kov{\'a}, Nikola and Kratochv{\'\i}l, Jan and Rz{\k{a}}{\.z}ewski, Pawe{\l}},
  journal={Algorithmica},
  volume={86},
  number={3},
  pages={782--807},
  year={2024},
  publisher={Springer}
}

@article{fiala2008locally,
  title={Locally constrained graph homomorphisms—structure, complexity, and applications},
  author={Fiala, Ji{\v{r}}{\'\i} and Kratochv{\'\i}l, Jan},
  journal={Computer Science Review},
  volume={2},
  number={2},
  pages={97--111},
  year={2008},
  publisher={Elsevier}
}

@inproceedings{kratochvil1994complexity,
  title={Complexity of graph covering problems},
  author={Kratochv{\'\i}l, Jan and Proskurowski, Andrzej and Telle, Jan Arne},
  booktitle={International Workshop on Graph-Theoretic Concepts in Computer Science},
  pages={93--105},
  year={1994},
  organization={Springer}
}

@article{kwon2024strong,
  title={Strong odd coloring of sparse graphs},
  author={Kwon, Hyemin and Park, Boram},
  journal={arXiv preprint arXiv:2401.11653},
  year={2024}
}

@article{CaroPST25,
  author       = {Yair Caro and
                  Mirko Petrusevski and
                  Riste Skrekovski and
                  Zsolt Tuza},
  title        = {On strong odd colorings of graphs},
  journal      = {Discret. Math.},
  volume       = {348},
  number       = {11},
  pages        = {114601},
  year         = {2025},
  url          = {https://doi.org/10.1016/j.disc.2025.114601},
  doi          = {10.1016/J.DISC.2025.114601},
  bibsource    = {dblp computer science bibliography, https://dblp.org}
}

@article{Pilipczuk25,
  author       = {Micha\l\ Pilipczuk},
  title        = {Strong Odd Colorings in Graph Classes of Bounded Expansion},
  journal      = {Electron. J. Comb.},
  volume       = {32},
  number       = {4},
  year         = {2025},
  url          = {https://doi.org/10.37236/14259},
  doi          = {10.37236/14259},
  bibsource    = {dblp computer science bibliography, https://dblp.org}
}

@article{fall,
    AUTHOR = {Dunbar, J. E. and Hedetniemi, S. M. and Hedetniemi, S. T. and
              Jacobs, D. P. and Knisely, J. and Laskar, R. C. and Rall, D. F.},
     TITLE = {Fall colorings of graphs},
      NOTE = {Papers in honour of Ernest J. Cockayne},
   JOURNAL = {J. Combin. Math. Combin. Comput.},
  FJOURNAL = {Journal of Combinatorial Mathematics and Combinatorial Computing},
    VOLUME = {33},
      YEAR = {2000},
     PAGES = {257--273},
      ISSN = {0835-3026,2817-576X}
}

@misc{b&fall,
      title={Optimal b-Colourings and Fall Colourings in {$H$}-Free Graphs}, 
      author={Jungho Ahn and Tala Eagling-Vose and Felicia Lucke and David Manlove and Fabricio Mendoza and Dani\"{e}l Paulusma},
      year={2026},
      eprint={2603.26214},
      archivePrefix={arXiv},
      primaryClass={math.CO},
      url={https://arxiv.org/abs/2603.26214}
}

@article{b-col,
  author  = {Robert W. Irving and David F. Manlove},
  title   = {The b-chromatic number of a graph},
  journal = {Discrete Applied Mathematics},
  volume  = {91},
  number  = {1},
  pages   = {127--141},
  year    = {1999},
  issn    = {0166-218X},
  doi     = {https://doi.org/10.1016/S0166-218X(98)00146-2},
  url     = {https://www.sciencedirect.com/science/article/pii/S0166218X98001462}
}

@book{ApHa1989,
  title     = {Every planar map is four colorable},
  author    = {Appel, Kenneth I. and Haken, Wolfgang},
  volume    = {98},
  year      = {1989},
  publisher = {American Mathematical Soc.}
}

@misc{McKay2012,
      title={A note on the history of the four-colour conjecture}, 
      author={Brendan D. McKay},
      year={2012},
      eprint={1201.2852},
      archivePrefix={arXiv},
      primaryClass={math.CO},
      url={https://arxiv.org/abs/1201.2852}, 
}

@Inbook{Karp2010,
author="Karp, Richard M.",
editor="J{\"u}nger, Michael
and Liebling, Thomas M.
and Naddef, Denis
and Nemhauser, George L.
and Pulleyblank, William R.
and Reinelt, Gerhard
and Rinaldi, Giovanni
and Wolsey, Laurence A.",
title="Reducibility Among Combinatorial Problems",
bookTitle="50 Years of Integer Programming 1958-2008: From the Early Years to the State-of-the-Art",
year="2010",
publisher="Springer Berlin Heidelberg",
address="Berlin, Heidelberg",
pages="219--241",
isbn="978-3-540-68279-0",
doi="10.1007/978-3-540-68279-0_8",
url="https://doi.org/10.1007/978-3-540-68279-0_8"
}

@incollection{frick1993survey,
  title={A survey of $(m,k)$-colorings},
  author={Frick, Marietjie},
  booktitle={Annals of Discrete Mathematics},
  volume={55},
  pages={45--57},
  year={1993},
  publisher={Elsevier}
}

@article{woodall1990improper,
  title={Improper colourings of graphs},
  author={Woodall, D.R.},
  journal={Graph colourings},
  volume={218},
  pages={45--63},
  year={1990},
  publisher={Longman Scientific and Technical Harlow}
}

@inproceedings{Sch78,
author = {Schaefer, Thomas J.},
title = {The complexity of satisfiability problems},
year = {1978},
isbn = {9781450374378},
publisher = {Association for Computing Machinery},
address = {New York, NY, USA},
url = {https://doi.org/10.1145/800133.804350},
doi = {10.1145/800133.804350},
booktitle = {Proceedings of the Tenth Annual ACM Symposium on Theory of Computing},
pages = {216–226},
numpages = {11},
location = {San Diego, California, USA},
series = {STOC '78}
}

@article{cheilaris2013unique,
  author       = {Panagiotis Cheilaris and
                  Bal{\'{a}}zs Keszegh and
                  D{\"{o}}m{\"{o}}t{\"{o}}r P{\'{a}}lv{\"{o}}lgyi},
  title        = {Unique-Maximum and Conflict-Free Coloring for Hypergraphs and Tree
                  Graphs},
  journal      = {{SIAM} J. Discret. Math.},
  volume       = {27},
  number       = {4},
  pages        = {1775--1787},
  year         = {2013},
  url          = {https://doi.org/10.1137/120880471},
  doi          = {10.1137/120880471},
  bibsource    = {dblp computer science bibliography, https://dblp.org}
}

@article{DemaineKP25,
  author       = {Erik D. Demaine and
                  Kritkorn Karntikoon and
                  Nipun Pitimanaaree},
  title        = {2-Colorable Perfect Matching is {NP}-complete in 2-Connected 3-Regular
                  Planar Graphs},
  journal      = {Theory Comput. Syst.},
  volume       = {69},
  number       = {2},
  pages        = {22},
  year         = {2025},
  url          = {https://doi.org/10.1007/s00224-025-10221-2},
  doi          = {10.1007/S00224-025-10221-2},
  bibsource    = {dblp computer science bibliography, https://dblp.org}
}

@article{fabrici2023proper,
  title={Proper conflict-free and unique-maximum colorings of planar graphs with respect to neighborhoods},
  author={Fabrici, Igor and Lu{\v{z}}ar, Borut and Rindo{\v{s}}ov{\'a}, Simona and Sot{\'a}k, Roman},
  journal={Discrete Applied Mathematics},
  volume={324},
  pages={80--92},
  year={2023},
  publisher={Elsevier}
}

@article{caro2022remarks,
  title={Remarks on odd colorings of graphs},
  author={Caro, Yair and Petru{\v{s}}evski, Mirko and {\v{S}}krekovski, Riste},
  journal={Discrete Applied Mathematics},
  volume={321},
  pages={392--401},
  year={2022},
  publisher={Elsevier}
}

@article{ahn2025proper,
  title={The proper conflict-free k-coloring problem and the odd k-coloring problem are {NP}-complete on bipartite graphs},
  author={Ahn, Jungho and Im, Seonghyuk and Oum, Sang-il},
  journal={Discrete Applied Mathematics},
  volume={377},
  pages={10--17},
  year={2025},
  publisher={Elsevier}
}

\appendix

 \section{Trivial cases}\label{sec:trivial}

\subsection{Case \boldmath$q = 1$}

In this subsection, we denote $\Sigma = \{\typesymbol{0},\typesymbol{2},\typesymbol{1},\typesymbol{v},\typesymbol{*},\typesymbol{=},\typesymbol{+},\typesymbol{?},\typesymbol{!}\}$. Fix $\sigma \in \Sigma$.

\begin{description}
	
	\item[type \typesymbol{0}$\sigma$]\label{link1:00}\label{link1:02}\label{link1:01}\label{link1:0v}\label{link1:0*}\label{link1:0=}\label{link1:0+}\label{link1:0?}\label{link1:0!} For $q = 1$ and $\sigma \in \Sigma$, a graph $G$ is \typesymbol{0}$\sigma$-$q$-colourable if and only if all connected components are even, which can be decided in polynomial time.
	
	\item[type \typesymbol{2}$\sigma$]\label{link1:20}\label{link1:22}\label{link1:21}\label{link1:2v}\label{link1:2*}\label{link1:2=}\label{link1:2+}\label{link1:2?}\label{link1:2!} For $q = 1$ and $\sigma \in \Sigma$, a graph $G$ is \typesymbol{2}$\sigma$-$q$-colourable if and only if all connected components are even and the graph has no isolated vertices, which can be decided in polynomial time.
	
	\item[type \typesymbol{1}$\sigma$]\label{link1:10}\label{link1:12}\label{link1:11}\label{link1:1v}\label{link1:1*}\label{link1:1=}\label{link1:1+}\label{link1:1?}\label{link1:1!} For $q = 1$ and $\sigma \in \Sigma$, a graph $G$ is \typesymbol{1}$\sigma$-$q$-colourable if and only if all connected components are odd, which can be decided in polynomial time.
	
	\item[type \typesymbol{v}$\sigma$]\label{link1:v0}\label{link1:v2}\label{link1:v1}\label{link1:vv}\label{link1:v*}\label{link1:v=}\label{link1:v+}\label{link1:v?}\label{link1:v!} For $q = 1$ and $\sigma \in \Sigma$, a graph $G$ is \typesymbol{v}$\sigma$-$q$-colourable if and only if all connected components are odd or isolated vertices, which can be decided in polynomial time.
	
	\item[type \typesymbol{*}$\sigma$]\label{link1:*0}\label{link1:*2}\label{link1:*1}\label{link1:*v}\label{link1:**}\label{link1:*=}\label{link1:*+}\label{link1:*?}\label{link1:*!} For $q = 1$ and $\sigma \in \Sigma$, every graph $G$ is \typesymbol{*}$\sigma$-$q$-colourable.
	
	\item[type \typesymbol{=}$\sigma$]\label{link1:=0}\label{link1:=2}\label{link1:=1}\label{link1:=v}\label{link1:=*}\label{link1:==}\label{link1:=+}\label{link1:=?}\label{link1:=!} For $q = 1$ and $\sigma \in \Sigma$, a graph $G$ is \typesymbol{=}$\sigma$-$q$-colourable if and only if all connected components are isolated vertices, which can be decided in polynomial time.
	
	\item[type \typesymbol{+}$\sigma$]\label{link1:+0}\label{link1:+2}\label{link1:+1}\label{link1:+v}\label{link1:+*}\label{link1:+=}\label{link1:++}\label{link1:+?}\label{link1:+!} For $q = 1$ and $\sigma \in \Sigma$, a graph $G$ is \typesymbol{+}$\sigma$-$q$-colourable if and only if the graph has no isolated vertices, which can be decided in polynomial time.
	
	\item[type \typesymbol{?}$\sigma$]\label{link1:?0}\label{link1:?2}\label{link1:?1}\label{link1:?v}\label{link1:?*}\label{link1:?=}\label{link1:?+}\label{link1:??}\label{link1:?!} For $q = 1$ and $\sigma \in \Sigma$, a graph $G$ is \typesymbol{?}$\sigma$-$q$-colourable if and only if all connected components are edges or isolated vertices, which can be decided in polynomial time.
	
	\item[type \typesymbol{!}$\sigma$]\label{link1:!0}\label{link1:!2}\label{link1:!1}\label{link1:!v}\label{link1:!*}\label{link1:!=}\label{link1:!+}\label{link1:!?}\label{link1:!!} For $q = 1$ and $\sigma \in \Sigma$, a graph $G$ is \typesymbol{!}$\sigma$-$q$-colourable if and only if all connected components are edges, which can be decided in polynomial time.
	
\end{description}

\subsection{Case \boldmath$q = 2$ for \typesymbol{=}$\sigma$-colouring}

\begin{description}
	
	\item[type =0]\label{link2:=0} For $q = 2$, a graph $G$ is \type=0-$q$-colourable if and only if the graph is bipartite and all connected components are even, which can be decided in polynomial time.
	
	\item[type =2]\label{link2:=2} For $q = 2$, a graph $G$ is \type=2-$q$-colourable if and only if the graph is bipartite, has no isolated vertices, and all connected components are even, which can be decided in polynomial time.
	
	\item[type =1]\label{link2:=1} For $q = 2$, a graph $G$ is \type=1-$q$-colourable if and only if the graph is bipartite and all connected components are odd, which can be decided in polynomial time.
	
	\item[type =v]\label{link2:=v} For $q = 2$, a graph $G$ is \type=v-$q$-colourable if and only if the graph is bipartite and all connected components are odd or isolated vertices, which can be decided in polynomial time.
	
	\item[type =$\star$]\label{link2:=*} For $q = 2$, a graph $G$ is \type=*-$q$-colourable if and only if the graph is bipartite, which can be decided in polynomial time.
	
	\item[type ==]\label{link2:==} For $q = 2$, a graph $G$ is \type==-$q$-colourable if and only if all connected components are isolated vertices, which can be decided in polynomial time.
	
	\item[type =+]\label{link2:=+} For $q = 2$, a graph $G$ is \type=+-$q$-colourable if and only if the graph is bipartite and has no isolated vertices, which can be decided in polynomial time.
	
	\item[type =?]\label{link2:=?} For $q = 2$, a graph $G$ is \type=?-$q$-colourable if and only if all connected components are edges or isolated vertices, which can be decided in polynomial time.
	
	\item[type =!]\label{link2:=!} For $q = 2$, a graph $G$ is \type=!-$q$-colourable if and only if all connected components are edges, which can be decided in polynomial time.
	
\end{description}

\subsection{Other easy cases for \boldmath$q=2$}

\begin{description}
	
	\item[type \type*+]\label{link2:*+} A graph containing an isolated vertex has no \type*+-colouring. All other graphs have such a colouring with two colours. Assume $G$ is a connected graph other than $K_1$. We choose a spanning tree $T$ of $G$. Since $T$ is bipartite, it has a proper $2$-colouring, which is a \type*+-colouring of $G$.
	
	\item[type ??]\label{link2:??} For $q = 2$, a graph $G$ is \type??-$q$-colourable if and only if all connected components are paths or cycles of length divisible by four, which can be decided in polynomial time.
	
	\item[type ?!]\label{link2:?!} For $q = 2$, a graph $G$ is \type?!-$q$-colourable if and only if all connected components are paths of odd length or cycles of length divisible by four, which can be decided in polynomial time.
	
	\item[type !?]\label{link2:!?} For $q = 2$, a graph $G$ is \type!?-$q$-colourable if and only if all connected components are paths of odd length or cycles of length divisible by four, which can be decided in polynomial time.
	
	\item[type !!]\label{link2:!!} For $q = 2$, a graph $G$ is \type!!-$q$-colourable if and only if all connected components are cycles of length divisible by four, which can be decided in polynomial time.
	
\end{description}

\subsection{Easy cases for \boldmath$q=3$}

\begin{description}
	
	\item[type =?]\label{link3:=?} For $q = 3$, a graph $G$ is \type=?-$q$-colourable if and only if all connected components are paths or cycles of length divisible by three, which can be decided in polynomial time.
	
	\item[type =!]\label{link3:=!} For $q = 3$, a graph $G$ is \type=!-$q$-colourable if and only if all connected components are cycles of length divisible by three, which can be decided in polynomial time.
	
\end{description}

\subsection{Polynomial-time solvable cases for every \boldmath$q$}\label{ss3easy}

\begin{description}
	
	\item[type 00]\label{link1234:00} We call colourings of type \type00 \emph{even-even colourings}. To allow for such a colouring, each connected component of the graph has to be even. All such graphs have an even-even colouring where all vertices receive the same colour. This colouring is called \emph{monochromatic}.
	
	\item[type \type*0 and \type*v and \type*? and $\star$=]\label{link1234:*0}\label{link1234:*v}\label{link1234:*?}\label{link1234:*=} For every graph the monochromatic colouring is of type \type*0, \type*v, \type*? and \type*=. 
	
	\item[type \boldmath$\star\star$]\label{link1234:**} For every graph, every colouring is of type \type**.
	
	\item[type 0=]\label{link1234:0=} The monochromatic colouring is of type \type0= if and only if all connected components are even. Otherwise, there is no such colouring. Moreover, such colourings are monochromatic in each connected component.
	
	\item[type 2=]\label{link1234:2=} All colourings of type \type2= are monochromatic in each connected component. They exist if and only if the graph has no isolated vertices and every connected component is even.
	
	\item[type 1=]\label{link1234:1=} The monochromatic colouring is of type \type1= if and only if all connected components are odd. Moreover, \type1=-colourings are monochromatic in each connected component.
	
	\item[type ==]\label{link1234:==} For $q>0$, a graph has a colouring of type \type== if and only if it is edgeless. All colourings work.
	
	\item[type +0 and +v and \type+* and +? and +=]\label{link1234:+0}\label{link1234:+v}\label{link1234:+*}\label{link1234:+?}\label{link1234:+=} In all cases monochromatic colouring works if and only if the graph has no isolated vertices. Moreover, such colourings are monochromatic in each connected component.
	
	\item[type 10]\label{link1234:10} A \type10-colouring only exists if all components are odd. In this case the monochromatic colouring is of type \type10. Moreover, such colourings are monochromatic in each connected component.
	
	\item[type 20]\label{link1234:20} Every component has to be even and the graph cannot have isolated vertices. In this case the monochromatic colouring is of type \type20. Moreover, such colourings are monochromatic in each connected component.
	
	\item[type v$=$]\label{link1234:v=} The monochromatic colouring is of type \type{v}=\ if and only if all connected components are odd or isolated vertices. Moreover, such colourings are monochromatic in each connected component.
	
	\item[type \typesymbol{?}\typesymbol{=}]\label{link1234:?=} The monochromatic colouring is of type \type?=\ if and only if all connected components are edges or isolated vertices. Moreover, such colourings are monochromatic in each connected component.
	
	\item[type \typesymbol{!}\typesymbol{=}]\label{link1234:!=} The monochromatic colouring is of type \type!= if and only if all connected components are edges. Moreover, such colourings are monochromatic in each connected component.
	
\end{description}

\end{document}